\documentclass[letterpaper,11pt]{article}

\usepackage{graphicx,xcolor} 
\usepackage[mathscr]{eucal}
\usepackage[fleqn]{amsmath}
\usepackage{amssymb,bbm,amsthm}
\usepackage{thm-restate,thmtools}
\usepackage{subcaption}
\usepackage[utf8]{inputenc}
\usepackage{graphicx}
\usepackage{hyperref}
\usepackage[capitalize]{cleveref}

\newcommand{\bigO}{\mathcal{O}}

\newcommand{\G}{\mathcal{G}}
\usepackage{comment,color}
\usepackage{mathtools}
\DeclarePairedDelimiter\ceil{\lceil}{\rceil}

\usepackage{amsthm}
\usepackage{float}
\usepackage{cancel}
\usepackage{algorithm}
\usepackage[noend]{algpseudocode}
\usepackage[margin=1in, papersize={8.5in,11in}]{geometry}
\usepackage{enumitem}
\usepackage{todonotes}

\declaretheorem[numberwithin=section]{theorem}
\declaretheorem[unnumbered, name=Theorem]{theorem*}
\declaretheorem[numberlike=theorem]{lemma}

\declaretheorem[numberlike=theorem]{proposition}

\declaretheorem[numberlike=theorem]{corollary}

\declaretheorem[unnumbered]{claim}

\declaretheorem[numberlike=theorem]{observation}

\declaretheorem[numberlike=theorem]{remark}
\declaretheorem[numberlike=theorem, name=Definition]{definition}
\declaretheorem[unnumbered, name=Definition]{definition*}
\declaretheorem[
  numberlike=theorem,
  shaded={rulecolor=black, rulewidth=1pt},
  name=Definition,
]{boxed-definition}

\declaretheorem[unnumbered, name=Construction]{construction*}

\declaretheorem[unnumbered, name=Conjecture]{conjecture*}
\declaretheorem[numberlike=theorem, name=Hypothesis]{hypothesis}
\declaretheorem[unnumbered, name=Hypothesis]{hypothesis*}

\declaretheorem[numberlike=theorem,name=Problem]{problem}

\newcommand{\FOP}[1]{$\mathrm{FOP}_k$}

\makeatletter
\renewcommand\paragraph{%
  \@startsection{paragraph}
    {4}
    {\z@}
    {3.25ex \@plus1ex \@minus.2ex}
    {-1em}
    {\normalfont\normalsize\bfseries\addperiod}}
\newcommand{\addperiod}[1]{#1\@addpunct{.}}

\newcommand{\var}[1]{{\rm var}(#1)}

\usepackage{tikz}
\usetikzlibrary{calc,fit,matrix,positioning,quotes,shapes,shapes.symbols}

\pgfdeclarelayer{bg}
\pgfsetlayers{bg,main}

\DeclarePairedDelimiter\set{\lbrace}{\rbrace}
\DeclareMathOperator\im{im}

\newcommand\Rat{\mathbb Q}

\newcommand{\conference}{0} 

\newcommand{\conferenceversion}[2]{
  \ifnum\conference=1
    #1%
  \else
    #2%
  \fi
}
\title{Classifying Identities: Subcubic Distributivity Checking and Hardness from Arithmetic Progression Detection}
\author{Bartłomiej Dudek\thanks{University of Wrocław, Wrocław, Poland} \and Nick Fischer\thanks{Max Planck Insitute for Informatics, Saarbrücken, Germany} \and Geri Gokaj\thanks{Karlsruhe Institute of Technology, Karlsruhe, Germany} \and Ce Jin\thanks{University of California at Berkeley, Berkeley, USA} \and Marvin Künnemann\footnotemark[3] \and Xiao Mao\thanks{Stanford University, Stanford, USA} \and Mirza Redzic\footnotemark[3]}
\date{}

\newenvironment{subproof}[1][\proofname]{%
  \begin{proof}[#1]%
}{%
  \end{proof}%
}

\begin{document}

\maketitle

\begin{abstract}
\noindent
We revisit the complexity of verifying basic identities, such as  associativity and distributivity, on a given finite algebraic structure. In particular, while Rajagopalan and Schulman (FOCS'96, SICOMP'00) gave a surprising randomized algorithm to verify associativity of an operation $\odot: S\times S\to S$ in optimal time $O(|S|^2)$, they left the open problem of finding any subcubic algorithm for verifying distributivity of given operations $\odot,\oplus: S\times S\to S$.

Our results are as follows:
\begin{enumerate}
    \item We resolve the open problem by Rajagopalan and Schulman by devising an algorithm verifying distributivity in strongly subcubic time $O(|S|^\omega)$, together with a matching conditional lower bound based on the Triangle Detection Hypothesis.
    \item We propose arithmetic progression detection in small universes as a consequential algorithmic challenge: We show that unless we can detect $4$-term arithmetic progressions in a set $X\subseteq\{1,\dots, N\}$ in time $O(N^{2-\epsilon})$, then (a) the 3-uniform 4-hyperclique hypothesis is true, and (b) verifying certain identities requires running time~$|S|^{3-o(1)}$.
    \item A careful combination of our algorithmic and hardness ideas allows us to \emph{fully classify} a natural subclass of identities: Specifically, any 3-variable identity over binary operations in which no side is a subexpression of the other is either: (1) verifiable in randomized time $O(|S|^2)$, (2) verifiable in randomized time $O(|S|^\omega)$ with a matching lower bound from triangle detection, or (3) trivially verifiable in time $O(|S|^3)$ with a matching lower bound from hardness of 4-term arithmetic progression detection.
    \item We obtain near-optimal algorithms for verifying whether a given algebraic structure forms a field or ring, and show that \emph{counting} the number of distributive triples is conditionally harder than verifying distributivity.
\end{enumerate}

\end{abstract}

\thispagestyle{empty}
\newpage

\tableofcontents
\thispagestyle{empty}
\newpage

\setcounter{page}{1}

\section{Introduction}

We consider fundamental questions about a given algebraic structure -- which consists of a ground set $S$ together with one or more operations $\oplus, \odot, \otimes$, etc.\  over $S$ -- such as: Does $(S,\oplus)$ form a group? Does $(S,\oplus, \odot)$ form a field? Does $(S,\oplus,\odot)$ form a ring? Here, we consider the case that $S$ is a finite set and each operation is binary, represented by its Cayley table.

Answering such questions essentially boils down to deciding whether (a collection of) certain identities hold for given Cayley tables, such as the associative law $(ab) c = a(bc)$ or the distributive law $a(b+c)= ab+ac$. Interestingly, already verifying associativity of a given operation turns out to be non-trivial: It was not until 1996 that Rajagopalan and Schulman \cite{rajagopalan00} showed how to beat the trivial $O(n^3)$-time baseline by a surprising randomized algorithm that checks associativity with constant probability in time $O(n^2)$, where $n=|S|$. (Note that since the Cayley table of a binary operation on an $n$-element set $S$ has size $\Theta(n^2)$, this algorithm is optimal, albeit randomized.) In particular, this leads to a randomized $O(n^2)$-time algorithm to check whether $(S,\oplus)$ forms a group\footnote{They also obtain a deterministic $\bigO(n^2\log n)$ algorithm, by making use of F.W. Light's associativity checking observation (see \cite{Clifford1961}) also discussed in Section \ref{sec:further-results}.}. Using a different approach, a recent breakthrough result due to Evra, Gadot, Klein, and Komargodski~\cite{groupverifictaion2023} obtains a \emph{deterministic} $O(n^2)$-time algorithm for this task (see Section~\ref{sec:further-results} below for further details). 

Rajagopalan and Schulman's algorithm for verifying associativity generalizes to verifying any \emph{read-once} identity in optimal time, where a read-once identity is any identity with an arbitrary number of operations and variables in which on each side of the equation, each variable appears exactly once. However, the complexity of verifying non-read-once identities (most importantly, distributivity) remains wide open.\footnote{Rajagopalan and Schulman~\cite{rajagopalan00} write: ``It remains to make any progress on verification of identities which are not read-once. A key example is the 'distributive' identity $a \cap (b \cup c) = (a \cap b) \cup (a \cap c)$; it is not known whether this can be verified in less than cubic time.''} Note that verifying the distributive law is a key task, e.g., for verifying whether a given structure $(S, \oplus, \otimes)$ forms a field or a ring. We thus ask the following driving question:

\begin{center}
\emph{Is there a complexity dichotomy of ``easy'' identities, checkable in time $O(n^2)$, \\\and ``hard'' identities, for which naive running time is (conditionally) best possible?}
\end{center}
In light of Rajagopalan and Schulman's work, the set of ``easy'' identities must contain all read-once formulas. However, we are not aware of published superquadratic conditional lower bounds for \emph{any} identity.

\subsection{Our Results}

We focus on identities with three variables, but remark that some of our results can be generalized to identities with four or more variables.

\subsubsection{Conditionally Optimal Verification of Distributivity}

As our first main result, we obtain a conditionally optimal $O(n^\omega)$-time randomized algorithm for verifying distributivity. Here, $\omega<2.37134$ denotes the exponent of matrix multiplication~\cite{AlmanDWXXZ25}.

\begin{theorem}[Distributivity Verification]\label{thm:distributivity}
Given a set $S$ and two binary operations $\oplus, \odot: S\times S\to S$, there is a randomized $\bigO(n^{\omega})$-time algorithm that decides with high probability whether for all $a,b,c\in S$, we have $a\odot (b\oplus c) = (a \odot b) \oplus (a\odot c)$. Any $O(n^{\omega-\epsilon})$-time algorithm with $\epsilon > 0$ for this task would refute the Triangle Detection Hypothesis.
\end{theorem}

The Triangle Detection Hypothesis (see, e.g.,~\cite{AbboudW14} and the discussion in \conferenceversion{the full version of the paper}{Appendix~\ref{sec:assumptions})} postulates that there is no $O(n^{\omega-\epsilon})$-time algorithm for detecting a triangle in an $n$-node graph.
As immediate consequences of our algorithm, we obtain a randomized $\bigO(n^{\omega})$-time algorithm for verifying whether $(S,\oplus, \odot)$ forms a field or a ring. (In Section~\ref{sec:further-results}, we describe improved algorithms for these tasks, exploiting that field/ring verification involves additional side constraints as well as a result in~\cite{groupverifictaion2023}.) Notably, Theorem~\ref{thm:distributivity} reveals that we cannot expect a \emph{dichotomy} into ``easy'' and ``brute-force-hard'' identities.

\subsubsection{A Classification Theorem via 4-Term Arithmetic Progressions}

The conditionally optimal verification algorithm for distributivity raises the question whether perhaps all 3-variable identities can be verified in strongly subcubic time. However, we first observe that such a powerful algorithmic result appears very challenging: We show that there are identities for which a strongly subcubic verification algorithm would refute the Strong Zero Triangle Hypothesis~\cite[Conjecture 7.5]{AbboudBKZ22} from fine-grained complexity (for further details and the proof, see \conferenceversion{the full version of the paper.}{Sections~\ref{sec:overview} and~\ref{sec:strongzerotriangle}; Section ~\ref{sec:assumptions} in the Appendix gives a definition of the Strong Zero Triangle Hypothesis}).

\begin{restatable}{theorem}{zerotrianglehard}\label{thm:hardest}
Given an algebraic structure $(S, \oplus, \odot)$ with $0\in S$, if we can decide whether it satisfies the identity \begin{equation}\label{eq:hardest}
\left((a\odot b)\oplus(a\odot c)\right)\oplus(b\odot c)=0
\end{equation}
for all $a,b,c\in S$ in time $O(n^{3-\epsilon})$ for some $\epsilon > 0$, then the Strong Zero Triangle Hypothesis fails.
\end{restatable}

In contrast, on the positive side, we are able to extend the techniques of the subcubic distributivity algorithm to solve a variety of additional identities in time $\bigO(n^{\omega})$ as well. Curiously, however, these generalizations reveal some identities of particular interest, which appear not to be solvable in strongly subcubic time using this approach, but also do not appear to admit Strong-Zero-Triangle-hardness. As an example, consider the identity
\begin{equation}
    \left((a\odot b) \oplus (a \odot c)\right) \oplus b = 0. \label{eq:intermediate-identity}
\end{equation}
Note that its left-hand side takes an ``intermediate'' position on the spectrum between distributivity, $(a\odot b) \oplus (a \odot c)$, and the hard identity \eqref{eq:hardest} of Theorem~\ref{thm:hardest}, $\left((a\odot b)\oplus(a\odot c)\right)\oplus(b\odot c)$.

As one of our most interesting contributions, both conceptually and technically, we justify the apparent hardness of this and related identities via a reduction from an algorithmic problem in additive combinatorics: Specifically, let 4-AP Detection denote the problem of deciding whether a set $A\subseteq \{1,\dots, N\}$ contains a non-trivial 4-term arithmetic progression $a, a+b, a+2b, a+3b$ (where $b\ne 0$). We show that an algorithm verifying \eqref{eq:intermediate-identity} in time $O(n^{3-\epsilon})$ would give an $O(N^{2-\epsilon'})$-time algorithm for 4-AP Detection. Let us discuss the relevance of this result.

\paragraph{Arithmetic Progression Detection as a Fine-Grained Hypothesis}
$k$-term arithmetic progressions, or $k$-APs for short, are central objects in additive combinatorics. Many fascinating results revolve around the study of $k$-AP-free sets, such as Roth's and Szemerédi's seminal theorems~\cite{TaoV06}. We consider the same question from a computational perspective: How efficiently can we decide if a given set $A \subseteq \{1,\dots, N\}$ is $k$-AP-free?

The naive algorithm runs in time $O(N^2)$ for any fixed $k$: Each combination of the first two terms $a_1, a_2 \in A$ uniquely determines the remaining terms $a_3, \dots, a_k$ (via $a_i = a_1 + (i - 1) (a_2 - a_1)$) and can thus be tested in time $O(1)$. For $k = 3$ there is a faster $\tilde O(N)$-time algebraic algorithm based on the Fast Fourier Transform. However, this algebraic approach does \emph{not} generalize to $k \geq 4$, and we are not aware of any polynomial improvement over the naive $O(N^2)$-time algorithm. This jump in complexity from $k = 3$ to $k = 4$ is not a coincidence, but is actually mirrored in additive combinatorics: $3$-AP-free sets are extremely well-understood, with almost matching lower and upper density bounds~\cite{Behrend1946,KelleyM23} based on an extensive toolkit such as Fourier analysis. To date the study of $4$-AP-free sets either suffers from astronomical quantitative losses, or relies on the intricate ``higher-order Fourier analysis''~\cite{Gowers01}. We view this as one of many reasons (to be described in \Cref{sec:ap}) to put forth the following hypothesis for all constants $k \geq 4$:
\gdef\N{N}%
\begin{restatable}[$k$-AP]{hypothesis}{hypkap} \label{hyp:k-ap}
There is no algorithm to detect a $k$-AP in a given set $A \subseteq \{0,\ldots,\N\}$ in time $O(\N^{2-\epsilon})$ (for any constant $\epsilon > 0$).\gdef\N{n}%
\end{restatable}

In Section~\ref{sec:ap} we further justify the relevance of the $k$-AP Hypothesis for understanding hardness in P, and argue that it deserves to be viewed as a new fine-grained hypothesis. In particular, we establish $k$-AP Detection as an ``arithmetic analogue'' of the well-established Hyperclique Hypothesis, review why current algorithmic approaches fail to break the $k$-AP Hypothesis, and point out further connections to central problems like $4$-SUM. We also develop a web 
of fine-grained reductions from 4-AP Detection (detailed in \conferenceversion{the full version of the paper}{Sections~\ref{sec:overview} and~\ref{sec:4APhard-identities}}) to interesting intermediate problems that connect to verifying identities, our core interest of this paper.

\paragraph{Main Classification Theorem}
Equipped with the techniques of the results discussed above, we are able to prove a classification theorem for a natural set of identities, specifically, all identities in which neither side is a subexpression of the other.

\begin{theorem}[Classification Theorem, informal version]
    Consider an identity $f(a,b,c) = g(a,b,c)$ where $f$ is not a subexpression of $g$ and vice versa. Each such identity falls into one of three regimes: \begin{enumerate}
    \item it can be checked in randomized time $O(n^2)$, or
    \item it can be checked in randomized time $O(n^{\omega})$; any $O(n^{\omega-\epsilon})$-time algorithm would violate the Triangle Detection Hypothesis;
    \item it can (trivially) be checked in time $O(n^3)$; any $O(n^{3-\epsilon})$-time algorithm would refute the 4-AP Detection Hypothesis.
    \end{enumerate}
\end{theorem}

This theorem successfully answers our driving question for a large and interesting class of identities. Curiously, what we initially expected to become a dichotomy turned out to be a natural \emph{trichotomy}.

In fact, we obtain more detailed results: Specifically, we can describe the set of identities in the first regime as all formulas where either side can be written as $H(G(a,b),c)$ for appropriate choices of $H,G$ and permutation of the variables. This generalizes the result by Rajagopalan and Schulman~\cite{rajagopalan00} which remained restricted to read-once identities. E.g.,  $(c((ab)(ba)))(c+c)=(ba)+(ab)$ lies in Regime 1, but is far from read-once.

Our results show that verifying any identity in Regime 3 in strongly subcubic time necessitates a strongly subquadratic algorithm for 4-AP Detection. This means that either: (1) 4-AP Detection is indeed quadratic-time hard, and we have fully settled the complexity for the most natural identities, or (2) 4-AP Detection can be solved in strongly subquadratic time -- in this case, the techniques behind such an algorithm might lead the way to obtaining strongly subcubic algorithm for the 4-AP-hard identities. Ultimately, the authors remain agnostic about which case is more likely.

\subsubsection{Further Results on Ring and Field Verification and Counting}
\label{sec:further-results}

A corollary of our $\bigO(n^{\omega})$-time randomized algorithm for distributivity verification together with the quadratic-time algorithms for group verification~\cite{rajagopalan00,groupverifictaion2023} immediately yield $\bigO(n^{\omega})$ time algorithms for verifying whether $(S,\oplus, \odot)$ forms a field (or ring, respectively). We improve over this baseline by obtaining (almost-)optimal algorithms for these tasks:

\begin{restatable}{theorem}{FieldVerificationTheorem}\label{th:field-verification}
    Given a set $F$ consisting of $n$ elements, together with binary operations $\oplus, \odot$, there exists a deterministic algorithm deciding if $(F,\oplus,\odot)$ forms a field in time $\bigO(n^2)$.
\end{restatable}
\begin{restatable}{theorem}{RingVerificationTheorem}\label{th:ring-verification}
    Given a set $R$ consisting of $n$ elements, together with binary operations $\oplus, \odot$, there exists a randomized algorithm deciding if $(R,\oplus,\odot)$ forms a ring in time $\bigO(n^2)$.
\end{restatable}
We make use of a breakthrough result of Evra, Gadot, Klein, and Komargodski~\cite{groupverifictaion2023} which constructs a basis for $S$ of size $O(\sqrt{|S|})$ in time $O(n^2)$, where a basis of $S$ is a set $B$ such that $B^2=S$ (i.e., each element in $S$ can be written as the product of two elements in $B$). In order to find this basis, Evra et al. make use of the classification of finite simple groups.

Prior to this result, F.W. Light  made a simple observation for the associativity checking problem (see \cite{Clifford1961}). If we are given a smaller set of generators $R \subseteq S$, it suffices to check associativity for all $a, c \in S$ but only for $b \in R$. We use a similar approach, leveraging the basis construction of Evra et al., to show our results for field and ring verification.


\paragraph{Counting Variant}
We remark that our algorithms for verifying distributivity and related identities avoid (explicitly or implicitly) representing the full search space, fundamentally exploiting that the identity needs to hold for \emph{all} choices of $a,b,c\in S$. We give evidence that this is a crucial aspect, by showing that \emph{counting} all distributive triples requires cubic running time assuming the Strong Zero Triangle Hypothesis. This adds distributivity verification to the currently rather small list of decision problems in P for which a fine-grained separation from its natural counting version is known.

\begin{restatable}{theorem}{DistCountingTheorem}\label{th:dist-counting}
     Given a set $S$ of $n$ elements, together with binary operations $\oplus, \odot$, if there exists an algorithm that counts the number of distributive triples, i.e., the triples $x,y,z\in S$ such that $x\odot(y\oplus z) = (x\odot y) \oplus (x\odot z)$
     in time $\bigO(n^{3-\varepsilon})$ for some $\varepsilon>0$, then
     the Strong Zero Triangle Hypothesis is false.
\end{restatable}

\subsection{Conclusion and Open Problems}
Our main theorem establishes a trichotomy for a large natural subclass of identities. In particular, the dividing line between identities verifiable in strongly subcubic time and identities requiring cubic time is fully explained by the 4-AP Hypothesis. As such, the perhaps most pressing open problem is to find further support or evidence against quadratic-time hardness of 4-AP Detection, and we encourage readers to try to refute it.

Orthogonal to this, a natural question is to extend our classification to also settle the (possibly slightly obscure) identities in which one side is a proper subexpression of the other. Furthermore, it would be interesting to classify (subclasses of) identities involving four or more variables. Some of our algorithmic and/or hardness ideas extend to such more complex formulas, however, a full classification will likely require a classification into more than 3 regimes.

Finally, our hardness results on counting distributive triples suggests that a classification of the natural counting problem over different identities will be substantially different from analyzing the verification task, which has been the main focus of this work.

\section{Preliminaries}\label{sec:preliminaries}
For positive integers $n,m$ with $m \leq n$, let $[m,n]$ denote the set $\{m,\dots, n\}$. With $[n]$ we denote the set $[1,n]$.
Let $\omega<2.37134$ \cite{AlmanDWXXZ25} denote the optimal exponent of multiplying two $n\times n$ matrices.
We use $\Tilde{\bigO}$ notation, which hides the poly-logarithmic factors.
In other words, $f(n)\in \Tilde{\bigO}(g(n))$ if and only if there exists a $k\in \bigO(1)$ such that $f(n)\in \bigO(g(n)\cdot \log^k(n))$. All the hardness hypotheses used in this work are described in \conferenceversion{the full version of the paper}{Appendix~\ref{sec:assumptions}}. Throughout, we simply write \emph{randomized} algorithm to refer to Monte Carlo randomized algorithms that succeed with high probability $1 - n^{-c}$ (for an arbitrarily large constant $c$).

\subsection{Algebraic Expressions}
We consider algebraic expressions built up from variables and binary operations $\odot_1,\ldots,\odot_p$.
While analyzing a concrete instance of such problem, we operate on a set $S$ of $n$ elements and operations $\odot_1,\ldots,\odot_p$ that operate on $S$ ($\odot_i:S\times S\rightarrow S$) and are provided as their Cayley tables (for each of them we have table $C_i\in S^{S\times S}$ where $C_i[a,b] = a \odot_i b$ for all $a,b\in S$). All variables in the expression can take any value from $S$.
The algebraic expression itself is part of the description of the problem and we assume it has constant size.
Then, our algorithmic task is to take set $S$ and all the Cayley tables and return True or False, depending on whether the given property is satisfied or not.

\begin{figure}[t] 
\centering
\begin{tikzpicture}[
  level distance=12mm,
  every node/.style={font=\normalsize},
  op/.style={inner sep=1.5pt},
  var/.style={inner sep=3pt},
  level 1/.style={sibling distance=32mm},
  level 2/.style={sibling distance=20mm}
]
\node[op] (root) {$\oplus$}
  child { node[op] (L) {$\odot$}
    child { node[var] (x1) {$a$} }
    child { node[var] (x2) {$b$} }
  }
  child { node[op] (R) {$\odot$}
    child { node[var] (x3) {$c$} }
    child { node[var] (x4) {$a$} }
  };

\end{tikzpicture}
\caption{Tree representation of the expression
$f(a,b,c) \equiv (a \odot b) \oplus (c \odot a)$, with ${\rm depth}(f) = 2$ and
${\rm var}(f) = \{a,b,c\}$.}
\label{fig:parsetree}
\end{figure}
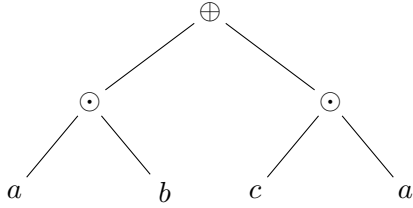

Every algebraic expression can be represented with a rooted, ordered binary tree with internal nodes labeled with operations and leaves labeled with variables that can take values from $S$ (see Figure~\ref{fig:parsetree}).
We define ${\rm depth}(f)$ as the depth of the tree representing $f$ (where a single-node tree, e.g., a variable, has depth 0) and we say that expression $f'(\bar v)$ is a subexpression of $f(\bar v)$ iff the tree representing $f'$ is a subtree of the tree representing $f$.
For an expression $f$, we define ${\rm var}(f)\ $
as the set of variables that $f$ depends on (i.e., that appear as leaves in the tree representing $f$).
As there is a bijection between such trees and considered algebraic expressions, later we identify a tree with its expression and vice versa.
For two expressions $f,g$, we write $f\equiv g$ to denote that $f$ and $g$ are \emph{symbolically equivalent expressions} (in other words, $f$ and $g$ represent the same algebraic expression).
On the other hand, given a set $S$ and the operations $\odot_1,\dots, \odot_ p$ as the Cayley tables, we write $f(x,y,z) = g(x,y,z)$ for some $x,y,z\in S$, to denote that the \emph{evaluations} of $f$ and $g$ on the elements $x,y,z$ yield the same element in~$S$.

\section{Subcubic Complexity of Distributivity Checking}
\label{sec:technical-overview-distributivity}

In order to provide an overview of our results, we start with establishing the complexity of Distributivity Verification.
Its proof is self-contained and gives an intuition of the nature of our results.

We provide a randomized algorithm running in $\bigO(n^\omega)$-time and a matching lower bound that there is no $\bigO(n^{\omega-\epsilon})$-time algorithm for Distributivity
Checking unless we can detect triangles in $\bigO(n^{\omega-\epsilon})$-time and the Triangle Detection Hypothesis \conferenceversion{\footnote{Triangle Detection cannot be solved in $\bigO(n^{\omega-\varepsilon})$ for any $\epsilon>0$.}}{(Hypothesis ~\ref{hypot:triangle})} fails:

\begin{theorem}\label{th:distributivity}
    Given a set $S$ of size $n$ together with two binary operations $\oplus, \odot: S\times S \to S$, there exists a randomized algorithm checking in time $\bigO(n^{\omega})$ if the algebraic structure $(S,\oplus, \odot)$ is distributive, that is if for all $a,b,c\in S$ we have $a\odot(b\oplus c) = (a\odot b)\oplus(a\odot c)$. Moreover, there is no algorithm solving this problem in time $\bigO(n^{\omega-\epsilon})$ for any $\epsilon > 0$ unless the Triangle Detection Hypothesis is false.
\end{theorem}

\subsection{Distributivity Checking is Triangle-Hard}

\begin{lemma}
Given a set $S$ of size $n$ with binary operations $\oplus$ and $\odot$, if there exists an algorithm deciding whether $\odot$ is distributive over $\oplus$ in time $T(n)$, then there exists an algorithm detecting a triangle in a graph on $n$ vertices in time $\bigO(T(n))$.
\end{lemma}

\begin{proof}
We reduce triangle detection to checking distributivity, i.e., whether for all $a,b,c\in S$ we have $a\odot(b\oplus c) = (a\odot b)\oplus(a\odot c)$.
Given a graph $G=(V,E)$ on $n$ vertices, define $S = V \cup \{\Delta, \infty\}$ and the operations:
\[
a\oplus b =
\begin{cases}
\Delta & \text{if } \{a,b\}\in E,\\
\infty & \text{otherwise,}
\end{cases}
\qquad
a\odot b =
\begin{cases}
b & \text{if } \{a,b\}\in E,\\
\infty & \text{otherwise.}
\end{cases}
\]
For all $x\in S$, let
\begin{align*}
    &\infty \oplus x = x \oplus \infty = \infty \odot x = x \odot \infty = \infty \\
&\Delta \oplus x = x \oplus \Delta = \Delta \odot x = x \odot \Delta = \infty.
\end{align*}
Intuitively, $\infty$ acts as an absorbing element, while $\Delta$ flags a triangle.

If any of $x,y,z$ equals $\infty$ or $\Delta$, both $x\odot(y\oplus z)$ and $(x\odot y)\oplus(x\odot z)$ equal $\infty$, so distributivity trivially holds. It remains to consider $x,y,z\in V$.

\begin{claim}
For all $a,b,c\in V$:
\begin{enumerate}
    \item $a\odot(b\oplus c)=\infty$.
    \item $(a\odot b)\oplus(a\odot c)=\infty$ iff $a,b,c$ do not form a triangle in $G$.
\end{enumerate}
\end{claim}

\begin{proof}
Since $b\oplus c\in\{\infty,\Delta\}$, and $a\odot\infty=a\odot\Delta=\infty$, (i) follows.

For (ii), if $a,b,c$ form a triangle, then $(a\odot b)\oplus(a\odot c)=b\oplus c=\Delta\neq\infty$.
Otherwise, if any of $\{a,b\},\{a,c\}$ is missing, one of $a\odot b,a\odot c$ is $\infty$, and hence $(a\odot b)\oplus(a\odot c)=\infty$.
If only $\{b,c\}$ is missing, then $b\oplus c=\infty$, yielding the same result.
\end{proof}

Hence, $(S,\oplus,\odot)$ fails to be distributive exactly when $G$ contains a triangle.
Constructing $S$ and its Cayley tables takes linear time, so running the distributivity checker on $(S,\oplus,\odot)$ detects triangles in $\bigO(T(n))$ time.
\end{proof}

\begin{corollary}\label{col:distributivity_lb}
Unless the Triangle Detection Hypothesis fails, there is no algorithm deciding whether $(A,\oplus,\odot)$ is distributive in time $\bigO(n^{\omega-\varepsilon})$ for any $\varepsilon>0$.
\end{corollary}

\subsection{Distributivity Checking via Triangle Counting}\label{sec:distributivity_alg}
In this section we provide the promised $\bigO(n^{\omega})$ algorithm for distributivity checking. \conferenceversion{}{In Section~\ref{sec:freivalds_distributivity} in the Appendix, we provide an alternative algorithm based on matrix multiplication and Freivalds' algorithm.}

Recall that $S$ is a set of $n$ elements, and we are given the binary operations $\oplus,\odot: S\times S \to S$ as input tables. The goal is to check whether the identity $a\odot(b\oplus c) = (a\odot b) \oplus(a\odot c)$ holds for all $a,b,c\in S$.

Consider the following two polynomials $f,g$ on $4n$ variables, $\{x_s\}_{s\in S},\{y_s\}_{s\in S},\{z_s\}_{s\in S},\{w_s\}_{s\in S}$:
\begin{align*}
 f &= \sum_{a,b,c\in S} x_a \, y_b \,  z_c \,  w_{a\odot(b\oplus c)},  \\
 g &= \sum_{a,b,c\in S} x_a \,  y_b \, z_c  \, w_{(a\odot b) \oplus(a\odot c)}.
\end{align*}

Observe that, $f-g =\sum_{a,b,c\in S} x_a \,  y_b \, z_c  \, (w_{a\odot(b\oplus c)}- w_{(a\odot b) \oplus(a\odot c)}) $ is the zero polynomial if and only if $a\odot(b\oplus c) = (a\odot b) \oplus(a\odot c)$ holds for all $a,b,c\in S$.
Therefore, our goal reduces to polynomial identity testing on $f-g$. We pick a prime field $\mathbb{F}_p$ (for sufficiently large $p$) and assign independently uniformly random values from $\mathbb{F}_p$  to the variables $\{x_s\}_{s\in S},\{y_s\}_{s\in S},\{z_s\}_{s\in S},\{w_s\}_{s\in S}$, and then evaluate $f-g$ on these values. If the result is non-zero, we return NO; otherwise we return YES. Since $f-g$ has degree four, by Schwartz--Zippel lemma, the probability of incorrectly reporting YES is at most $4/p$.
It remains to describe how to evaluate $f-g$. We will evaluate $f$ and $g$ separately and subtract the results.

We first describe how to evaluate $g$. Define an edge-weighted tripartite graph $(P,Q,R)$, where $P,Q,R$ are vertex sets each of size $n$ and labeled by elements of $S$. Let $p_s\in P$ denote the vertex in $P$ labeled by $s\in S$, and similarly define vertices $q_s\in Q,r_s\in R$. Now we define the edges in the graph (note that we allow parallel edges):
\begin{itemize}
    \item For every pair of $a\in S,b\in S$, add an edge between $p_a\in P$ and $q_{a\odot b}\in Q$ with weight $x_ay_b$.
    \item For every pair of $a'\in S,c\in S$, add an edge between $p_{a'}\in P$ and $r_{a'\odot c}\in R$ with weight $z_c$.
    \item For every pair of $e\in S, f\in S$, add an edge between $q_{e}\in Q$ and $r_f\in R$ with weight $w_{e\oplus f}$.
\end{itemize}

Three edges with one of each type described above form a triangle if and only if the vertices coincide, $p_a=p_{a'},  q_e  = q_{a\odot b} , r_f = r_{a'\odot c}$, i.e., $a=a', e = a\odot b, f = a'\odot c$.   Then, the weight of this triangle (namely, the product of the weights of its three edges) equals
\[(x_ay_b)\cdot z_c \cdot w_{e\oplus f} = x_ay_b z_c w_{(a\odot b)\oplus (a\odot c)}.\]
The triangles are in bijection with the triples $(a,b,c)\in S\times S\times S$. Therefore,
the polynomial $g$ equals the sum of the weights of all triangles in the graph.

It remains to compute the total weight of all triangles in the graph $(P,Q,R)$.  For a pair of vertices $p_i\in P,q_j\in Q$, let $W^{PQ}_{i,j}$ denote the total weight of edges connecting $p_i,q_j$ (there could be zero, one, or multiple such edges). Similarly define $W^{QR}_{j,k}$ and $W^{RP}_{k,i}$. Then, clearly, the total weight of all triangles equals the trace of the matrix product $W^{PQ}W^{QR}W^{RP}$, which can be computed in $O(n^{\omega})$ time.

Evaluating polynomial $f$ is even simpler and we proceed in a very similar way. Again we create a tripartite graph $(P,Q,R)$ with parallel edges allowed. Now for every pair $b\in S, c\in S$ we add an edge between $p_b\in P$ and $q_{b\oplus c}\in Q$ with weight $y_bz_c$; and for every pair $e\in S, a\in S$ we add an edge between $q_e\in Q$ and $r_{a\odot e}\in R$ with weight $x_aw_{a\odot e}$. Then the polynomial $f$ equals the sum of all length-2 paths in $P\times Q\times R$, which can be calculated in $O(n^2)$ time.

Finally, we choose $p=\Omega(n)$. As a single assignment of variables $\{x_s\}_{s\in S},\{y_s\}_{s\in S},\{z_s\}_{s\in S},\{w_s\}_{s\in S}$ from $\mathbb{F}_p$ might give us false equality between non-equal polynomials, we need to repeat the above procedure $O(1)$ times in order to boost success probability. This yields the following claim:

\begin{lemma}
  Given a set $S$ of size $n$ together with two binary operations $\oplus, \odot: S\times S \to S$, there exists a randomized algorithm checking if the algebraic structure $(S,\oplus, \odot)$ is distributive in $O(n^{\omega})$ time that succeeds with high probability.
\end{lemma}
\noindent
The above lemma, together with Corollary~\ref{col:distributivity_lb} concludes the proof of Theorem~\ref{th:distributivity}.

\section{Technical Overview: Beyond Distributivity Checking} \label{sec:overview}
Beyond distributivity checking, our second main contribution is the following theorem, providing the complexity classification of many identities into three regimes: (1) quadratic time, (2) Triangle Detection time with matching upper and conditional lower bounds, and (3) cubic time with matching lower bounds based on the 4-AP Hypothesis.
\begin{theorem}[Trichotomy of Natural Identities]\label{thm:main-theorem}
    For any two expressions $f(a,b,c)$, $g(a,b,c)$, such that $f$ is not a subexpression of $g$ and vice versa, given an algebraic structure $(S,\oplus_1,\dots, \oplus_p)$, checking if $f(a,b,c) = g(a,b,c)$ for each $a,b,c\in S$ falls precisely in one of three regimes.
    \begin{enumerate}
        \item It can be checked in randomized $\bigO(n^2)$ time if both $f,g$ can be written as $I(H(a,b),c)$ (or any symmetric form), for some expressions $H,I$.
        \item  It can be checked in randomized $\bigO(n^\omega)$ time if both $f,g$ can be written as $K(J(I(a,b),c), H(a,b))$ (or any symmetric form), for some expressions $H,I,J,K$. Any $\bigO(n^{\omega-\varepsilon})$ algorithm for identities in this regime would refute Triangle Detection Hypothesis.
        \item It can be checked (trivially) in $\bigO(n^3)$ time. Any $\bigO(n^{3-\varepsilon})$ algorithm would refute the $4$-AP Detection Hypothesis.
    \end{enumerate}
\end{theorem}
In this section we give an overview of the main ideas behind this trichotomy: In \Cref{sec:overview:sec:algos} we quickly comment how our approach for distributivity checking extends to other identities in regimes~(1) and~(2). Then, in \Cref{sec:overview:sec:ap}, we explain how some selected identities from regime (3) can be shown to be 4AP-hard. Finally, in \Cref{sec:overview:sec:classes} we show how to obtain the full classification into the three regimes.

\subsection{Subcubic Algorithms} \label{sec:overview:sec:algos}
As all identities on three variables can be verified with a brute-force approach in $\bigO(n^3)$ time, we need to focus only on the quadratic and triangle-time regimes.

It turns out that our approach for verifying Distributivity almost immediately generalizes to verifying all identities of the form $f(a,b,c)=g(a,b,c)$ where $f$ and $g$ are from the two sub-cubic regimes. Again we reduce this question to verification of identity of polynomials defined separately for $f$ and $g$; and apply Schwartz--Zippel lemma.
In order to evaluate the  polynomials efficiently, it suffices to put appropriate weights on edges in a 3-partite graph and count triangles or 2-paths, which can be done in $\bigO(n^\omega)$ or $\bigO(n^2)$ time, respectively.
We provide the details in \conferenceversion{the full version of the paper}{Section~\ref{se:algorithms-main}}.

\subsection{Lower Bounds} \label{sec:overview:sec:ap}
From now on we focus on the conceptually very different challenge to establish fine-grained lower bounds for the identities in regime~(3) based on the 4-AP Hypothesis. We follow a clean approach: to identify a well-suited \emph{intermediate} problem, to show that this intermediate problem is 4-AP-hard, and then to reduce it to checking the hard identities.

\paragraph{Intermediate Problem: Square Detection}
The \emph{Square Detection} problem is to decide if a given $0$-$1$-matrix $M$ contains a \emph{square} \makebox{$M(i, j) = M(i, j + k) = M(i + k, j + k) = M(i + k, j) = 1$}. This problem was communicated to us by Adam Polak~\cite{Polak_personal_communication}. It turns out to be a convenient intermediate problem in our reduction chain, but is arguably mathematically natural and interesting in its own right. Square Detection can be solved in time $O(n^3)$ for $n \times n$ matrices (by enumerating all triples $i, j, k$), and no faster algorithm is known. This is to be expected, as in the next paragraph we show that it is 4-AP-hard.

\paragraph{Reduction from 4-AP Detection to Square Detection}
First observe that 4-AP and Square Detection have different complexities, hence we cannot expect a reduction that transforms one 4-AP instance into one equivalent Square Detection instance. Instead, it is reasonable to aim to reduce a given 4-AP instance $A \subseteq \set{0, \dots, n}$ to $N = \sqrt{n}$ Square Detection instances, each of size $N \times N$. If Square Detection is in subcubic time $O(N^{3-\epsilon})$, then such a reduction would imply that 4-AP Detection is in subquadratic time $O(N \cdot N^{3-\epsilon}) = O(n^{2-\epsilon/2})$, contradicting the 4-AP Hypothesis.

We now describe the idea to construct one such Square Detection instance. Suppose that there is a \emph{labeling} of the $N \times N$ grid, $\ell : N \times N \to \set{0, \dots, n}$, which satisfies that labels of the corners of every square form a 4-AP. Then consider the Square Detection instance $M$, where we place $1$-entries at exactly those grid points whose labels appear in $A$. This construction is sound: If $M$ is a yes instance, then~$A$ contains a $4$-AP. However, the construction is not complete-we only capture the $4$-APs in~$A$ that appear at the labels of the corners of some square.

Hence, to implement a complete reduction along these lines the crucial question is: Are there $N$ labelings such that \emph{every} 4-AP in $\set{0, \dots, n}$ appears as a square in at least one of these labelings? To our initial surprise such labelings exist indeed, at least modulo some technical annoyances. Though tricky, our final proof is unexpectedly (and pleasantly) short;\conferenceversion{the proof can be found in the full version of the paper}{ see \cref{sec:4APhard-identities}}.

\paragraph{Reduction from Square Detection to Identity Checking}
Let us finally return to the identity checking problem. Before we consider general identities, we focus on obtaining lower bounds for the following concrete, representative identity:
\begin{equation} \label{eq:overview}
    ((a\odot b) \oplus (a \odot c)) \oplus c = \infty.
\end{equation}
Note that this is exactly identity~\eqref{eq:intermediate-identity} discussed in the introduction (up to renaming ``$0$'' to ``$\infty$'').

We will describe how to reduce a Square Detection instance $M$ to checking this identity. In fact, by a simple trick it is equivalent to reduce to the augmented identity
\begin{equation} \label{eq:overview-augmented}
    \underbrace{((a\oplus_1 b) \oplus_3 (a \oplus_2 c)) \oplus_4 c}_{f(a, b, c)} = \infty,
\end{equation}
where $\oplus_1, \oplus_2, \oplus_3, \oplus_4$ are four different operations that we are free to define. Consider the following creative reduction. Intuitively, we will identify the three variables $a, b, c$ with the three indices $i, j, k$ via the following correspondence:
\begin{align*}
    a &= i, \\
    b &= i - j, \\
    c &= i + j + k.
\end{align*}
Then our goal is to assign the four operations $\oplus_1, \oplus_2, \oplus_3, \oplus_4$ in such a way that $f(a, b, c) \neq \infty$ if and only if the four constraints $M(i, j) = 1$, $M(i, j + k) = 1$, $M(i + k, j + k) = 1$, and $M(i + k, j) = 1$ are satisfied.

Focus on the first constraint, $M(i, j) = 1$. To encode this constraint into $\oplus_1$ we ensure that $a \oplus_1 b = \infty$ whenever $M(a, a - b) = 0$. We further make sure that the other operations preserve $\infty$ (i.e., $\infty \oplus_3 x = \infty$) so that whenever the constraint $M(i, j) = 1$ is violated we immediately satisfy the identity $f(a, b, c) = \infty$. We can similarly encode the second constraint $M(i, j + k)$ into $\oplus_2$.

Encoding the third constraint, $M(i + k, j + k) = 1$, into $\oplus_3$ is slightly more interesting. Note that $\oplus_3$ does not directly depend on the variables in $f$, but rather on the outputs of the operations $a \oplus_1 b$ and $a \oplus_2 c$. But we have not yet defined these operations---safe for the cases where already one of the first two constraints are falsified. We can thus carefully define the operations $\oplus_1, \oplus_2$ in such a way that $\oplus_3$ gets access to $i + k$ and $j + k$, as desired. Specifically, by picking $a \oplus_1 b := b = i - j$ and $a \oplus_2 c := c - a = j + k$ we encode the constraint $M(i + k, j + k) = 1$ by ensuring that $x \oplus_3 y = \infty$ whenever $M(x + y, y) = 0$. The same game can be played to encode the final constraint into $\oplus_4$.

Overall we find it fascinating how this idea relates algebraic identity testing to a problem as seemingly unrelated as Square Detection.

\subsection{Classification} \label{sec:overview:sec:classes}
Before we pinpoint the complexity of \emph{all} identities, let us first focus on a simpler class which captures the main difficulties, and can be seen as the core of our trichotomy. Intuitively, we restrict our attention to testing identities where one side is a \emph{constant} symbol, say $f(a, b, c) = \infty$. Formally, the \emph{Constant Term Identity Checking} (CTIC) is to decide if $f(x,y,z) = f(x',y',z')$ for all triples ${(x,y,z), (x',y',z')\in S^3}$. 
At the end of this overview we describe how to generalize to non-constant identities.

In the next paragraphs we provide an overview of our heavy tools and techniques required for establishing the trichotomy of CTIC.
Although obtaining lower bounds for some canonical identities is somewhat direct from the considered hypotheses, we need more structural insights in order to classify general (constant-term) identities.


\paragraph{Subexpression Embedding Gadget}
The first such tool is a convenient \emph{bookkeeping} trick, which allows us to assume that each operation does not only compute its value, but also carries information encoding the full history of operations in the subtree below. This is exactly the trick which allowed us to consider identities \eqref{eq:overview-augmented} rather than \eqref{eq:overview}.

More precisely, for any fixed expression $f(a,b,c)$, given an algebraic structure $(S,\oplus_1,\dots, \oplus_p)$, we can construct an algebraic structure $(\tilde S,\tilde \oplus_1,\dots, \tilde\oplus_p)$ such that the following properties are satisfied.
\begin{enumerate}[label=(\roman*)]
    \item \emph{Bookkeeping:} For any triple $x,y,z\in S$, and any subexpression $u(a,b,c)$ of $f$, $\tilde u(x,y,z) = \left(u(x,y,z), \mathbf{u}\right)$, where $\mathbf{u}$ is a symbolical expression representing $u$. 
    Moreover, for inputs $(x,y,z)\not \in S^3$, we have $\tilde f(x,y,z) = \infty$.
    \item \emph{Efficiency:} Given $(S,\oplus_1,\dots, \oplus_p)$, we can construct $(\tilde S, \tilde\oplus_1,\dots, \tilde\oplus_p)$ deterministically in time $\bigO(|S|^2)$ (in particular, $|\tilde S| = \bigO(|S|)$).
\end{enumerate}
Equipped with this tool, we consider the generalization of the CTIC problem that we call \emph{Multichromatic Constant Term Identity Checking Problem} (MCTIC): given $(S,\oplus_1,\dots, \oplus_p)$, and $A,B,C\subseteq S$, decide if $f(x,y,z) = f(x',y',z')$ for all triples $x,x'\in A$, $y,y' \in B$, $z,z'\in C$.
Note that setting $A= B =C = S$ yields precisely the CTIC problem.
What's more, the Subexpression Embedding Gadget can be used to prove that the complexity of these two problems is equal (up to a constant factor), hence it is sufficient to classify the complexity of MCTIC problem.

\paragraph{Collapsing Lemmas} One of the crucial tools that we use for our classification are the so-called \emph{Collapsing Lemmas}.
Intuitively, we say that an expression $f$ \emph{collapses} into $g$ if $f$ can be used to simulate $g$.
More precisely, if we are given an input to $g$, we can construct an input for $f$, such that they always agree (i.e. $f(x,y,z) = g(x,y,z)$ for every triple $x,y,z$).
Throughout the paper we prove several Collapsing Lemmas. At a high level they state that if two subexpressions $u,w$ of $f$ satisfy certain properties, they can be independently collapsed into much simpler expressions, while preserving the desired structure of $f$.

A particularly powerful Collapsing Lemma (see \conferenceversion{full version for details}{\Cref{lemma:collapsing-non-similar}}) relies on the following notion of \emph{similar} expressions. We say that $f(a,b,c)$ is similar to $g(a,b,c)$ (written $f\sim g$) if there exists an expression $X(a,b,c)$ as well as expressions $H_f, H_g$ such that $f(a,b,c) \equiv H_f(X(a,b,c))$ and $g(a,b,c) \equiv H_g(X(a,b,c))$---that is, $f$ and $g$ can be decomposed into a common subexpression~$X(a,b,c)$. We note that $\sim$ is an equivalence relation\conferenceversion{}{ (see Lemma~\ref{lemma:equivalence-of-similarity})}.
The Collapsing Lemma states that if $f(a,b,c)$ contains two \emph{dissimilar} subexpressions $I(a,b) \not\sim J(a,b)$ that both depend on two variables $a, b$, then, for arbitrary given operation $\odot$, we can collapse $I(a,b)$ into $a\odot b$ and $J(a,b)$ into $b$ (or vice versa) independently. This allows us to dramatically simplify formulas throughout: For instance, the overly complicated expression $(a(b + (a - a)) \oplus b) \otimes (a \odot b)$ might as well be replaced by $b \otimes (a \odot b)$.

In the following paragraph we give an overview of how this Collapsing Lemma can be used to prove the Triangle Detection hardness of the CTIC for the expressions that cannot be written as $H(G(a,b),c)$ for any expressions $G,H$, by highlighting one example of such expression as a case in point.
For details and full hardness classification of such expressions we refer the reader to \conferenceversion{the full version of the paper.}{Proposition \ref{prop:triangle-hardness}.}

\paragraph{Triangle-Time Regime}
Consider the following algebraic expression: $f(a,b,c)\equiv (c(abb)c) + (babb)$ (where by $ab$ we understand $a\cdot b$, and by e.g.\ $abc$ we understand $(a\cdot b)\cdot c$).
It is easy to see that for no expressions $G,H$, can $f$ be written as $H(G(a,b),c)$ (or any symmetric form).
Moreover, we can observe that $abb\not\sim babb$.
Hence, by the Collapsing Lemma presented above, we can collapse $abb$ and $babb$ into $b$ and $a\cdot b$.
Hence, $f$ is collapsed into $f'(a,b,c) \equiv (cbc) + (a\cdot b)$.
Given a graph $\Gamma = (V,E)$, we can now set the operation $\cdot$ as follows.
\[
x+y = x\cdot y = \begin{cases}
    x & \text{if $x=y$}\\
    x & \text{if $\{x,y\}\in E$}\\
    \infty & \text{otherwise}
\end{cases}
\]
where $\infty$ is an element that we choose from $S$ that corresponds to an isolated vertex in $\Gamma$ (hence $\infty \cdot x = x\cdot \infty = \infty$, same for $+$).
It is now easy to verify that the expression $f'(a,b,c)$ evaluates to $\infty$ if and only if the triple $a,b,c$ does not form a triangle, and to $c$ otherwise.

\paragraph{Cubic-Time Regime} Recall that earlier in this section we focused on a concrete expression $f(a,b,c)\equiv ((a\odot b) \oplus (a\odot c)) \oplus c$, for which we argued that, assuming the $4$-AP Hypothesis, the CTIC on this expression falls into the Cubic-Time Regime.
In our full proof we similarly consider the following six important expressions, and show similar hardness results.
\begin{align*}
    f_1(a,b,c) &= \big ((a\oplus_1 b)\,  \oplus_3 \, (a\oplus_2 c)\big ) \oplus_4 c, \\
    f_2(a,b,c) &=  (a\oplus_1 b)\, \oplus_4\, \big ((a\oplus_2 c)\oplus_3 b \big), \\
    f_3(a,b,c) &=  \big (\big ((a\oplus_1 b)\oplus_2 c\big )\oplus_3 a\big ) \oplus_4 b, \\
    f_4(a,b,c) &=  \big (\big ((a\oplus_1 b)\oplus_2 c\big )\oplus_3 a\big ) \oplus_4 c, \\
    f_5(a,b,c) &=  \big ((a\oplus_1 b)\oplus_3 (a\oplus_2 c)\big )\oplus_4 a, \\
    f_6(a,b,c) &= (a\oplus_1 b) \oplus_4 \big (a\oplus_3 \big (c\oplus_2 (a\oplus_1 b)\big )\big).
\end{align*}
In fact, for four out of these six expressions, an adaptation of the reduction from Square Detection as described above is possible, but for the remaining two such an approach seems not to be enough.
To overcome this technical barrier, we come up with yet another interesting intermediate 4-AP-hard problem that we call \emph{T Detection}.

Perhaps it seems odd that we specifically consider six out of the infinitely many hard identities. But as it turns out, building on the tools introduced earlier, we show that \emph{every} expression $g(a,b,c)$ that does not fall in the Triangle-time regime can be collapsed into one of these six special functions $f_1, \dots, f_6$.

As a case in point, we sketch the proof for one specific expression here:
\[g(a,b,c) \equiv (bba) + (aba + (c(bba))c)\]
Here, for simplicity we understand $ab$ to be $a\cdot b$, $aba$ to be $(a\cdot b)\cdot a$, and we understand the $\cdot$ operation to have priority over $+$ operation (as usual). It is easy to verify that $g$ cannot be written as $J(I(H(a,b),c), G(a,b))$ (or any symmetric shape) for any expressions $G,H,I,J$.
We can first apply the Collapsing Lemma to collapse (independently) the subexpressions $bba$ and $aba$ into $a\oplus_1 b$ and $b$ respectively.
This implies that $g$ can be collapsed into the following expression.
\[
g'(a,b,c) \equiv (a\oplus_1 b) + (b + (c(a\oplus_1 b))c)
\]
Now by using the Subexpression Embedding Gadget and Chromaticity, we can make sure that $c(a\oplus_1 b) = a\oplus_1 b$.
In particular, this collapses $g$ further into
\[
g''(a,b,c) \equiv (a\oplus_1 b) + (b + (a\oplus_1 b)c).
\]
It is now straightforward to see that we can obtain $f_6$ from $g''$.
Of course, we omitted many details here and only considered a single cherry-picked expression. We refer the reader to the full version of the paper for more general proofs with all the details included. 

\paragraph{Generalization to Non-Constant Identities}
The only remaining step towards proving our lower bounds is to show an efficient reduction from CTIC to the respective identity checking.
In particular, informally, we show that in almost all cases, verifying $f=g$ is at least as hard as the more difficult of the two CTIC problems for $f$ and $g$.

The strategy towards proving this is to add an element to our set $S$ that we call $\infty$, such that ${f(\infty, \infty, \infty) = \infty = g(\infty, \infty, \infty)}$. 
It then suffices to show that we can modify any given algebraic structure $(S,\oplus_1,\dots, \oplus_p)$ in such a way that for any triple $x,y,z$, the value $f(x,y,z)$ remains unchanged, while $g(x,y,z)= \infty$ (and vice versa) for any natural pair of expressions $f,g$. This gives a canonical reduction from the Constant Term Identity Checking to (natural) Identity Checking and thus concludes our classification task.

\paragraph{Strong Zero Triangle Hardness}
Finally, we outline how to show \Cref{thm:hardest,th:dist-counting}, our hardness results based on the Strong Zero Triangle Hypothesis.

Let Small-Weight Zero Triangle denote the problem of detecting a zero-weight triangle in an $n$-node graph with edge weights in $\{-n, \dots, n\}$. The current best algorithm for detecting zero-weight triangles in graphs with weights $\{-M, \dots, M\}$ runs in time $O(Mn^\omega)$ \cite{AlonGM97}, thus a reasonable hypothesis is that Small-Weight Zero Triangle cannot be solved in subcubic time $O(n^{3-\epsilon})$. This is known as the \emph{Strong Zero Triangle Hypothesis}~\cite{AbboudBKZ22}. In \conferenceversion{the full version of the paper}{Section \Cref{sec:assumptions}}, we will review this and other hardness hypotheses in detail.

Both theorems rely on the simple bookkeeping trick introduced previously.
Let us first talk about how to prove the hardness of the first identity $\left((a\odot b)\oplus(a\odot c)\right)\oplus(b\odot c)=0.$
The simple idea of encoding edge weights into the operation $x\odot y$ works. In order to add the edge weights, we define $x \oplus y$ to return the sum of edge weights if they exist or some annihilator element otherwise. We incorporate a second variable, which tracks how many times we performed the $\oplus$ operation, so that the final $\oplus$ operation performs a last equality to zero check.

In order to show the small-weight zero triangle hardness of counting distributive triples, we proceed as follows.
We want to show hardness of the task of counting $x,y,z\in S$ such that $x\odot(y\oplus z) = (x\odot y) \oplus (x\odot z)$.
We like to encode a zero triangle $u,v,w \in V$ in such a way that $x_u,y_v, z_w \in S$ form a distributive triple that is $x_u \odot (y_v \oplus z_w) = (x_u \odot y_v) \oplus (x_u \odot z_w)$ holds.
The high level idea of the construction now consists of the left side of the distributivity equation to equal the edge-weight sum $c(u,v) + c(v,w)$ and the right hand side to equal $-c(u,w)$. Clearly, this happens when $u,v,w$ are a zero-triangle. This can be again achieved with the bookkeeping trick, with the drawback that we create some extra help-variables.
We ensure that any triple $x,y,z \in S$, where at least one of $x,y$ or $z$ is a help-variable, will always be distributive.  In the end it remains to count, whether we have more distributive triples than the help-variables will inevitably produce.
A detailed discussion of Strong-Zero-Triangle hardness variants is provided in the full version of the paper. 

\newcommand\AP{\operatorname{AP}}

\section{Arithmetic Progression Detection: A New Source of Hardness?} \label{sec:ap}
Recall that a sequence of numbers $a_1, \dots, a_k$ with common step-width $a_2 - a_1 = \dots = a_k - a_{k-1}$ is called \emph{$k$-term arithmetic progression}, or \emph{$k$-AP} for short. The \emph{$k$-AP Detection} problem is to decide if a given set of numbers $A \subseteq \{0,\ldots,n\}$ contains a non-trivial $k$-AP. We emphasize that here~$n$ is not only a bound on the size of~$A$, but on the underlying universe. The trivial $O(n^2)$-time algorithm is to enumerate all combinations of the first two terms in the progression, and to test in constant time $O(1)$ whether the remaining $k-2$ terms are present. Perhaps surprisingly at first sight, this algorithm is essentially the fastest-known for~$k \geq 4$. Indeed, only (sub-)logarithmic improvements are known: either based on the tabulation-based ``Four-Russians'' approach \conferenceversion{}{(e.g., combining \Cref{lem:ap-to-hyperclique} with~\cite{Nagle10,AbboudFS24})}, or by exploiting Szemerédi's theorem which states that sets with sufficiently large density necessarily contain a $k$-AP,\footnote{Specifically, the state-of-the-art bounds imply that sets with density at least $(\log n)^{-c}$ necessarily contain a $4$-AP~\cite{GreenT17}, and sets with density at least $\exp(-(\log \log n)^{c_k})$ necessarily contain a $k$-AP for $k > 4$~\cite{LengSS24}.} and therefore need not be checked explicitly. Given this state of the art, it is natural to propose the following hypothesis for all constants $k \geq 4$.

\hypkap*

In what follows, we will make a case that \Cref{hyp:k-ap} deserves to be established as a fine-grained hypothesis. This involves arguing that the hypothesis is (1) \emph{plausible} and (2) \emph{useful} (in the sense that it has many interesting applications). The well-established fine-grained hypotheses are both, though of course there is a trade-off between (1) and (2): the stronger the assumption, the less plausible it is, yet the more powerful it becomes in the design of fine-grained reductions. We will argue that the 4-AP Hypothesis is among the stronger and yet still very believable hypotheses. In addition, the 4-AP Hypothesis is arguably simple to state, and, as we will see shortly, fits naturally into the big picture of fine-grained complexity.

Specifically, we justify the plausibility of \Cref{hyp:k-ap} in \Cref{sec:ap:sec:arith-graphs,sec:ap:sec:approaches}, and discuss its usefulness in \Cref{sec:ap:sec:uses}.

\subsection{Arithmetic Analogue of Hyperclique} \label{sec:ap:sec:arith-graphs}
A main argument for the plausibility of the $k$-AP Hypothesis is that it can be seen as the \emph{arithmetic analogue} of another well-established hardness assumption---the \emph{Hyperclique Hypothesis}. To appreciate this connection, let us take a moment to survey this analogy in a broader context.

\paragraph{Arithmetic versus Graph Problems}
A recurring phenomenon in fine-grained complexity is that \emph{arithmetic} problems (i.e., problems on sets or sequences of numbers) have natural counterparts among \emph{graph} or \emph{matrix} problems. For example, two of the three primary fine-grained hypotheses---the famous \emph{3SUM} and \emph{APSP} Hypotheses---are part of this pattern: 3SUM\footnote{Decide if a given size-$n$ set contains three numbers that sum to zero.} is a fundamental arithmetic problem~\cite{GajentaanO95}, known to be tightly connected to the corresponding \emph{Zero Triangle}\footnote{Decide if an edge-weighted $n$-vertex graph contains a triangle whose edge weights sum to zero.} problem on graphs~\cite{VassilevskaW18}. APSP\footnote{Compute all pairs of distances in an edge-weighted $n$-vertex graph.} is a fundamental graph problem, known to be fine-grained equivalent to \emph{Min-Plus Product}\footnote{Given integer $n \times n$ matrices $A, B$, compute a third matrix $C$ defined by $C[i, j] = \min_k (A[i, k] + B[k, j])$.}, whose arithmetic analogue is called \emph{Min-Plus Convolution}\footnote{Given length-$n$ integer sequences $a, b$, compute a third sequence $c$ defined by $c[k] = \min_{i+j=k} (a[i] + b[j])$.}~\cite{CyganMWW17,KunnemannPS17}. This analogy underlies many interesting fine-grained connections, which in particular led Zero Triangle and Min-Plus Convolution to become hardness assumptions in their own right. Many further instances exist---for instance, between \emph{Min-Max Convolution} and \emph{Min-Max Product} (which is equivalent to \emph{All-Pairs Bottleneck Paths})~\cite{VassilevskaWY07}.

For all these pairs of problems the analogy is not just based on intuition, but is indeed \emph{systematic} in the following sense. For one, known algorithms transfer from one setting to the other up to mechanical changes (for instance, replacing FFT and fast matrix multiplication as the two algebraic primitives for arithmetic and graph problems, respectively), Moreover, in all cases there is a mechanical fine-grained reduction from the arithmetic problem to the corresponding graph problem (that transforms sets or sequences of numbers into their \emph{Cayley graphs}). This implication is only known to be one-directional, though. In particular, arithmetic hardness assumptions are generally (1) \emph{less plausible} than their graph counterparts, but (2) have \emph{more useful} implications.

Despite the missing formal equivalence, we are aware of \emph{no separation} between the arithmetic and graph variants for any corresponding pair. Every known algorithm and lower bound synchronously either applies to both variants (under mechnical adaptions) or to neither. This is widely regarded as heuristic evidence that corresponding pairs of arithmetic and graph problems are at least \emph{morally} equivalent.

\paragraph{Hyperclique Detection}
The \emph{$(k, h)$-Hyperclique} problem is to detect a $k$-clique in an $h$-uniform hypergraph (i.e., a set of $k$ vertices among which all possible hyperedges are present). The fastest-known algorithms for this problem essentially take brute-force time $O(n^k)$, which motivated Lincoln, Vassilevska~W., and Williams~\cite{LincolnWW18} to propose the following hypothesis for all constants $k > h \geq 3$.

\begin{hypothesis}[$(k, h)$-Hyperclique]
There is no algorithm to detect a $k$-clique in an $h$-uniform $n$-vertex hypergraph in time $O(n^{k-\epsilon})$ (for any constant $\epsilon > 0$).
\end{hypothesis}

Though this is a rather recent hypothesis (proposed in 2018~\cite{LincolnWW18}), its plausibility is well-defended (see e.g.\ the discussion in~\cite{LincolnWW18}, and also~\cite{AbboudBDN18}), and it has already led to many and diverse conditional lower bounds in the few years since~\cite{AbboudBDN18,LincolnWW18,BringmannFK19,CarmeliZBKS20,KunnemannM20,AnGIJKN21,BringmannS21,DalirrooyfardW22,Kunnemann22,GorbachevK23,Zamir23,GokajK25,FischerKRS25, FuLR25}, notably for subgraph detection problems~\cite{LincolnWW18,BringmannS21,DalirrooyfardW22, FuLR25}, geometric problems~\cite{Kunnemann22,GorbachevK23}, and logic-definable problem classes~\cite{BringmannFK19,KunnemannM20,AnGIJKN21,GokajK25,FischerKRS25}. 

In light of the wide-spread analogy between arithmetic and graph problems, it is perhaps surprising that so far no arithmetic analogue of Hyperclique has been identified. As we will see, this missing analogue is exactly AP Detection.

\paragraph{AP Detection versus Hyperclique Detection}
The state of affairs for $k$-AP Detection and \makebox{$(k, k-1)$}-Hyperclique detection is strikingly similar. First, both Hypotheses are wrong for $k = 3$: $3$-AP Detection can be solved in near-linear time by the Fast Fourier Transform, whereas $(3, 2)$-Hyperclique Detection (i.e., \emph{Triangle Detection} in normal graphs) is in matrix multiplication time~$O(n^\omega)$. This is an instance of the natural correspondence between FFT and matrix multiplication discussed earlier.

Second, for $k \geq 4$ both problems remain stuck at their brute-force running times. The \emph{same} algorithmic approaches fail for the \emph{same} reasons, as we discuss further in \Cref{sec:ap:sec:approaches}.

Third, and most importantly, we establish a formal reduction from $k$-AP Detection to $(k, k-1)$-Hyperclique Detection. Recall that this is in line with all other pairs of related problems, where the arithmetic problem reduces to the graph problem.

\begin{theorem} \label{thm:ap-to-hyperclique}
For every $k \geq 4$, the $k$-AP Hypothesis implies the $(k, k-1)$-Hyperclique Hypothesis.
\end{theorem}

At a high level, this reduction follows the same mechanical idea used in other arithmetic-to-graph reductions: transforming a set of numbers into its Cayley graph. The challenge is to find an appropriate generalization to hypergraphs. Conceptually, our construction is inspired by Gowers' proof~\cite{Gowers07} of Szemerédi's theorem. \conferenceversion{For the reduction details, we refer the reader to the full version of the paper.}{}

Another notable connection is that the difficulty of both problems increases as the relevant parameters grow: $k$ for $k$-AP and $h$ for $(k, h)$-Hyperclique. This is formalized in the following theorem. \Cref{fig:ap-reducts} illustrates the resulting web of known reductions. \begin{figure*}[t]
\centering
\def\gap{2.5mm}
\definecolor[named]{LightBlue}{HTML}{8D9DB6}
\definecolor[named]{PaleBlue}{HTML}{F2F5F8}
\begin{tikzpicture}[
    prob/.style={
        draw,
        fill=PaleBlue,
        rounded corners,
        font=\small,
        minimum width=34mm,
        minimum height=8mm,
        inner sep=2mm,
        align=center
    },
    important/.style={
        fill=LightBlue,
    },
    reduct/.style={
        draw,
        >=latex,
        shorten <=.8mm,
        shorten >=.8mm,
    },
    every edge quotes/.append style={font=\small},
    region/.style={
        draw=black!75,
        dashed,
        rounded corners,
    },
]

\matrix[
    anchor=east,
    inner sep=0pt,
    every node/.style={anchor=center},
    column sep=18mm,
    row sep=8mm,
    row 3/.style={row sep=14mm},
] at (.5\textwidth-0.4pt-\gap, 0) {
    &
    \node[prob] (k3hc) {$(k, 3)$-Hyperclique}; &
    \node[prob] (k4hc) {$(k, 4)$-Hyperclique}; &
    \node[prob] (kkm1hc) {$(k, k\mathord-1)$-Hyperclique}; \\

    &
    \node[prob] (53hc) {$(5, 3)$-Hyperclique}; &
    \node[prob] (54hc) {$(5, 4)$-Hyperclique}; &
    \\
    
    &
    \node[prob] (43hc) {$(4, 3)$-Hyperclique}; &
    &
    \\

    &
    \node[prob, important] (4ap) {$4$-AP}; &
    \node[prob, important] (5ap) {5-AP}; &
    \node[prob, important] (kap) {$k$-AP}; \\

    &
    \node[prob] (4sum) {$4$-SUM}; &
    &
    \\
};

\newcommand\drawhalfdottedreduct[3]{
    \path[draw, shorten <=.8mm] (#1.east) -- ([shift={(3mm, 0)}]#1.east);
    \path[draw, dashed] ([shift={(3mm, 0)}]#1.east) to["#3"] ([shift={(-4mm, 0)}]#2.west);
    \path[draw, shorten >=.8mm, >=latex, ->] ([shift={(-4mm, 0)}]#2.west) -- (#2.west);}
\path[reduct, ->] (4ap) to["Thm.~\ref{thm:ap-larger-k}"] (5ap);
\drawhalfdottedreduct{5ap}{kap}{Thm.~\ref{thm:ap-larger-k}}
\path[reduct, ->] (4ap) to["Thm.~\ref{thm:ap-to-hyperclique}", swap] (43hc);
\path[reduct, ->] (5ap) to["Thm.~\ref{thm:ap-to-hyperclique}", swap] (54hc);
\path[reduct, ->] (kap) to["Thm.~\ref{thm:ap-to-hyperclique}", swap] (kkm1hc);
\path[reduct, ->] (53hc) to (54hc);
\path[reduct, ->] (53hc) to (43hc);
\path[reduct, ->] (k3hc) to (k4hc);
\drawhalfdottedreduct{k4hc}{kkm1hc}{}
\path[reduct, ->] (k3hc) to (53hc);
\path[reduct, ->] (k4hc) to (54hc);
\path[reduct, ->] (4ap) to["\conferenceversion{}{Thm.~\ref{thm:4ap-to-4sum}}~\cite{GokajK25}"] (4sum);

\begin{pgfonlayer}{bg}
    \path (-.5\textwidth+0.4pt, 0) coordinate (pic left);
    \newcommand\drawregionaround[3]{
        \path (#1.north east) + (\gap, \gap) coordinate (region north east);
        \path (#2.south) + (0, -\gap) coordinate (region south);
        \path[region] (pic left |- region north east)
            node[below right=2mm] {$n^{#3\pm o(1)}$}
            rectangle (region north east |- region south);}
    \drawregionaround{kap}{4sum}{2}
    \drawregionaround{43hc}{43hc}{4}
    \drawregionaround{54hc}{54hc}{5}
    \drawregionaround{kkm1hc}{kkm1hc}{k}
\end{pgfonlayer}
\end{tikzpicture}
\caption{Illustrates the known fine-grained reductions between $k$-AP Detection and other central hard problems.} \label{fig:ap-reducts}
\end{figure*}
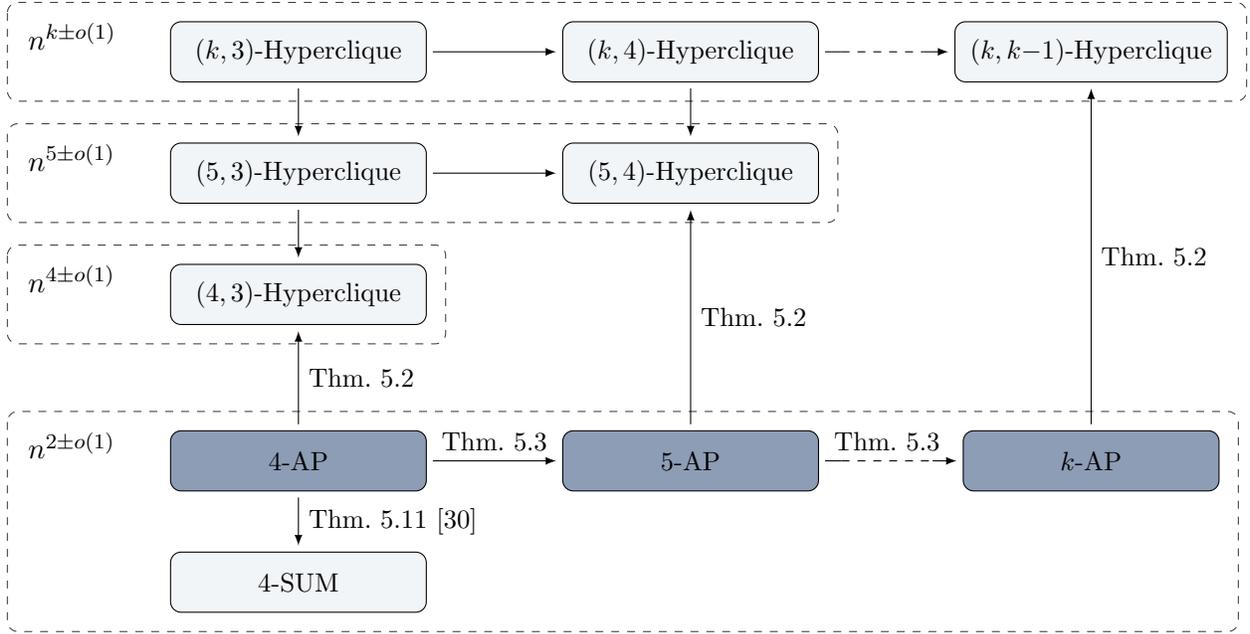

\begin{theorem} \label{thm:ap-larger-k}
The $k$-AP Hypothesis implies the $k'$-AP Hypothesis for all $k' \geq k$.
\end{theorem}

In summary, we have established that AP Detection is indeed the correct arithmetic analogue of the Hyperclique Detection problem. For this reason alone the $k$-AP Hypothesis deserves its place among the other fine-grained assumptions. One could even argue that it should be regarded as \emph{morally} equivalent to Hyperclique. Of course, there have been fewer explicit attempts to break the $k$-AP Hypothesis, so it is fair to consider it with care.

\subsection{Algorithmic Approaches Fail} \label{sec:ap:sec:approaches}
To further substantiate the believability of the $k$-AP Hypothesis, we review some natural algorithmic approaches and explain why they fail to refute it. These known approaches broadly fall into two categories: algebraic and combinatorial.

\paragraph{Algebraic Algorithms}
Since $3$-APs can be detected in near-linear time by an algebraic algorithm, perhaps it is reasonable to expect that the same could apply to $4$-APs. The algebraic algorithm for detecting $3$-APs can be viewed as efficiently evaluating the following \emph{trilinear} polynomial
\begin{equation*}
    \sum_{i, j=1}^n A_i B_{i+j} C_{i+2j}
\end{equation*}
on the $0$-$1$-inputs indicating the elements present in the given set. There is a rich theory surrounding the complexity of such trilinear problems; see e.g.~\cite{Blaser13}. Specifically, it is known that the algebraic complexity of any trilinear polynomial is \emph{exactly captured} by a notion called \emph{tensor rank}. For this particular trilinear polynomial the tensor-rank is linear. More broadly, this theory is what the FFT and Strassen-like algorithms for fast matrix multiplication are founded on.

To detect $4$-APs, the natural generalization is to consider the following \emph{$4$-linear} polynomial
\begin{equation*}
    \sum_{i, j=1}^n A_i B_{i+j} C_{i+2j} D_{i+3j}.
\end{equation*}
Unfortunately, most of the established theory does \emph{not} extend from trilinear to $4$-linear problems. In particular, the tensor rank characterization does not carry over. For this reason it remains open whether this polynomial can be evaluated in subquadratic complexity, even though one can actually prove that this 4-linear polynomial has maximal tensor rank $n^2$ (in analogy to~\cite[Section~8]{LincolnWW18}). This is by far not the only 4-linear polynomial with this state of affairs; another example is---not surprisingly at this point---the following polynomial to count $4$-hypercliques in a $3$-uniform hypergraph:
\begin{equation*}
    \sum_{i, j, k, \ell=1}^n A_{ijk} B_{ij\ell} C_{ik\ell} D_{jk\ell}.
\end{equation*}
To our knowledge it is a challenging open problem to find \emph{any} 4-linear polynomial with nontrivial algebraic complexity upper bounds (that cannot be realized by a trivial reduction to the trilinear case, see also~\cite{AustrinKK22}), or lower bounds (see~\cite{Raz10} and \cite[Open Problem~7]{ShpilkaY10}). Hence, to achieve a nontrivial algebraic algorithm for 4-AP Detection one likely has to overcome a deeper barrier.

\paragraph{Combinatorial Algorithms}
A very different approach is to consider purely \emph{combinatorial} algorithms, i.e., algorithms that do not rely on algebraic primitives such as FFT or matrix multiplication. The hope is that dense sets without $k$-term progressions exhibit a certain combinatorial structure that can be algorithmically exploited. This line of reasoning has led to surprising combinatorial algorithms for the related Triangle Detection problem~\cite{ArlazarovDKF70,BansalW12,Chan15,Yu18,AbboudFKLM24}, based on deep \emph{regularity} results in graphs~\cite{BansalW12,AbboudFKLM24}. The fastest-known such algorithm runs in time \smash{$n^3 / 2^{\Omega((\log n)^{1/7})}$}.

This has direct consequences for $3$-AP Detection: By the reduction from $3$-AP Detection to Triangle Detection (\conferenceversion{see full version of the paper}{\Cref{lem:ap-to-hyperclique}}), \cite{AbboudFKLM24} implies a combinatorial \smash{$n^2 / 2^{\Omega((\log n)^{1/7})}$}-time algorithm for $3$-AP Detection. An alternative (marginally slower) \smash{$n^2 / 2^{\Omega((\log n)^{1/9})}$}-time algorithm is to simply run the brute-force $O(|A|^2)$-time algorithm if \smash{$|A| \leq n / 2^{\Omega((\log n)^{1/9})}$}; otherwise the Kelley-Meka~\cite{KelleyM23} bounds for Roth's theorem, refined in~\cite{BloomS23}, imply that $A$ \emph{must} contain a $3$-AP, so we can simply report ``yes''.\footnote{In fact, these two alternatives are tightly related as the combinatorial Triangle Detection algorithm~\cite{AbboudFKLM24} itself is based on the ideas by Kelley and Meka~\cite{KelleyM23}, further developed by Kelley, Lovett and Meka~\cite{KelleyLM24}.} In summary, currently there is no combinatorial algorithm in truly subquadratic-time algorithms even for $3$-AP Detection. Moreover, such a hypothetical subquadratic-time algorithm could only follow along the lines of the first approach, as the Kelley-Meka bounds~\cite{KelleyM23} for 3-AP-free sets are almost tight~\cite{Behrend1946}. 

For detecting $4$-APs the state of affairs is even worse: The only improvements over the $O(n^2)$-time baseline are (sub-)logarithmic, either based on the simple ``Four-Russians'' trick mentioned before, or similarly using Szemerédi's theorem~\cite{GreenT17} to avoid testing dense sets explicitly. It is conceivable that higher lower-order improvements can be achieved, perhaps by directly combining the ``Four-Russians'' trick with ``higher-order Fourier analysis'' (a deep theory developed by Gowers~\cite{Gowers01} to quantitatively improve Szemerédi's theorem, see e.g.\ the textbook~\cite{TaoV06}) in analogy to~\cite{BansalW12}; however, the involved quantitative bounds are still quite weak and would only lead to logarithmic speed-ups, if any.

In summary, for a combinatorial algorithm to break the $4$-AP Hypothesis one would have to overcome two barriers: (1) To extend from $3$-APs to $4$-APs (i.e., from graphs to hypergraphs) without quantitative losses, which likely involves extending the Kelley-Meka bounds to $4$-AP-free sets, and (2) to achieve truly subquadratic-time algorithms for $3$-AP Detection (i.e., to improve the quasi-polynomial speed-up to polynomial). Barrier (1) already appears hard, but (2) seems completely out of reach for now.

\paragraph{A New Approach?}
Currently, both algebraic and combinatorial methods fail to refute the $k$-AP Hypothesis. Of course it is conceivable that a radically different (perhaps hybrid) approach could exist. While we cannot rule out this possibility, we consider it likely that such an alternative approach would analogously be applicable to the Hyperclique problem as well.

\subsection{Usefulness for Reductions} \label{sec:ap:sec:uses}
In \Cref{sec:ap:sec:arith-graphs,sec:ap:sec:approaches} we have argued that the $k$-AP Hypothesis is plausible. Let us finally comment why the $k$-AP Hypothesis, and in particular the $4$-AP Hypothesis, is also very useful in the design of fine-grained reductions.

The first indication is, of course, that in this paper we establish several new lower bounds based on the 4-AP Hypothesis---not only for verifying algebraic identities, but also for problems like Square Detection which we believe are of independent interest.

Second, \Cref{thm:ap-to-hyperclique} entails that \emph{all} lower bounds conditioned on the $(k, k-1)$-Hyperclique hypothesis can now alternatively be based on the $k$-AP Hypothesis. As mentioned before, this long list~\cite{AbboudBDN18,LincolnWW18,BringmannFK19,CarmeliZBKS20,KunnemannM20,AnGIJKN21,BringmannS21,DalirrooyfardW22,Kunnemann22,GorbachevK23,Zamir23,GokajK25,FischerKRS25} involves lower bounds for many important problems, such as Klee's measure problem~\cite{Kunnemann22} or induced 4-cycle detection~\cite{DalirrooyfardW22}, to highlight a few.

Third, it follows from work by Gokaj and Künnemann~\cite{GokajK25}, and was independently also observed by Adam Polak and collaborators~\cite{Polak_personal_communication}, that the 4-AP Hypothesis also implies the \emph{4-SUM Hypothesis}\footnote{Every algorithm to decide if in a given set there are four numbers that sum to zero takes time $n^{2-o(1)}$.}, and for this reason \emph{all} lower bounds based on 4-SUM can be alternatively based on the $4$-AP Hypothesis. This includes known lower bounds for geometric problems~\cite{KunnemannN22,GokajKST25}, but also other arithmetic problems~\cite{JinX23}. In \conferenceversion{the full version of the paper}{ Section ~\Cref{sec:ap:sec:4sum}}, we include a short self-contained proof of this fact.

Fourth, let us highlight a conceptual advantage of the $k$-AP Hypothesis over all related fine-grained assumptions: It concerns a ``dense'' arithmetic problem. Indeed, 3SUM and all its variants can be viewed as problems on a sparse integer set (over a universe of at least quadratic size). This is exactly the reason why the $k$-AP Hypothesis succeeds as a hardness assumption for other dense problems (like Square Detection), and why attempts based on the 3SUM Hypothesis fail. For that reason in particular, we are optimistic that the $k$-AP Hypothesis will lead to more fine-grained lower bounds in the future.
\conferenceversion{}{
\subsection{Basic Reductions for \texorpdfstring{$k$}{k}-AP} \label{sec:reducts:sec:chromatic}
We start with some basic reductions involving the $k$-AP problem. First, it will be convenient on several occasions to also consider the following \emph{multichromatic} version of the problem:

\begin{problem}[Multichromatic $k$-AP Detection]
Given sets $A_1, \dots, A_k \subseteq \set{0, \dots, n}$, decide if there is a (possibly trivial) $k$-AP $(a_1, \dots, a_k) \in A_1 \times \dots \times A_k$.
\end{problem}

The first step is to show that the monochromatic and multichromatic versions of $k$-AP Detection are equivalent. We rely on the following lemma based on Behrend's seminal construction of 3-AP-free sets that follows e.g.\ from~\cite[Theorem~3.10]{DudekGS20}.

\begin{lemma}[Behrend's Construction] \label{lem:behrend}
For any integer $n$, the set $\set{0, \dots, n}$ can be partitioned into sets $F_1, \dots, F_N$ such that $N \leq n^{o(1)}$ and such that each set $F_i$ contains no non-trivial 3-AP. The sets $F_1, \dots, F_N$ can be computed in deterministic time $O(n)$.
\end{lemma}

\begin{lemma}[Monochromatic versus Multichromatic $k$-AP] \label{lem:ap-chromatic}
If multichromatic $k$-AP Detection is in time $T(n)$, then monochromatic $k$-AP Detection is in time $O(T(n))$. Conversely, if monochromatic $k$-AP Detection is in time $T(n)$, then multichromatic $k$-AP Detection is in time $T(n) \cdot n^{o(1)}$.
\end{lemma}
\begin{proof}
The first claim is immediate by color-coding, i.e., randomly partition $A = A_1 \cup \dots \cup A_k$. With probability $k^{-k} = \Omega(1)$ any fixed (non-trivial) $k$-AP is hashed perfectly into $A_1 \times \dots \times A_k$. See~\cite{AlonYZ95,SchmidtS90} for the standard derandomization of this trick.

For the rest of this proof focus on the converse direction. Let $A_1, \dots, A_k \subseteq \set{0, \dots, n}$ be the given sets. We construct 3-AP-free sets $F_1, \dots, F_N$ that partition the universe $\set{0, \dots, n}$ by \Cref{lem:behrend}. Then we construct sets $B_\ell$, where $\ell$ ranges over all $k$-tuples $\ell \in [N]^k$, by
\begin{equation*}
    B_\ell = \bigcup_{i=1}^k ((A_i \cap F_{\ell_i}) + i \cdot 10 n).
\end{equation*}
We test if there is a monochromatic $k$-AP in any of the sets $B_\ell$. If yes, we report that there is a multichromatic $k$-AP in $(A_1, \dots, A_k)$.

Observe that $B_\ell \subseteq \set{0, \dots, O(n)}$, and thus the running time is $T(O(n)) \cdot N^k = T(n) \cdot n^{o(1)}$ (using that $T(O(n)) = O(T(n))$ as any $k$-AP over the universe $\set{0, \dots, O(n)}$ can be trivially reduced to $O(1)$ instances over the universe $\set{0, \dots, n}$). We analyze the correctness in the following.

We first show that if there is a $k$-AP $(b_1, \dots, b_k)$ in some set $B_\ell$, then there is a $k$-AP in $(A_1, \dots, A_k)$. The $k$-AP cannot be fully contained in a \emph{segment} $(A_i \cap F_{\ell_i}) + i \cdot 10n$ since $F_{\ell_i}$ does not contain a 3-AP (and thus also no $k$-AP). It follows that the $k$-AP spans at least two distinct segments. However, two adjacent segments in $B_\ell$ are placed at a distance of at least $9 n$, so it follows that the step-width $b_{i+1} - b_i$ is at least~$9 n$. Therefore, two elements $b_i, b_{i+1}$ can never appear in the same segment, and it follows that $b_1, \dots, b_k$ spans all $k$ segments, i.e., $b_i \in (A_i \cap F_{\ell_i}) + i \cdot 10 n$. Define $a_i = b_i - i \cdot 10n$. Then $(a_1, \dots, a_k)$ is a multichromatic $k$-AP in $(A_1, \dots, A_k)$ as claimed.

It remains to show that if there is a $k$-AP $(a_1, \dots, a_k)$ in $(A_1, \dots, A_k)$, then there is a $k$-AP in some set~$B_\ell$. \Cref{lem:behrend} guarantees that the sets $B_1, \dots, B_N$ cover the entire set $\set{0, \dots, n}$. In particular, there exist $\ell_1, \dots, \ell_k \in [N]$ such that $a_i \in B_{\ell_i}$ for all $i$. Then the sequence $(b_1, \dots, b_k)$ defined by $b_i = a_i + i \cdot 10n$ is a $k$-AP in $B_\ell$ for $\ell = (\ell_1, \dots, \ell_k)$.
\end{proof}

\begin{proof}[Proof of \Cref{thm:ap-larger-k}]
The proof is a simple consequence of \Cref{lem:ap-chromatic}. If monochromatic $k'$-AP Detection is in time~$T(n)$, then multichromatic $k'$-AP Detection is in time $T(O(n))$. From this it follows for all $k \leq k'$ that multichromatic $k$-AP Detection is in time $T(O(n))$ as well, as we can simply take the sets $A_{k+1}, \dots, A_{k'}$ to be copies of the full universe~$\set{0, \dots, O(n)}$. Finally, from the converse direction of \Cref{lem:ap-chromatic} it follows that monochromatic $k$-AP Detection is in time $T(O(n)) \cdot n^{o(1)}$.
\end{proof}

\subsection{Reduction from \texorpdfstring{$k$}{k}-AP to \texorpdfstring{$(k, k-1)$}{(k, k-1)}-Hyperclique} \label{sec:reducts:sec:hyperclique}
In this section we establish the reduction from $k$-AP Detection to $(k, k-1)$-Hyperclique Detection. Again, throughout $k \geq 3$ is assumed to be constant.

\begin{lemma} \label{lem:ap-to-hyperclique}
If $(k, k-1)$-Hyperclique Detection is in time $T(n)$, then $k$-AP Detection is in time $O(T(n^{2/k}))$.
\end{lemma}

The proof of \Cref{lem:ap-to-hyperclique} requires some setup. Let $\AP^k \subseteq \Rat^k$ denote the subset of $k$-term progressions. Observe that $\AP^k$ is a vector space of dimension $2$ as it can be described as the set of all points $(a_1, \dots, a_k)$ satisfying the $k-2$ linearly independent equations $a_i = (2-i) a_1 + (i-1) a_2$ for $i=3, \dots, k$. Define the~\makebox{$k \times k$} matrix $M$ by
\begin{equation*}
    M = \frac{1}{k-1}
    \begin{bmatrix}
        0      & -1     & -2     & \cdots & -(k-1) \\
        1      & 0      & -1     & \cdots & -(k-2) \\
        2      & 1      &  0     & \cdots & -(k-3) \\
        \vdots & \vdots & \vdots & \ddots & \vdots \\
        k-1    & k-2    & k-3    & \cdots & 0      \\
    \end{bmatrix}.
\end{equation*}
The point of this matrix is that its image (i.e., column space) is exactly the set of all $k$-APs; see the following observation. The normalization by $\frac{1}{k-1}$ is not strictly necessary but natural later on.

\begin{lemma} \label{lem:ap-image}
$\im M = \AP^k$.
\end{lemma}
\begin{proof}
Each column in $M$ is a $k$-term progression, hence its column space satisfies $\im M \subseteq \AP^k$. Conversely, any $k$-term progression $a = (a_1, \dots, a_k)$ can be expressed as a matrix-vector product $a = M x$, for instance for $x = (a_2(k-1), -a_1(k-1), 0, \dots, 0)$.
\end{proof}

\begin{lemma} \label{lem:ap-ruler-trick}
For any integer $n \geq 1$ there is a set $X \subseteq \set{0, \dots, O_k(n)}$ of size $O_k(n^{2/k})$ such that for every $k$-AP $a=(a_1, \dots, a_k) \in \set{0, \dots, n}^k$ there is some $x \in X^k$ with $M x = a$. The set $X$ can be constructed in deterministic time $O(|X|)$.
\end{lemma}
\begin{proof}
The intuition behind this lemma is simple: The affine vector space $\set{x : M x = a}$ has dimension~\makebox{$k - 2$}, and from this it follows that there are at least $n^{k-2}$ solutions $x \in \set{0, \dots, O_k(n)}^k$. Hence, we can choose $X$ as a \emph{uniformly random} set with sampling rate $\Theta_k(n^{-(k-2)/k})$ to hit at least one such solution in~$X^k$ with constant probability. The resulting set $X$ has size $O_k(n^{2/k})$ as claimed. In the following paragraphs we spell out the formal details of a deterministic construction; we encourage the readers mainly interested in the concepts behind the reduction to skip this part.

Let $q = \ceil{n^{1/k}}$. We will view integers in their base-$q$ representation; i.e., for $b \in \set{0, \dots, n}$ let us write $b[\ell]$ for the integer obtained from $b$ by zeroing out all $q$-ary digits except for the $\ell$-th one. Thus,~$b[\ell]$ is a multiple of $q^\ell$ that is less than $q^{\ell+1}$, and we can express \smash{$b = \sum_{\ell=0}^{k-1} b[\ell]$}. We pick
\begin{equation*}
    X = \set*{\sum_{\ell=0}^{k-1} d_\ell \cdot q^\ell : \text{$d_0, \dots, d_{k-1} \in \set{-O_k(q), \dots, O_k(q)}$ and at most two $d_\ell$'s are nonzero}}.
\end{equation*}
That is, $X$ essentially consists of the integers in $\set{0, \dots, n}$ where only two $q$-ary digits are nonzero; however, for technical reasons we also allow ``negative digits'' in the definition of $X$. We can nevertheless bound its size by $|X| \leq k^2 \cdot O_k(q)^2 = O_k(n^{2/k})$.

We now argue that for any $k$-AP $a=(a_1, \dots, a_k) \in \set{0, \dots, n}^k$ there exists $x \in X^k$ with $M x = a$. Pick $x = (x_1, \dots, x_k)$ as
\begin{equation*}
    x_j =
    \begin{cases}
        -(k-2) a_1[0] + (k-1) a_2[0] + (k-1) a_2[1] &\text{for $j = 1$,} \\
        (k-1)((j-3)a_1[j-1] - (j-2)a_2[j-1] - (j-1)a_1[j] + j a_2[j]) &\text{for $1 < j < k$,} \\
        (k-1)((k-3)a_1[k-1] - (k-2)a_2[k-1]) - a_1[0] &\text{for $j = k$.}
    \end{cases}
\end{equation*}
Observe that $x_j \in X$ (when choosing the hidden constants in the definition of $X$ sufficiently large), as in each case we assign only two of the ``$q$-ary digits'' of $x_j$ to a (possibly negative) value bounded by $O_k(q)$. This complicated definition is chosen in such a way that
\begin{align*}
    (M x)_i
    = \sum_{j=1}^k M_{ij} x_j
    = \sum_{j=1}^k \frac{i - j}{k-1} \cdot x_j
    = (2 - i) \sum_{\ell=0}^{k-1} a_1[\ell] + (i - 1) \sum_{\ell=0}^{k-1} a_2[\ell],
\end{align*}
as can be (tediously) verified by comparing the coefficients of $a_1[\ell]$ and $a_2[\ell]$ in the last step. It follows that
\begin{equation*}
    (M x)_i = (2 - i) a_1 + (i - 1) a_2 = a_1 + (i - 1) (a_2 - a_1) = a_i, \vphantom{\sum}
\end{equation*}
and thus $M x = a$ as claimed.
\end{proof}

\begin{proof}[Proof of \Cref{lem:ap-to-hyperclique}]
We reduce from a given multichromatic $k$-AP instance $A_1, \dots, A_k \subseteq \set{0, \dots, n}$ by \Cref{lem:ap-chromatic}. Construct the set $X$ by the previous \Cref{lem:ap-ruler-trick} (in negligible time). Then construct a $(k-1)$-uniform $k$-partite hypergraph $G$ on vertex parts $X_1, \dots, X_k$, where we identify each part with a copy of $X$. We add a hyperedge $(x_1, \dots, x_{i-1}, x_{i+1}, \dots, x_k) \in X_1 \times \dots \times X_{i-1} \times X_{i+1} \times \dots \times X_k$ if and only if
\begin{equation*}
    \sum_{j=1}^k M_{ij} x_j \in A_i.
\end{equation*}
Note that this expression is indeed well-defined as the missing term $x_i$ appears with coefficient~\makebox{$M_{ii} = 0$} in the sum. We will argue that this graph contains a size-$k$ hyperclique if and only if there is a $k$-AP in the given instance. Thus, it suffices to run the fast hyperclique detection algorithm on $G$. As $|X_1|, \dots, |X_k| = O(n^{2/k})$ (by \Cref{lem:ap-ruler-trick}) the total time is $O(T(n^{2/k}))$. (Here again, we use the fact that $T(O(N)) = O(T(N))$ by the trivial reduction from one Hyperclique instance of size $O(N)$ to $O(1)$ instances of size $N$.)

For the correctness, first suppose that there is a hyperclique $x = (x_1, \dots, x_k) \in X_1 \times \dots \times X_k$ in~$G$. We show that there is a $k$-AP in $(A_1, \dots, A_k)$. Indeed, take $a := M x$. By construction we have $a_i \in A_i$ for all~\makebox{$i \in [k]$}, and \Cref{lem:ap-image} implies that $a = Mx \in \im M \subseteq \AP^k$ is a $k$-term progression. Conversely, suppose that there is a $k$-AP $a = (a_1, \dots, a_k)$ in $(A_1, \dots, A_k)$. \Cref{lem:ap-ruler-trick} guarantees that there is a vector $x \in X^k$ such that $M x = a$. But then by construction $(x_1, \dots, x_k)$ forms a hyperclique in $G$.
\end{proof}

\begin{remark}
We remark that the same proof shows that \emph{counting} (multichromatic) $k$-APs reduces to \emph{counting} $(k, k-1)$-hypercliques. The only technical overhead is to make sure in \Cref{lem:ap-ruler-trick} that there is \emph{exactly one} solution $x$ to $M x = a$, which can be achieved by guessing some of the bits in $a_1, \dots, a_k$.
\end{remark}

\subsection{Reduction from \texorpdfstring{$4$}{4}-AP to \texorpdfstring{$4$}{4}-SUM} \label{sec:ap:sec:4sum}
Finally, we include a short self-contained proof that the 4-AP Hypothesis implies the 4-SUM Hypothesis~\cite{GokajK25}. A similar reduction has been obtained independently~\cite{Polak_personal_communication}.

\begin{theorem}[{{\cite{GokajK25}}}] \label{thm:4ap-to-4sum}
The $4$-AP Hypothesis implies the $4$-SUM Hypothesis.
\end{theorem}
\begin{proof}
We describe the transformation from a given multichromatic $4$-AP instance $(A_1, A_2, A_3, A_4)$ to an equivalent multichromatic $4$-SUM instance. By \Cref{lem:ap-chromatic}, and the basic fact that $4$-SUM and its colored counterpart are equivalent, this entails the claimed reduction. We transform each of the sets $A_i$ as follows:
\begin{align*}
    B_1 &= \set{(a_1, 2a_1) : a_1 \in A_1}, \\
    B_2 &= \set{(-2a_2, -3a_2) : a_2 \in A_2}, \\
    B_3 &= \set{(a_3, 0) : a_3 \in A_3}, \\
    B_4 &= \set{(0, a_4) : a_4 \in A_4}.
\end{align*}
Then the tuples $(b_1, b_2, b_3, b_4) \in B_1 \times B_2 \times B_3 \times B_4$ satisfying $b_1 + b_2 + b_3 + b_4 = (0, 0)$ stand in one-to-one correspondence to the tuples $(a_1, a_2, a_3, a_4) \in A_1 \times A_2 \times A_3 \times A_4$ satisfying $a_3 = 2a_2 - a_1$ and $a_4 = 3a_2 - 2a_1$. These are exactly the defining equations for a $4$-AP. Hence, the (2-dimensional) $4$-SUM instance $(B_1, B_2, B_3, B_4)$ is equivalent to the given $4$-AP instance $(A_1, A_2, A_3, A_4)$. A final modification is to reduce from 2-dimensional 4-SUM to the standard integer version by the typical embedding $(x, y) \mapsto x + y \cdot 10n$.
\end{proof}}

\section{4-AP-Hard Identities} \label{sec:4APhard-identities}
In this section, we prove the hardness of several representative identity checking problems, based on the 4-AP hypothesis. Our reduction from 4-AP detection involves two intermediate problems to be defined in the upcoming subsections: Square Detection and T Detection.

\subsection{4-AP-Hardness of Square Detection} \label{subsec:squaredetect}
Our first intermediate problem, Square Detection, was communicated to us by Adam Polak~\cite{Polak_personal_communication}.

\begin{problem}[Square Detection] \label{prob:cornersquare}
Given an $n\times n$ matrix $M$ over $\{0,1\}$, determine if there is a \emph{square}, i.e., $1\le i,j\le n$ and $k > 0$ such that \[M(i,j) = M(i+k,j) = M(i+k,j+k) = M(i,j+k)=1.\]
\end{problem}

\Cref{prob:cornersquare} has a simple $\bigO(n^3)$ time algorithm.
A very similar problem that we call T detection will be defined later in \cref{subsec:tdetect}. For our actual reductions to the identity checking problems, it will be easier to start from the multichromatic version of the problem, which allows us to drop the $k > 0$ condition.

\begin{problem}[Multichromatic Square Detection] \label{prob:colorcornersquare}
Given four $n\times n$ matrices $M_1,M_2,M_3,M_4$ over $\{0,1\}$, determine if there is a \emph{multichromatic square}, i.e., $1\le i,j\le n$ and $k$ such that $M_1(i,j) = M_2(i+k,j) = M_3(i+k,j+k) = M_4(i,j+k)=1$.
\end{problem}

We emphasize that in the multichromatic version $k$ is allowed to be positive, zero, or negative. It is standard to reduce from Square Detection problem to its multichromatic version via color-coding \cite{AlonYZ95}. In the following central lemma we first establish 4-AP-hardness of Multichromatic Square Detection. Then, with some more technical overhead we complete the 4-AP-hardness of Monochromatic Square Detection in \cref{lem:square-hardness}.

\begin{lemma} \label{lem:square-chrom-hardness}
Multichromatic Square Detection cannot be solved in time $O(n^{3-\epsilon})$ (for any constant $\epsilon > 0$), unless the 4-AP Hypothesis fails.
\end{lemma}
\begin{proof}
By \Cref{lem:ap-chromatic} we can reduce from a $4$-chromatic $4$-AP instance $A_1, A_2, A_3, A_4 \subseteq \set{0, \dots, n}$. We assume that all given integers are multiples of $6$ (by multiplying all numbers by $6$ if necessary, preserving arithmetic progressions). Let $N = \ceil{\sqrt n}$. Enumerate all $\delta \in \set{-O(N), \dots, O(N)}$, and construct binary matrices $M_1^\delta, M_2^\delta, M_3^\delta, M_4^\delta$ indexed by $\set{-O(N), \dots, O(N)} \times \set{-O(N), \dots, O(N)}$ as follows (the hidden constants will become clear in the correctness proof):
\begin{alignat*}{9}
    M_1^\delta(i, j) &= 1 &&\qquad\text{iff}\qquad &&N \cdot (6i &&+ 0j) &&+ (0\delta &&- 0i &&- 3j&&) &&\in A_1, \\
    M_2^\delta(i, j) &= 1 &&\qquad\text{iff}\qquad &&N \cdot (4i &&+ 1j) &&+ (1\delta &&- 2i &&- 2j&&) &&\in A_2, \\
    M_3^\delta(i, j) &= 1 &&\qquad\text{iff}\qquad &&N \cdot (2i &&+ 2j) &&+ (2\delta &&- 4i &&- 1j&&) &&\in A_3, \\
    M_4^\delta(i, j) &= 1 &&\qquad\text{iff}\qquad &&N \cdot (0i &&+ 3j) &&+ (3\delta &&- 6i &&- 0j&&) &&\in A_4.
\end{alignat*}
We solve the $O(N)$ Square Detection instances $(M_1^\delta, M_2^\delta, M_3^\delta, M_4^\delta)$, and report the presence of a $4$-AP if and only if at least one of these instances contains a square. If Square Detection is in subcubic time $O(N^{3-\epsilon})$, then the total time is $O(N \cdot N^{3-\epsilon}) = O(n^{2-\epsilon/2})$ which contradicts the $4$-AP Hypothesis.

For the correctness, first assume that $M_1^\delta(i, j) = M_2^\delta(i, j + k) = M_3^\delta(i + k, j) = M_4^\delta(i + k, j + k) = 1$ is a square. Then by definition we have
\begin{alignat*}{9}
    &N \cdot (6i &&+ 0j &&+ 0k) &&+ (0\delta &&- 0i &&- 3j &&- 0k&&) &&\in A_1, \\
    &N \cdot (4i &&+ 1j &&+ 1k) &&+ (1\delta &&- 2i &&- 2j &&- 2k&&) &&\in A_2, \\
    &N \cdot (2i &&+ 2j &&+ 2k) &&+ (2\delta &&- 4i &&- 1j &&- 4k&&) &&\in A_3, \\
    &N \cdot (0i &&+ 3j &&+ 3k) &&+ (3\delta &&- 6i &&- 0j &&- 6k&&) &&\in A_4.
\end{alignat*}
Clearly these four points form a 4-AP (with step-width $N \cdot (-2i + j + k) + (\delta - 2i - j - 2k)$) in the given instance.

Conversely, take any 4-AP $a_1, a_2, a_3, a_4$ in the given instance. We express $a_1$ as $a_1 = N \cdot 6i - 3j$ for some integers $i, j \in \set{0, \dots, O(N)}$, recalling that $a_1 \in \set{0, \dots, N^2}$ is a multiple of $6$. Then pick $k, \delta \in \set{-O(N), \dots, O(N)}$ to express the step-width $a_2 - a_1$ as \makebox{$N \cdot (-2i + j + k) + (\delta - 2i - j - 2k)$}. 
Symmetrically it follows that $M_1^\delta(i, j) = M_2^\delta(i, j + k) = M_3^\delta(i + k, j) = M_4^\delta(i + k, j + k) = 1$.
\end{proof}

The above proof only works for the multichromatic version of Square Detection. In the following we describe how also the hardness of monochromatic Square Detection follows.

\begin{lemma}[Square-Free Matrices] \label{lem:square-free}
For every integer $n \geq 1$, there are binary $n \times n$ matrices $S_1, \dots, S_N$ for some $N = n^{o(1)}$, such that each position $(i, j) \in [n]^2$ 
has a 1-entry in exactly one of the matrices, and such that each matrix $S_\ell$ does not contain a (nontrivial) square. The matrices can be computed in deterministic time $O(n^2)$.
\end{lemma}
\begin{proof}
We will even construct the matrices in such a way that each $S_\ell$ does not even contain a solution to $S_\ell(i, j) = S_\ell(i, j + k) = S_\ell(i + k, j) = 1$ (in combinatorics this is called a \emph{corner}). The following construction is well-known via 3-AP-free sets. Specifically, let $F_1, \dots, F_N$ denote the partition of~$\set{-n, \dots, n}$ into 3-AP-free sets as constructed by \cref{lem:behrend}. Then let $S_\ell$ be the matrix with 1's exactly at the positions $(i, j)$ for which $i - j \in F_\ell$. Clearly each position $(i, j)$ becomes $1$ in at least one matrix. Moreover, any corner $S_\ell(i, j + k) = S_\ell(i + k, j) = 1$ must be trivial, as it otherwise $F_\ell$ would contain the non-trivial 3-AP $i - j - k, i - j, i - j + k$.
\end{proof}

\begin{lemma} \label{lem:square-hardness}
Monochromatic Square Detection cannot be solved in time $O(n^{3-\epsilon})$ (for any constant $\epsilon > 0$), unless the 4-AP Hypothesis fails.
\end{lemma}
\begin{proof}
By \Cref{lem:square-chrom-hardness} it suffices to reduce from a multichromatic Square Detection instance $M_1, M_2, M_3, M_4$. Let $S_1, \dots, S_N$ denote the square-free matrices from the previous \Cref{lem:square-free}. Enumerate all 4-tuples $\ell = (\ell_1, \ell_2, \ell_3, \ell_4) \in [N]^4$, and for each such tuple construct the following $3n \times 3n$ matrix consisting of nine $n \times n$ blocks
\begin{equation*}
    M^\ell =
    \begin{bmatrix}
        M_1 \cap S_{\ell_1} & \mathbf{0} & M_2 \cap S_{\ell_2} \\
        \mathbf{0}          & \mathbf{0} & \mathbf{0}          \\
        M_3 \cap S_{\ell_3} & \mathbf{0} & M_4 \cap S_{\ell_4} \\
    \end{bmatrix};
\end{equation*}
here, by $P \cap Q$ we denote the binary matrix that contains 1-entries only in the positions where both $P$ and $Q$ are also 1. We show that the monochromatic squares across all matrices $M^\ell$ stand in one-to-one correspondence to the multichromatic squares in $M_1, M_2, M_3, M_4$. It follows that if monochromatic Square Detection is in subcubic time $O(n^{3-\epsilon})$, then multichromatic Square Detection is in time $O(n^{3-\epsilon} \cdot N^4) = O(n^{3-\epsilon+o(1)})$, which contradicts the 4-AP Hypothesis by \cref{lem:square-chrom-hardness}.

To see this one-to-one correspondence, first consider any square $M_1(i, j) = M_2(i + k, j) = M_3(i + k, j + k) = M_4(i, j + k) = 1$ (for some possibly zero $k$). We pick the unique 4-tuple $\ell \in [N]^4$ so that $S_{\ell_1}(i, j) = S_{\ell_2}(i + k, j) = S_{\ell_3}(i + k, j + k) = S_{\ell_4}(i, j + k) = 1$. Then by construction $M^\ell$ contains the monochromatic square $M^\ell(i, j) = M^\ell(i + k', j) = M^\ell(i + k, j + k') = M^\ell(i, j + k') = 1$, where $k' = k + 2n \neq 0$.

Conversely, any square $M^\ell(i, j) = M^\ell(i + k', j) = M^\ell(i + k, j + k') = M^\ell(i, j + k') = 1$ must necessarily satisfy $k' > n$ (as each of the nonzero blocks in $M^\ell$ is square-free). But then it follows that $M_1, M_2, M_3, M_4$ contains the multichromatic square specified by $i$, $j$, and $k = k' - 2n$.
\end{proof}

\subsection{4-AP-Hardness of T Detection} \label{subsec:tdetect}
Next, we introduce the (Multichromatic) \emph{T Detection} problem, which is analogous to the Multichromatic Square Detection problem from the previous section.

\begin{problem}[Multichromatic T Detection] \label{prob:colorcornerT}
Given four $n\times n$ matrices $M_1,M_2,M_3,M_4$ over $\{0,1\}$, determine if there is a \emph{multichromatic T}, i.e., $1\le i,j\le n$ and $k$ such that $ M_1(i,j+k) = M_2(i,j-k) = M_3(i+k,j)=M_4(i,j) =1$.
\end{problem}

Again, we allow $k$ to be positive, zero, or negative (i.e., pictorially we allow the T-shape or its vertically flipped version). We show that Multichromatic T Detection is 4-AP-hard as well. The proof is analogous to the proof given in Section~\ref{subsec:squaredetect} for the Square Detection problem. One can similarly also reduce from 4-AP to the monochromatic version of T Detection using similar techniques as before, but we omit this extra step here.

\begin{lemma} \label{lem:t-hardness}
Multichromatic T Detection cannot be solved in time $O(n^{3-\epsilon})$ (for any constant $\epsilon > 0$), unless the 4-AP Hypothesis fails.
\end{lemma}
\begin{proof}
Again we reduce from a multichromatic $4$-AP instance $A_1, A_2, A_3, A_4 \subseteq \set{0, \dots, n}$. Assume that all given integers are multiples of $3$ (as otherwise we can multiply all numbers by $3$, preserving arithmetic progressions). Let $N = \ceil{\sqrt n}$. Enumerate all offsets $\delta \in \set{-O(N), \dots, O(N)}$, and construct binary matrices $M_1^\delta, M_2^\delta, M_3^\delta, M_4^\delta$ indexed by $\set{-O(N), \dots, O(N)} \times \set{-O(N), \dots, O(N)}$ as follows (the hidden constants will become clear in the correctness proof):
\begin{alignat*}{9}
    M_1^\delta(i, j) &= 1 &&\qquad\text{iff}\qquad &&N \cdot (0\delta &&+ 3i) &&+ (0i &&+ 3j&&) &&\in A_1, \\
    M_2^\delta(i, j) &= 1 &&\qquad\text{iff}\qquad &&N \cdot (1\delta &&+ 2i) &&+ (2i &&+ 3j&&) &&\in A_2, \\
    M_3^\delta(i, j) &= 1 &&\qquad\text{iff}\qquad &&N \cdot (2\delta &&+ 1i) &&+ (4i &&+ 3j&&) &&\in A_3, \\
    M_4^\delta(i, j) &= 1 &&\qquad\text{iff}\qquad &&N \cdot (3\delta &&+ 0i) &&+ (6i &&+ 3j&&) &&\in A_4.
\end{alignat*}
We solve these $O(N)$ T Detection instances $(M_1^\delta, M_2^\delta, M_3^\delta, M_4^\delta)$, and report the presence of a $4$-AP if and only if at least one of these instances contains a T. If T Detection is in subcubic time $O(N^{3-\epsilon})$, then the total time is $O(N \cdot N^{3-\epsilon}) = O(n^{2-\epsilon/2})$, which would refute the $4$-AP Hypothesis.

For the correctness, first assume that $M_1^\delta(i, j-k) = M_2^\delta(i, j) = M_3^\delta(i, j + k) = M_4^\delta(i + k, j) = 1$ is a T. Then by definition we have
\begin{alignat*}{9}
    &N \cdot (0\delta &&+ 3i) &&+ (0i &&+ 3j &&- 3k&&) &&\in A_1, \\
    &N \cdot (1\delta &&+ 2i) &&+ (2i &&+ 3j &&+ 0k&&) &&\in A_2, \\
    &N \cdot (2\delta &&+ 1i) &&+ (4i &&+ 3j &&+ 3k&&) &&\in A_3, \\
    &N \cdot (3\delta &&+ 0i) &&+ (6i &&+ 3j &&+ 6k&&) &&\in A_4.
\end{alignat*}
Clearly these four points form a 4-AP (with step-width $N \cdot (\delta -i) + (2i + 3k)$) in the given instance.

Conversely, take any 4-AP $a_1, a_2, a_3, a_4$ in the given instance. We express $a_1$ as \makebox{$a_1 = N \cdot 3i + 3j$} for some integers $i, j \in \set{0, \dots, O(N)}$, recalling that $a_1 \in \set{0, \dots, N^2}$ is a multiple of $3$. Then pick $k, \delta \in \set{-O(N), \dots, O(N)}$ to express the step-width $a_2 - a_1$ as \makebox{$N \cdot (\delta -i) + (2i - 3k)$}. Symmetrically it follows that $M_1^\delta(i, j-k) = M_2^\delta(i, j) = M_3^\delta(i, j + k) = M_4^\delta(i + k, j) = 1$ is a T.
\end{proof}

\subsection{4-AP-Hardness of 6 Identities}
In this section, we show that the 6 central identities for the hardest regime which we outlined before in the overview cannot be verified faster than cubic-time unless the 4-AP-Hypothesis fails. To this end, we will first show that four of these identities are Square Detection-hard, and later that the remaining two are T Detection-hard.

\begin{lemma}
\label{lem:4aphard1234}
   Let $f(a,b,c)$ be any one of the following expressions defined over operations $\oplus_1,\oplus_2,\oplus_3,\oplus_4$:
   \begin{enumerate}
       \item  $f(a,b,c)= \big ((a\oplus_1 b)\,  \oplus_3 \, (a\oplus_2 c)\big ) \oplus_4 c$.
       \label{item:4-ap-instance1}
       \item
       $f(a,b,c)=  (a\oplus_1 b)\, \oplus_4\, \big ((a\oplus_2 c)\oplus_3 b \big) $.
       \label{item:4-ap-instance2}
       \item  $f(a,b,c)=  \Big (\big ((a\oplus_1 b)\oplus_2 c\big )\oplus_3 a\Big ) \oplus_4 b$.
       \label{item:4-ap-instance3}
       \item  $f(a,b,c)=  \Big (\big ((a\oplus_1 b)\oplus_2 c\big )\oplus_3 a\Big ) \oplus_4 c$.
       \label{item:4-ap-instance4}
   \end{enumerate}

   Given four $n\times n$ matrices $M_1,M_2,M_3,M_4$ as an instance for the Multichromatic Square Detection problem (\Cref{prob:colorcornersquare}),
   we can in $\bigO(n^2)$ time construct a set $S$ of size $\bigO(n)$, and Cayley tables for the operations $\oplus_1,\oplus_2,\oplus_3,\oplus_4$, such that there exists a triple $(x_a,x_b,x_c)\in S\times S\times S$ such that $f(x_a,x_b,x_c)\neq \infty$ if and only if there exists a multichromatic square.
\end{lemma}

Our proofs for all the four expressions follow the same strategy. We will present the proof for the first expression in full detail, and then describe how to modify the proof for the remaining three expressions.

\begin{proof}[Proof of \Cref{lem:4aphard1234}, \Cref{item:4-ap-instance1}]
For some absolute constant $C$ ($C=10$ should suffice for all the six proofs), define the universe $S \coloneqq \{-Cn,-Cn+1,\dots,Cn-1,Cn\} \cup \{\infty\}$. Recall a multichromatic square corresponds to integers $i,j,k$ such that
$M_1(i,j) = M_2(i+k,j) = M_3(i+k,j+k) = M_4(i,j+k)=1$. Here, for convenience, we define $M_1(i,j)=M_2(i,j)=M_3(i,j)=M_4(i,j)= 0$ whenever $i$ or $j$ is not in the range $\{1,2,\dots,n\}$.

We first write down the formal definition of the Cayley tables of $\oplus_1,\oplus_2,\oplus_3,\oplus_4$. It will become less mysterious when we verify its correctness soon.
Firstly, define $\infty \oplus_i x = x \oplus_i \infty = \infty$ for all $i\in [4]$ and $x\in S$. The remaining entries are defined as follows: for $x,y\in S\setminus \{\infty\}$,
\begin{align*}
    x \oplus_1 y  &= \begin{cases}
                         y & M_1(x+y, x)=1,\\
                         \infty & \text{otherwise.}
                \end{cases}\\
    x \oplus_2 y  &= \begin{cases}
                         y-x & M_2(y-x, x)=1,\\
                         \infty & \text{otherwise.}
                \end{cases}\\
    x \oplus_3 y  &= \begin{cases}
                         y-x & M_3(y, y-x)=1,\\
                         \infty & \text{otherwise.}
                \end{cases}\\
    x \oplus_4 y  &= \begin{cases}
                         0 & M_4(y-x, x)=1,\\
                         \infty & \text{otherwise.}
                \end{cases}
\end{align*}
Note that the output value of these operations always stay inside $[-Cn,Cn]\subset S$, since otherwise the entries in $M_1,M_2,M_3,M_4$ being checked must have out-of-bound indices and thus cannot be $1$.

Now we verify the correctness of this construction. First, suppose there is a multichromatic square $M_1(i,j) = M_2(i+k,j) = M_3(i+k,j+k) = M_4(i,j+k)=1$. We now verify step-by-step that the assignment
\begin{align}
    x_a &= j, \nonumber \\
    x_b &= i-j, \label{eqn:ijkassignment}\\
    x_c &= i+j+k, \nonumber
\end{align}
satisfies $f(x_a,x_b,x_c)= \big ((x_a\oplus_1 x_b)\,  \oplus_3 \, (x_a\oplus_2 x_c)\big ) \oplus_4 x_c = 0 \neq \infty$:
\newcommand{\gr}[1]{{\color{gray}#1}}
\begin{align}
    x_a \oplus_1 x_b &= j\oplus_1(i-j) = i-j  & &\text{if } M_1(i,j)=1, \nonumber \\
    x_a \oplus_2 x_c &= j\oplus_2 (i+j+k) = i+k  & &\text{if } M_2(i+k,j)=1, \nonumber \\
  (x_a\oplus_1 x_b)\,  \oplus_3 \, (x_a\oplus_2 x_c) &=  (i-j) \oplus_3 (i+k)  =  j+k & &\text{if } M_3(i+k,j+k)=1, \nonumber \\
 \big ((x_a\oplus_1 x_b)  \oplus_3  (x_a\oplus_2 x_c)\big ) \oplus_4 x_c& = (j+k)\oplus_4 (i+j+k) = 0 & &\text{if } M_4(i,j+k)=1. \label{eqn:fourchecksteps}
\end{align}

It remains to verify the other direction. Suppose $x_a,x_b,x_c\in S$ satisfy $f(x_a,x_b,x_c)\neq \infty$.
Then, by our definition that $\infty\oplus_i x=x \oplus_i \infty = \infty$ for all $x\in S$ and $i\in [4]$, we know all subexpressions of $f(a,b,c)$ must evaluate to  values in $S\setminus \{\infty\}$ under the assignment $(a,b,c)=(x_a,x_b,x_c)$. In particular, $x_a,x_b,x_c\in S\setminus \{\infty\}$.
Now, let $(i,j,k)$ be the unique triple satisfying \Cref{eqn:ijkassignment} given $(x_a,x_b,x_c)$.  Then, from the steps above, we conclude that the entries $M_1(i,j), M_2(i+k,j), M_3(i+k,j+k), M_4(i,j+k)$ must all equal $1$, since otherwise some of the operations $\oplus_1,\oplus_2,\oplus_3,\oplus_4$ would return $\infty$. Therefore, there must exist a multichromatic square. This finishes the proof of correctness of our construction.
\end{proof}
For the remaining three expressions, the proof strategy is the same. It only remains to modify the four checking steps in \Cref{eqn:fourchecksteps} and the bijection between $(x_a,x_b,x_c)$ and $(i,j,k)$ (\Cref{eqn:ijkassignment}). An essential property to be satisfied is that, in each line of \Cref{eqn:fourchecksteps}, the return value of the operation and the two matrix indices should all be expressible as linear combinations of the two operands (for example, in the first line of \Cref{eqn:fourchecksteps}, $i-j,i,j$ are all linear combinations of $j$ and $i-j$).
This property allows us to (uniquely) construct the Cayley tables for $\oplus_1,\oplus_2,\oplus_3,\oplus_4$ so that the proof is valid.
Hence, in the following proofs, we no longer have to write down the definitions of $\oplus_1,\oplus_2,\oplus_3,\oplus_4$ (which are less informative anyway), and only need to describe how to modify \Cref{eqn:ijkassignment} and \Cref{eqn:fourchecksteps}.

\begin{proof}[Proof sketch of \Cref{lem:4aphard1234}, \Cref{item:4-ap-instance2}]
Modify the proof of \Cref{lem:4aphard1234}, \Cref{item:4-ap-instance1} as follows: Let the assignment be
\begin{align*}
    x_a &= j,  \\
    x_b &= i-j, \\
    x_c &= i+j+k.
\end{align*}
We check $f(x_a,x_b,x_c)=  (x_a\oplus_1 x_b)\, \oplus_4\, \big ((x_a\oplus_2 x_c)\oplus_3 x_b \big) = 0$ as follows:
\newcommand{\gr}[1]{{\color{gray}#1}}
\begin{align*}
    x_a \oplus_1 x_b &= j \oplus_1 (i-j) = i  & &\text{if } M_1(i,j)=1,  \\
     x_a \oplus_2 x_c  &=
 j\oplus_2 (i+j+k) = i+k  & &\text{if } M_2(i+k,j)=1,  \\
 \big ((x_a\oplus_2 x_c)\oplus_3 x_b \big)
 &=  (i+k) \oplus_3 (i-j)  =  j+k & &\text{if } M_3(i+k,j+k)=1,  \\
 (x_a\oplus_1 x_b)\, \oplus_4\, \big ((x_a\oplus_2 x_c)\oplus_3 x_b \big)  & = i\oplus_4 (j+k) = 0 & &\text{if } M_4(i,j+k)=1. \qedhere
\end{align*}
\end{proof}

\begin{proof}[Proof sketch of \Cref{lem:4aphard1234}, \Cref{item:4-ap-instance3}]
Modify the proof of \Cref{lem:4aphard1234}, \Cref{item:4-ap-instance1} as follows: Let the assignment be
\begin{align*}
    x_a &= i-j,  \\
    x_b &= i, \\
    x_c &= i+j+k.
\end{align*}
We check $f(x_a,x_b,x_c)=  \Big (\big ((x_a\oplus_1 x_b)\oplus_2 x_c\big )\oplus_3 x_a\Big ) \oplus_4 x_b = 0$ as follows:
\newcommand{\gr}[1]{{\color{gray}#1}}
\begin{align*}
    x_a \oplus_1 x_b &= (i-j)\oplus_1 i = j  & &\text{if } M_1(i,j)=1,  \\
    \big ((x_a\oplus_1 x_b)\oplus_2 x_c\big ) &=
 j\oplus_2 (i+j+k) = i+k  & &\text{if } M_2(i+k,j)=1,  \\
 \Big (\big ((x_a\oplus_1 x_b)\oplus_2 x_c\big )\oplus_3 x_a\Big ) &=  (i+k) \oplus_3 (i-j)  =  j+k & &\text{if } M_3(i+k,j+k)=1,  \\
 \Big (\big ((x_a\oplus_1 x_b)\oplus_2 x_c\big )\oplus_3 x_a\Big ) \oplus_4 x_b & = (j+k)\oplus_4 i = 0 & &\text{if } M_4(i,j+k)=1. \qedhere
\end{align*}
\end{proof}

\begin{proof}[Proof sketch of \Cref{lem:4aphard1234}, \Cref{item:4-ap-instance4}]
Modify the proof of \Cref{lem:4aphard1234}, \Cref{item:4-ap-instance1} as follows: Let the assignment be
\begin{align*}
    x_a &= i-j,  \\
    x_b &= i, \\
    x_c &= i+j+k.
\end{align*}
We check $f(x_a,x_b,x_c)=  \Big (\big ((x_a\oplus_1 x_b)\oplus_2 x_c\big )\oplus_3 x_a\Big ) \oplus_4 x_c = 0$ as follows:
\newcommand{\gr}[1]{{\color{gray}#1}}
\begin{align*}
    x_a \oplus_1 x_b &= (i-j)\oplus_1 i = j  & &\text{if } M_1(i,j)=1,  \\
    \big ((x_a\oplus_1 x_b)\oplus_2 x_c\big ) &=
 j\oplus_2 (i+j+k) = i+k  & &\text{if } M_2(i+k,j)=1,  \\
 \Big (\big ((x_a\oplus_1 x_b)\oplus_2 x_c\big )\oplus_3 x_a\Big ) &=  (i+k) \oplus_3 (i-j)  =  j+k & &\text{if } M_3(i+k,j+k)=1,  \\
 \Big (\big ((x_a\oplus_1 x_b)\oplus_2 x_c\big )\oplus_3 x_a\Big ) \oplus_4 x_c & = (j+k)\oplus_4 (i+j+k) = 0 & &\text{if } M_4(i,j+k)=1. \qedhere
\end{align*}
\end{proof}

Curiously, for the remaining two identities, we do not see how to reduce from Square Detection. Instead, we will reduce from T Detection.

\begin{lemma}
\label{lem:4aphard5}
The same statement as in \Cref{lem:4aphard1234}, but with Multichromatic Square Detection replaced by Multichromatic T Detection (\Cref{prob:colorcornerT}), holds for the following expression:
\begin{enumerate}
{\setcounter{enumi}{\numexpr4\relax}}
    \item $f(a,b,c)=  \big ((a\oplus_1 b)\oplus_3 (a\oplus_2 c)\big )\oplus_4 a $. \label{item:4-ap-instance5}
    \item $f(a,b,c) = (a\oplus_1 b) \oplus_4 \Big (a\oplus_3 \big (c\oplus_2 (a\oplus_1 b)\big )\Big )$
    \label{item:4-ap-instance6}
\end{enumerate}
\end{lemma}

\begin{proof}[Proof sketch of \Cref{lem:4aphard5}, \Cref{item:4-ap-instance5}]
The proof structure is still the same as \Cref{lem:4aphard1234}, \Cref{item:4-ap-instance1}, except that we need to use multichromatic T detection (\Cref{prob:colorcornerT}) instead.

Let the assignment be
\begin{align*}
    x_a &= i,  \\
    x_b &= j+k, \\
    x_c &= j-k.
\end{align*}
We check $f(x_a,x_b,x_c)=
  \big ((x_a\oplus_1 x_b)\oplus_3 (x_a\oplus_2 x_c)\big )\oplus_4 x_a =0$ as follows:
\newcommand{\gr}[1]{{\color{gray}#1}}
\begin{align*}
    x_a \oplus_1 x_b &= i\oplus_1 (j+k) = i+j+k  & &\text{if } M_1(i,j+k)=1,  \\
    x_a\oplus_2 x_c &=
 i\oplus_2 (j-k) = i-j+k  & &\text{if } M_2(i,j-k)=1,  \\
 \big ((x_a\oplus_1 x_b)\oplus_3 (x_a\oplus_2 x_c)\big )
&=  (i+j+k) \oplus_3 (i-j+k)  =  j & &\text{if } M_3(i+k,j)=1,  \\
  \big ((x_a\oplus_1 x_b)\oplus_3 (x_a\oplus_2 x_c)\big )\oplus_4 x_a
& = j\oplus_4 i = 0 & &\text{if } M_4(i,j)=1. \qedhere
\end{align*}
\end{proof}

\begin{proof}[Proof sketch of \Cref{lem:4aphard5}, \Cref{item:4-ap-instance6}]
Again, the proof structure is the same as \Cref{lem:4aphard1234}, \Cref{item:4-ap-instance1} using multichromatic T detection (\Cref{prob:colorcornerT}).

Let the assignment be
\begin{align*}
    x_a &= j,  \\
    x_b &= i+j, \\
    x_c &= i-j+k.
\end{align*}
We check $f(x_a,x_b,x_c)=
(x_a\oplus_1 x_b) \oplus_4 \Big (x_a\oplus_3 \big (x_c\oplus_2 (x_a\oplus_1 x_b)\big )\Big )$ as follows:
\newcommand{\gr}[1]{{\color{gray}#1}}
\begin{align*}
    x_a \oplus_1 x_b &= j\oplus_1 (i+j) = i  & &\text{if } M_4(i,j)=1,  \\
   x_c\oplus_2 (x_a\oplus_1 x_b) &=
 (i-j+k)\oplus_2 i = i-j+k  & &\text{if } M_2(i,j-k)=1,  \\
 x_a\oplus_3 \big (x_c\oplus_2 (x_a\oplus_1 x_b)\big )
&=  j \oplus_3 (i-j+k)  =  i+j+k & &\text{if } M_3(i+k,j)=1,  \\
(x_a\oplus_1 x_b) \oplus_4 \Big (x_a\oplus_3 \big (x_c\oplus_2 (x_a\oplus_1 x_b)\big )\Big )
& = i\oplus_4 (i+j+k) = 0 & &\text{if } M_1(i,j+k)=1. \qedhere
\end{align*}
\end{proof}

\section{Algebraic Identity Checking Beyond Distributivity} \label{sec:trichotomy}
In this section we aim to classify the complexity of all natural algebraic identities on three variables, namely all algebraic identities that can be written as $f(a,b,c) = g(a,b,c)$ such that neither $f$ is a subexpression of $g$ nor $g$ of $f$.
Our strategy is to first consider an intermediate problem that we call \emph{Constant Term Identity} problem, that corresponds to the special case when $g(a,b,c)$ is a constant function and show that we can fully classify this intermediate problem.
We then show how we can use this special case to show a more general classification.
In the next section we introduce the Constant Term Identity problem and establish complexity of all possible expressions $f$ on three variables.

\subsection{Complexity of Constant Term Identity Problem}
We first recall the Constant Term Identity problem.
\begin{definition}[Constant Term Identity problem]
    Consider an algebraic expression $f(a,b,c)$ over operations $\odot_1,\dots, \odot_p$. Given a set $S$, Cayley tables for $\odot_1,\dots, \odot_p:S\times S\to S$ the \emph{Constant Term Identity problem} is to decide if for all $a,a',b,b',c,c'\in S$ it holds that $f(a,b,c) = f(a',b',c')$.
\end{definition}

Throughout this section we distinguish a single element of $S$ that is the result of evaluation of $f(a,b,c)$ for some triple $a,b,c\in S$ and we denote this element as $\infty$. Then the goal is to verify if $f(a,b,c)=\infty$ for all $a,b,c\in S$ and often we will use that statement of the Constant Term Identity problem. We note that we do not assume any properties of the Cayley tables regarding $\infty$, i.e., it may be the case that $\infty\odot_i a\ne \infty$ for some $a\in S$. However, for some auxiliary operations that we introduce, we will treat $\infty$ differently than other elements of $S$.

The goal of this section is to prove the following trichotomy theorem.
\begin{theorem}[Trichotomy of the Constant Term Identity Problem]\label{thm:trichotomy}
    Let $f(a,b,c)$ be an algebraic expression over operations $\odot_1,\dots, \odot_p$.
    Then the following holds.
    \begin{description}
        \item[Quadratic-Time Regime]\label{regime:easy-vs-triangle-hard} If there exist algebraic expressions on two variables $G(a,b),H(a,b)$ such that $f$ can be written as $f(a,b,c) \equiv H(G(a,b),c)$ (for some permutation of $a,b,c$), then we can solve the Constant Term Identity Problem on $f$ in time $\bigO(n^2)$.

        \item[Triangle-Time Regime]\label{regime:triangle-vs-ap-hard}
        If $f$ cannot be written as $f(a,b,c) \equiv H(G(a,b),c)$ (for any permutation of $a,b,c$), for no $\varepsilon>0$ is there an algorithm solving Constant Identity problem on $f$ in $\bigO(n^{\omega-\varepsilon})$, unless Triangle Hypothesis fails.

        If there exist algebraic expressions on two variables $G(a,b)$, $H(a,b)$, $I(a,b)$, $J(a,b)$, such that $f$ can be written as $f(a,b,c) \equiv J(I(H(a,b),c), G(a,b))$ (for some permutation of $a,b,c$), then there is a randomized algorithm solving Constant Term Identity problem on $f$ in time $\bigO(n^{\omega})$, that succeeds with high probability.

        \item[Cubic-Time Regime]\label{regime:apsp-hard}
        If $f$ cannot be written as $f(a,b,c) \equiv J(I(H(a,b),c), G(a,b))$ (for any permutation of $a,b,c$), 
        for no $\varepsilon>0$ is there an algorithm solving this problem in $\bigO(n^{3-\varepsilon})$, unless $4$-AP Hypothesis fails. 

         There is a (trivial) algorithm solving Constant Term Identity problem on $f$ in time $\bigO(n^3)$, for all expressions $f(a,b,c)$ on three variables.
    \end{description}
\end{theorem}
Intuitively, this theorem tells us that we have three regimes in which solving Constant Term Identity for $f(a,b,c)$ can belong to: quadratic, triangle hard and cubic.
Moreover, we provide a precise separation between the three regimes and show that the claimed complexity is tight.

Towards proving this theorem, we construct the algorithms and the conditional lower bounds separately.
All of the algorithms that we construct are essentially an extension of the framework used in the algorithm for verifying Distributivity that we presented in Section~\ref{sec:distributivity_alg}.
The more interesting part is constructing the conditional lower bounds.
On a high level, we obtain our lower bounds via the following 4 steps that we explain later in detail:
\begin{enumerate}
    \item\label{item:embedding-gadget} Embedding subexpressions into our operations
    \item\label{item:chromaticity} Introducing chromaticity
    \item\label{item:similarity} Exploring structural similarity
    \item\label{item:collapsing} Collapsing subexpressions
\end{enumerate}
We now briefly give some intuition for each of the four steps and then dive into the details in the dedicated subsections.

\noindent\ref{item:embedding-gadget}. We first construct a gadget that we call \emph{Subexpression Embedding Gadget}. Intuitively, we embed to the elements of the input set $S$ all of the subexpressions of the expression $f$, and deterministically modify the set of operations as to satisfy the following conditions:
\begin{enumerate}[label=(\roman*)]
    \item Consistency: the embedded instance will be equivalent to the original instance.
    \item Small size: the size of the embedded instance is (up to a constant factor) the same as the size of the original instance.
    \item Bookkeeping: each performed operation in the parse tree of $f$ additionally keeps the history of all previously performed operations.
\end{enumerate}
\noindent\ref{item:chromaticity}. We then consider a generalization of our problem, namely the \emph{Multichromatic Constant Term Identity} problem, where on top of the Cayley tables of relevant operations, our input also consists of subsets $A,B,C$ of our universe $S$, and we want to decide if $f(x,y,z)$ evaluates to a fixed constant for all $x\in A, y\in B, z\in C$.
We can then use the Subexpression Embedding Gadget to show that these two problems are equivalent in a sense that a fast algorithm for one implies a fast algorithm for the other. Hence, we conclude that it is sufficient to show lower bounds for the \emph{multichromatic} version of the problem.

\noindent\ref{item:similarity}. We then argue that all of the expression that fail to satisfy one of the conditions of the theorem share some structural similarities in terms of the subexpressions that they must/cannot contain.

\noindent \ref{item:collapsing}. Finally, we exploit this notion of similarity, together with the earlier constructions to show that we can \emph{collapse} all of the expressions that fail to satisfy some condition of the theorem into only a handful of simple \emph{representative} expressions, in a sense that a fast algorithm for solving the problem on any of the expressions that fall into this category would imply a fast algorithm for solving the problem on one of the representative expressions.
We can then afford to use the ad-hoc constructions for each of the representative expressions to argue conditional hardness by reducing from the appropriate hard problem.

The rest of this section will be dedicated to expanding on each of these four steps and proving the claimed lower bounds.
First we introduce the notion of decomposing an expression into another expression:

\begin{definition}[Decomposition]
    Let $f(\overline v)$ and $H(\overline v)$ be algebraic expressions. We say $f$ is \emph{decomposable} into $H$ if there exists an algebraic expression $F(x)$ over a single variable, such that $f(\overline v) \equiv F(H(\overline v))$.
\end{definition}

For example, if $f(a,b)\equiv (a+b)+(a+b)\cdot(a+b)$, we have $H(a,b)\equiv (a+b), F(x)\equiv x+x\cdot x$ and $f(a,b)\equiv F(H(a,b))$.
Observe that when $f(\bar v) = F(H(\bar v))$, the tree of $f$ is exactly the tree of $F$ with all its leaves replaced by copies of the tree of $H$.
This leads to the following property of decomposable expressions:

\begin{observation}\label{obs:recursive-decomposability}
    Let $f(\overline v)$ be an algebraic expression decomposable into $H$. Then each subexpression $f'(\overline{v})$ of $f(\bar v)$ that satisfies ${\rm depth}(f')\ge {\rm depth}(H)$ is also decomposable into $H$.
\end{observation}

If two expressions are decomposable into one expression, we say that they are \emph{similar}:
\begin{definition}[Similarity Relation $\sim$]
    Let $f(a,b,c),g(a,b,c)$ be two algebraic expressions. If there exists an algebraic expression $H$ such that both $f$ and $g$ are decomposable into $H$, we say that $f$ and $g$ are \emph{similar}, and write $f\sim g$.
\end{definition}
For example, the expressions $f(a,b)\equiv (a+b)+(a+b)\cdot(a+b)$ and $g(a,b)\equiv (a+b)\cdot(a+b)\cdot(a+b)$ are similar, as both are decomposable into $(a+b)$.
In fact, similarity is an equivalence relation.
\begin{lemma}[Equivalence of $\sim$]\label{lemma:equivalence-of-similarity}
    Similarity is an equivalence relation.
\end{lemma}
\begin{proof}
    Reflexivity and symmetry trivially hold, so we only prove transitivity.
    Let $f(\overline v),g(\overline v),h(\overline v)$ be algebraic operations such that $f\sim g$ and $g\sim h$, where $\overline{v}$ is a vector of variables.
    Then by definition, we can write $f(\overline v) \equiv F(K(\overline v))$, $g(\overline v)\equiv G(K(\overline v)) \equiv G'(K'(\overline v))$, and $h(\overline v) \equiv H(K'(\overline v))$.
    Without loss of generality, assume that ${\rm depth}(K) \geq {\rm depth}(K')$.
    As $K$ is a subexpression of $g$ and $g$ is decomposable to $K'$, by Observation \ref{obs:recursive-decomposability} we get that $K$ is decomposable into $K'$. Hence, we can write $K(\overline v) \equiv I(K'(\overline v))$ for some expression $I(x)$.
    Hence $f(\bar v)\equiv F(K(\overline v)) \equiv F(I(K'(\overline v)))$, proving that $f\sim h$.
\end{proof}
Before proving our main theorem, we need to introduce another gadget that will be useful for further analysis.
\paragraph{Subexpression Embedding}
In the next paragraphs we consider the following subexpression embedding gadget which, intuitively, embeds subexpressions of $f$ into elements of $S$, while preserving the special symbol $\infty$.
\newcommand{\expr}[1]{#1}
\begin{definition}[Subexpression embedding]\label{def:subexpression_embedding}
    Let $S$ be a set of elements such that $\infty\in S$ and $f(\overline v)$ be an algebraic expression over operations $\oplus_1,\dots, \oplus_p$ on $S$.
Let $\mathcal T$ be the set of all subexpressions of $f$.
We define the set $\tilde S$ as $\tilde S := \left(S\times \mathcal T\right) \cup \{\infty\} $ and operations $\tilde \oplus_1,\dots, \tilde \oplus_p: \tilde S\times \tilde S\to \tilde S$ as follows.
For $i\in [p]$, elements $x,y\in S$ and subexpressions $t_1,t_2\in \mathcal T$ we define:
\begin{equation*}
    (x,t_1) \tilde \oplus_i (y,t_2) = \begin{cases}
        \left(x\oplus_i y, t_1\oplus_i t_2\right) & \text{if $x\oplus_i y\neq \infty$ and $\expr{t_1\oplus_i t_2}\in \mathcal T$}\\
        \infty & \text{otherwise}
    \end{cases}
\end{equation*}
and $\infty \tilde \oplus_i \tilde x = \tilde x \tilde \oplus_i \infty = \infty$ for any $\tilde x \in \tilde S$.
Finally, we denote by $\tilde{f}(a,b,c)$ the expression corresponding to $f$ where each operation $\oplus_i$ is replaced by $\tilde{ \oplus}_i$.
\end{definition}

Recall the notation used here: $x\oplus_i y$ represents the evaluation of the operation and corresponds to an element of $S$, while $t_1\oplus_it_2$ is a symbolic expression, and will be an element of $\mathcal T$.
Observe that regardless of Cayley tables for operations $\oplus_1,\dots, \oplus_p$, the expression $\tilde f((x,t_1),(y,t_2),(z,t_3))$ will evaluate to $\infty$ for all triples $x,y,z \in S$, $t_1,t_2,t_3\in \mathcal T$, unless $t_1 \equiv \expr{a}$, $t_2 \equiv \expr{b}$, $t_3 \equiv \expr{c}$.
Indeed, otherwise at some node of the evaluation tree we will obtain an expression $t\not\in\mathcal T$ for which we set $\infty$ that will be pushed to the root of the tree and be the final result of the evaluation. On the other hand, we remark that if $t_1 \equiv \expr a, t_2 \equiv \expr b, t_3 \equiv \expr c$, then for each $x,y,z\in S$ we have that $\tilde f ((x,t_1),(y,t_2),(z,t_3))$ will evaluate to either $(f(x,y,z),f)$, or to $\infty$, depending on whether $f(x,y,z)$ evaluates to $\infty$.
These two properties combined give us:
\begin{observation}\label{obs:chromaticity}
Consider an expression $f(a,b,c)$ over $S$ and let $\tilde f$ be defined in subexpression embedding (Definition~\ref{def:subexpression_embedding}) for $f$ and $S$.
For all subexpressions $t_1,t_2,t_3$ of $f$ and all $x,y,z\in S$ we have:
\[\tilde f((x,t_1),(y,t_2),(z,t_3)) =
\begin{cases}
\left(f(x,y,z) , f\right) &\text{if } f(x,y,z) \ne \infty, t_1\equiv\expr{a}, t_2\equiv\expr{b}, t_3\equiv\expr{c}\\
\infty  & \text{otherwise}
\end{cases}\]
\end{observation}
Recall that we assume that expression $f$ is of constant size, so we have a constant number of its subexpressions.
Then the observation above shows that solving the Constant Term Identity problem on $f$ is equivalent to solving Constant Term Identity problem on $\tilde f$, with the number of elements larger by a constant factor than the original instance:
\begin{corollary}\label{cor:subexpr_embedding}
Consider an expression $f(a,b,c)$ over $S$ and let $\tilde f$ and $\tilde S$ be defined in subexpression embedding (Definition~\ref{def:subexpression_embedding}) for $f$ and $S$.
Then $|\tilde S|= \bigO(|S|)$; and $f(x,y,z) = \infty$ for all $x,y,z \in S$ if and only if $\tilde f(\tilde x, \tilde y, \tilde z) = \infty$ for all $\tilde x, \tilde y, \tilde z \in \tilde S$.
\end{corollary}
Informally, this construction gives us a way to decide if $f(x,y,z) = \infty$ for all $x,y,z\in S$, while simultaneously keeping track of where we are in the parse tree of $f$ at each step.
Using this notion we will be able to prove technical lemmas that will simplify the landscape of the identities that we need to check.
As a first case in point, by leveraging the subexpression embedding construction, we show that deciding if $f(x,y,z)=\infty$ for each $x,y,z\in S$ is at least as hard as deciding if $g(x,y,z) = \infty$ for each subexpression $g$ of $f$.
\begin{lemma}[Subexpression Monotonicity Lemma]\label{lemma:subexpression-monotonicity}
    Let $f(\overline v)$ be an expression, let $g(\overline v)$ be a subexpression of $f$ and their variables take values from set $S$ of $n$ elements. If there is an algorithm solving the Constant Term Identity problem for $f(\overline v)$ in time $T(n)$, then there is an algorithm solving the Constant Term Identity problem on $g(\overline v)$ in time $\bigO(T(n))$.
\end{lemma}
\begin{proof}
    Assume $f$ is defined over operations $\oplus_1,\dots, \oplus_p$ and $f\equiv L\oplus R$ for some subexpressions $L,R$ and operation $\oplus\in \{\oplus_1,\dots, \oplus_p\}$.
    We remark that it is sufficient to prove that we can solve the Constant Term Identity problem for $L$ in $\bigO(T(n))$ using $T(n)$ algorithm for $f$ (reasoning for $R$ is symmetric), as the claim of the lemma will then follow by induction.
    Also, assume that ${\rm depth}(L)\ge 2$, since otherwise it is trivial to decide if for all $\bar v$, $L(\bar v)=\infty$ in linear time.
    Let $\tilde f$ and $\tilde S$ be defined in subexpression embedding (Definition~\ref{def:subexpression_embedding}) for $f$ and $S$.
    Further, define a new sequence of operations $\odot_1,\dots, \odot_p$ as follows. Let $i\in[p]$ and for each pair of elements $x,y\in S\setminus \{\infty\}$ and subexpressions $t_1, t_2$ of $f$ set:
    \[
    (x,t_1) \odot_i (y,t_2) = \begin{cases}
        (x,t_1) \tilde\oplus_i (y,t_2) & \text{if }t_1\oplus_i t_2\not \equiv f \\
        (x,f) & \text{otherwise}
    \end{cases}
    \]
    Further, for any $\tilde x\in \tilde S$, let
    \[
    \infty \odot_i \tilde x = \infty
    \]
    Finally, for any element $x\in S$ and expression $t_1$, let
    \[
    (x,t_1) \odot_i \infty = \begin{cases}
        \infty & \text{if $t_1\ne L$}\\
        (x, f) & \text{otherwise}
    \end{cases}
    \]
    Let $f^*$ be the expression obtained by replacing each $\oplus_i$ operation in $f$ by $\odot_i$. As $f^*$ is the same expression as $f$, but with differently labeled operations, we can use the algorithm for solving Constant Term Identity problem on~$f$ in time $T(n)$ to solve Constant Term Identity problem on~$f^*$ in time $T(\bigO(n))=\bigO(T(n))$ as $T(n)=\bigO(n^{|v|})$.
    To conclude our proof, we show that there is a triple $\tilde x, \tilde y, \tilde z\in \tilde S$ such that $f^*(\tilde x, \tilde y, \tilde z) \ne \infty$ if and only if there is a triple $x,y,z \in S$ such that $L(x,y,z) \ne \infty$.

    Assume first that there is a triple $x,y,z \in S$ such that $L(x,y,z) \ne \infty$.
    Let $\odot_j$ be the operation in the root of $f^*$.
    Then by definition of $\odot_h$ we have $f^*((x,a),(y,b),(z,c)) = (L(x,y,z), f)\ne \infty$.
    Assume now that there is no such triple, i.e., for each $x,y,z\in s$ we have $L(x,y,z) = \infty$.
    Assume for contradiction that there exists a triple $\tilde x, \tilde y, \tilde z \in \tilde S$ such that $f^*(\tilde x, \tilde y, \tilde z)\ne \infty$.
    Denote by $\tilde L, L^*$ the subexpressions of $\tilde f$ and $f^*$ respectively that correspond to $L$.
    By Observation \ref{obs:chromaticity} and $L(x,y,z)=\infty$, we know that $\tilde L (\tilde x, \tilde y, \tilde z) = \infty$.
    On the other hand, $L^*(\tilde x, \tilde y, \tilde z)\ne \infty$ as otherwise we would have $f^*(\tilde x, \tilde y, \tilde z)=\infty \odot_j R^*(\tilde x, \tilde y, \tilde z) = \infty$ which contradicts that $f^*(\tilde x, \tilde y, \tilde z)\ne \infty$.
    By definition of $\odot$ we have the following claim:
    \begin{claim}
        If $\tilde L (\tilde x, \tilde y, \tilde z) = \infty$ and $L^*(\tilde x, \tilde y, \tilde z)\ne \infty$,  then $L^*(\tilde x, \tilde y, \tilde z) = (v,f)$ for some $v\in S$.
    \end{claim}
    \begin{subproof}
    We show by induction of the size of the subexpression $e$ of $L$ that $e^*(\tilde x, \tilde y, \tilde z)\in\{\tilde e(\tilde x, \tilde y, \tilde z), \infty\}\cup \{(v,f) : v\in S\}$.
    Indeed, this is the case because if either of the subtrees of $e^*$ evaluates to $\infty$ or some $(v,f)$ then $e^*(\tilde x, \tilde y, \tilde z)$ also evaluates to either $\infty$ or $(v',f)$ as $f\oplus_i t\not\in \mathcal{T}$ for all $t$.
    Next, $\tilde u \odot_i \tilde w \ne \tilde u \tilde\oplus_i \tilde w$ holds only in two cases: either when $\tilde u = (x,L)$ and $\tilde w=\infty$ or when $\tilde u = (x,t_1),\tilde w= (y,t_2)$ and $t_1\oplus_i t_2\equiv f$.
    In both cases $\tilde u \odot_i \tilde w = (x,f)$ instead of $\infty$ or $(x\oplus_i y,f)$, respectively, so the inductive step follows.

    Finally, as $L^*(\tilde x, \tilde y, \tilde z)\not\in \{\infty,\tilde L(\tilde x, \tilde y, \tilde z)\}$ so $L^*(\tilde x, \tilde y, \tilde z)=(v,f)$ for some $v\in S$.

    \end{subproof}
    Recall that $f \equiv L \oplus R$, which implies that for some $\tilde u\in \tilde S$ we have that $f^*(\tilde x, \tilde y, \tilde z) = (v,f) \odot_j \tilde u$.
    Now there are two cases, either $\tilde u = \infty$ in which case this clearly evaluates to $\infty$. Otherwise, $\tilde u = (u,t_2)$ for some $u\in S$ and subexpression $t_2$. But clearly, for each subexpression $t_2$, $f\oplus_j t_2$ is not a subexpression of $f$, hence $(v,f) \odot \tilde u =\infty$ in this case as well.
    Both cases lead to contradiction with the assumption that $f^*(\tilde x, \tilde y, \tilde z) \ne \infty$ and the claim follows.
\end{proof}
\paragraph{Multichromatic Constant Term Identity Problem} So far we only considered the version of problem in which our input consists of a single set $S$, over which we define all our operations and then want to decide if for all $x,y,z\in S$ it holds that $f(x,y,z) = \infty$.
However, often it is helpful to consider a more general version of the problem where, in addition to our universe, we are restricting the sets in which variables $x,y,z$ can live. More formally we define the following problem:
\begin{definition}[Multichromatic Constant Term Identity]
Let $f(a,b,c)$ be a fixed algebraic expression over operations $\oplus_1,\dots, \oplus_p$. Given a set $S$ and subsets $A,B,C\subseteq S$ with Cayley tables of operations $\oplus_1,\dots, \oplus_p$ defined on $S$, the \emph{Multichromatic Constant Term Identity} problem is to decide if for all triples $a\in A, b\in B, c\in C$ it holds that $f(a,b,c) = \infty$.
\end{definition}

Clearly the Constant Term Identity problem is a special case of this problem where $A=B=C=S$.
In particular, if we can solve the Multichromatic Constant Term Identity problem on $f$ in time $T(|S|)$, then we can also solve the Constant Term Identity problem on $f$ in $\bigO(T(|S|))$.
However, using the idea of subexpression embedding gadget, we can actually construct the reduction in the other direction as well, thus showing the equivalence between these two problems.
This will be particularly useful for showing hardness of Constant Term Identity problem on some family of expressions.
We also note that without loss of generality we can assume that the sets $A,B,C$ are pairwise disjoint, since otherwise we can copy the set $S$ three times and consider each set $A,B,C$ to be the corresponding subset in its respective copy.
We will use this fact to construct the Cayley tables for some of the technical lemmas later on.
\begin{lemma}
    Let $f(a,b,c)$ be a fixed algebraic expression over elements from set $S$ of $n$ elements. If there exists an algorithm solving Constant Term Identity on $f$ in time $T(n)$, then Multichromatic Constant Term Identity can be solved in time $\bigO(T(n))$.
\end{lemma}
\begin{proof}
    If $f$ depends only on a proper subset of $\{a,b,c\}$, this claim is trivial, since Multichromatic Constant Term Identity can be solved in time $\bigO(n^2)$, linear in the size of input. Hence, we assume that $f$ depends on all three variables.
    Suppose $f$ operates on $S$ and let $\tilde f$ and $\tilde S$ be defined in subexpression embedding (Definition~\ref{def:subexpression_embedding}) for $f$ and $S$.
    Additionally, for each $x\in S\setminus A$, replace the row/column corresponding to the element $(x,\expr a)$ in each Cayley table of $\tilde \oplus_i$ to all-$\infty$ row/column, that is set $(x,\expr a)\tilde \oplus_i \tilde y =\tilde y \tilde \oplus_i  (x,\expr a) = \infty$ for each $\tilde y\in \tilde S$.
    Repeat the analogous construction for each row/column in each Cayley table corresponding to $(y,\expr b)$ for $y\in S\setminus B$, and $(z,\expr c)$ for $z\in S\setminus C$.
    \begin{claim}
        For each triple $(x,y,z)\not \in A\times B\times C$, and for each triple of subexpressions $t_1,t_2,t_3$ of $f$, we have that $\tilde f((x,t_1),(y,t_2),(z,t_3)) = \infty$.
    \end{claim}
    \begin{subproof}
        Firstly, if $(t_1,t_2,t_3)\not\equiv (\expr a,\expr b,\expr c)$ the claim follows as, at the latest, in the root we will obtain an expression that is not in $\mathcal{T}$, similarly as in Observation \ref{obs:chromaticity}. So we only need to consider the case when $(t_1,t_2,t_3) \equiv (\expr a,\expr b,\expr c)$.
        Assume without loss of generality that $x\not\in A$. By assumption that $f$ depends on all three variables, there exists a subexpression $t(a,b,c)$ of $f$ such that $t(a,b,c)\equiv \expr a \oplus_i t'(a,b,c)$ for some other subexpression $t'$ and some $i\in [p]$. The case when $t(a,b,c)\equiv t'(a,b,c) \oplus_i \expr a$ is symmetric.
        Then the corresponding subexpression $\tilde t$ of $\tilde f$ satisfies $\tilde t((x,\expr a),(y,\expr b),(z,\expr c)) = (x,\expr a) \tilde \oplus_i \tilde r$ for some $\tilde r\in \tilde S$.
        However, since $x\not \in A$, this implies that $\tilde t((x,\expr a),(y,\expr b),(z,\expr c)) = \infty$.
        Moreover, since $\infty \tilde \oplus_i \tilde r = \tilde r \tilde\oplus_i \infty = \infty$ by definition of subexpression embedding gadget, the desired claim follows directly.
    \end{subproof}
    Now it remains to show that for all $(x,y,z)\in A\times B\times C$ we have that $\tilde f((x,\expr a),(y,\expr b),(z,\expr c)) = \infty$ if and only if $f(x,y,z) = \infty$.
    However, this follows immediately from Observation \ref{obs:chromaticity}, as the all-$\infty$ row/column change in Cayley tables of $\tilde \oplus_i$ does not affect $(x,a),(y,b)$ or $(z,c)$ for $x\in A,y\in B$ and $z\in C$ respectively.
\end{proof}
Recall that two expressions are similar if there exists a common expression that they both decompose to.
We now prove that if we have two non-similar expressions over the same set of operations, we can set up the corresponding Cayley tables to simulate the case where one of the expressions has a brand new operation (that does not appear in any of the two expressions) in the root.
This will help us collapse the non-similar expressions to significantly reduce the number of expressions we need to prove lower bounds for.
\begin{lemma}\label{lemma:operations}
    Let $f(a,b,c), g(a,b,c)$ be expressions over operations $\oplus_1,\dots, \oplus_p$ and assume $f(a,b,c) \equiv L \oplus_i R$ for some $i\in [p]$. Let $f'(a,b,c) \equiv L \oplus' R$ for some operation $\oplus' \not\in \{\oplus_1,\dots, \oplus_p\}$.
    If $f\not\sim g$ and ${\rm depth}(f)\ge  {\rm depth}(g)$, then given a set $S$ together with its subsets $A,B,C$ and Cayley tables for $\oplus_1,\dots, \oplus_p, \oplus'$ we can in time $\bigO(n^2)$ construct the set $S'\supseteq S$, and Cayley tables for operations $\overline \oplus_1,\dots,\overline \oplus_p$ and a unary operation $q:S'\to S$ such that if $\overline f, \overline g$ denote the expressions with $\oplus_i$ replaced with $\overline \oplus_i$, the following is satisfied.
    \begin{itemize}
        \item For each $x\in A,y\in B,z\in C$, the value of $g(x,y,z) = q(\overline g(x,y,z))$.
        \item For each $x\in A,y\in B,z\in C$, we have that $q(\overline f(x,y,z)) = f'(x,y,z)$.
    \end{itemize}
\end{lemma}
\begin{proof}
    Let $\mathcal{T}$ be the union of all subexpressions of both $f$ and $g$. Apply the same construction as in the Subexpression Embedding gadget and let $\tilde f, \tilde f', \tilde g$ be the embedded expressions corresponding to $f,f',g$ respectively.
    Let $\oplus_j$ be the operation that appears in the root of $f$.
    Now modify $\tilde \oplus_j$ as follows:
    \[
        (x,t_1) \tilde\oplus_j (y, t_2) \gets \begin{cases}
            (x, t_1) \tilde \oplus' (y, t_2) &\text{if $t_1\equiv L$ and $t_2\equiv R$}\\
            (x, t_1) \tilde\oplus_j (y, t_2) &  \text{otherwise}
        \end{cases}
    \]
    Set $\tilde A:= \{(x, a)\mid x\in A\}$, $\tilde B:= \{(y, b)\mid y\in B\}$, $\tilde C:= \{(z,c)\mid z\in C\}$. Intuitively, we embed the corresponding leaf of $f$ (resp. $g$) to each of the sets $A,B,C$.
    Now set $S' = S\times \mathcal T$ and define $q: (x,T)\mapsto x$.
    We first show that for each $\tilde x\in \tilde A, \tilde y\in \tilde B, \tilde z \in \tilde C$, the embedded expression $\tilde g$ remains unchanged.
    \begin{claim}
        For each $x\in A, y\in B, z\in C$ we have $g(x,y,z) = \tilde q(g((x, a), (y,{b}), (z, {c})))$.
    \end{claim}
    \begin{subproof}
        We first notice that for each subexpression $T$ of $g$ that does not contain the operation $\oplus_j$, this follows trivially from Observation \ref{obs:chromaticity}.
        Now assume that there is a subexpression $T$ of $g$ that can be written as $T \equiv t_1 \oplus_j t_2$.
        We can argue inductively that the statement holds for both $t_1$ and $t_2$. Now, by our construction, the only way that $T$ evaluates to a different value is if $ t_1\equiv L, t_2 \equiv R$.
        However, since $f\equiv L \oplus_j R$ and ${\rm depth}(f)\ge {\rm depth}(g)$, this would imply $f\equiv g$, contradicting that $f\not\sim g$.
    \end{subproof}
    \begin{claim}
        For each $x\in A, y\in B, z\in C$ we have $f'(x,y,z) = \tilde f((x, a), (y,{b}), (z, {c}))$.
    \end{claim}
    \begin{subproof}
        Follows directly from our construction, Observation \ref{obs:chromaticity} and the fact that $f\not\sim g$.
    \end{subproof}
    The proof of the lemma now follows from the two claims above.
\end{proof}
We now introduce the concept of \emph{collapsing} one expression into the other.
On a high level, an expression $f$ can be collapsed into an expression $g$ if, given an instance\footnote{Here by instance we mean a multichromatic set $S$ together with Cayley tables defining each of the operations appearing in $g$.} of $g$, we can set up the Cayley tables for $f$ such that $f(a,b,c)$ and $g(a,b,c)$ always evaluate to the same element.
\begin{definition}[Expression Collapsing]
    We say an expression $f(a,b,c)$ can be \emph{collapsed} into an expression $g(a,b,c)$ if given any set $S$ with (pairwise disjoint) subsets $A,B,C$ and Cayley tables for operations in $g$, we can construct a set $S'\supseteq S$ of size $\bigO(|S|)$, the Cayley tables for operations in $f$ and a unary function $q: S'\rightarrow S$ such that $q(f(a,b,c)) = q(g(a,b,c))$ for each $a\in A, b\in B, c\in C$.
\end{definition}
\begin{lemma}\label{lemma:collapsing-2-vars}[Collapsing Lemma]
    Let $f(a,b)$ be an expression over operations $\oplus_1,\dots, \oplus_p$ that depends on both $a,b$. Then given a set $S$, disjoint subsets $A,B\subseteq S$ and an arbitrary operation $\odot: A\times B \rightarrow S$, we can construct a set $S' \supseteq S$ of size $\bigO(|S|)$ and new Cayley tables for $\oplus_1,\dots, \oplus_p$ and a unary function $q: S'\rightarrow S$ such that for each $a\in A$, $b\in B$ we have $q(f(a,b)) = a\odot b$.
\end{lemma}
\begin{proof}
    Let $S' := S \cup \{(x,1) : x\in S\}$. In order to unify notation for defining $\oplus_i$, we say that $(x,0)$ represents $x$ for all $x\in S$.
    We now construct for each $i\in [p]$ and $x,y\in S$ the operation $\oplus_i$ as follows:
    \[
    (x,0) \oplus_i (y,0) =\begin{cases}
        (x,0) & \text{if }x=y\\
        (x\odot y,1) & \text{if } x\in A, y\in B \\
        (y\odot x,1) & \text{if } x\in B, y\in A
    \end{cases}
    \]
    Note that the three mentioned cases are the only possible, as the expression $f$ does not have constants (i.e. elements from $S$, non-variables) and we consider only $f$ evaluated on $a\in A, b\in B$ where $A$ and $B$ are disjoint.
    Moreover, let $(x,0) \oplus_i (y,1)=(y,1)$ and finally let $(x,1)\oplus_i (y, b)= (x,1)$ for any $b\in \{0,1\}$.
    Observe that this construction guarantees that the $\odot$ operation is performed exactly once. Indeed, the second parameter counts the number of operations performed, so $\odot$ is applied at most once, and since $f$ depends on both $a$ and $b$, $\odot$ will be performed at least once.
    Then $f(a,b) = (a\odot b, 1)$, so we set $q((x,b))=x$ for all $x\in S, b\in\{0,1\}$ and the claim follows.
\end{proof}
We can now apply the Collapsing Lemma to show that if we have two non-similar expressions $f(a,b), g(a,b)$, we can collapse them in such a way that (without loss of generality) $f(x,y)$ always evaluates to $x\odot y$ and $g(x,y)$ always evaluates to $x$ for any given operation $\odot$.
\begin{lemma}\label{lemma:collapsing-non-similar}
    Let $f(a,b)$ and $g(a,b)$ be two expressions that depend on both $a$ and $b$ over operations $\oplus_1,\dots, \oplus_p$, such that $f(a,b)\not\sim g(a,b)$.
    Given a set $S$, two disjoint subsets $A,B\subseteq S$ and an operation $\odot: S\times S \to S$, we can construct in $\bigO(|S|^2)$ time a set $S'\supseteq S$ of size $\bigO(|S|)$, the Cayley tables for the operations $\oplus_1,\dots, \oplus_p: S'\times S'\to S'$ and a unary function $q: S'\to S$ such that either:
    \begin{enumerate}
        \item $q(f(x,y)) = y$ and $q(g(x,y)) = x\odot y$ for every $x\in A$, $y\in B$, or
        \item $q(f(x,y)) = x\odot y$ and $q(g(x,y)) = y$ for every $x\in A$, $y\in B$.
    \end{enumerate}
\end{lemma}
\begin{proof}
Both expression depend on two variables, so their depth is at least 1.
We prove the claim by induction on ${\rm depth}(f)+{\rm depth}(g)$.
If ${\rm depth}(f)={\rm depth}(g)=1$, $f$ and $g$ consist of a single operation and we have two cases. If they consists of distinct operations $\oplus_f$ and $\oplus_g$, we can set: $u \oplus_f v = u$ if $u\in B$ and $v$ otherwise; and $u \oplus_g v = u \odot v$ if $u\in A, v\in B$ and $v \odot u$ otherwise and the claim follows (the first case).
If they consist of the same operation $\oplus$, it must be the case that $f\equiv a \oplus b$ and $g\equiv b\oplus a$ (or the other way round), as $f\not\sim g$.
Then we set
$u \oplus v = u \odot v$ if $u\in A, v\in B$ and $u$ otherwise (note that now $u\in B$) and the claim follows (first case if $g\equiv a \oplus b$, second otherwise).

    For the inductive step, assume without loss of generality that ${\rm depth}(f)\ge {\rm depth}(g)$. We write $f\equiv L \oplus_f R$, where by Lemma \ref{lemma:operations}, we can assume that $\oplus_f$ does not appear in any subexpression of $g$, $L$, or $R$.
    As ${\rm var}(f)=\{a,b\}$ (i.e. $f$ depends on both $a$ and $b$), there are three cases depending on which variables do subtrees $L$ and $R$ depend on (other cases are symmetric):
    \begin{enumerate}[label=(\roman*)]
        \item ${\rm var}(L)=\{a\}, b\in {\rm var}(R)$. In this case we set tables $\oplus_i$ for all $\oplus_i\ne \oplus_f$ in the following way: $u\oplus_i v=u $ if $u\in B$ or $v$ otherwise. Then for $x\in A, y\in B$ we have: $R(x,y)=y$ and $g(x,y)=y$  as $b\in {\rm var}(R)\cap{\rm var}(g)$ and $L(x,y)=x$ (as ${\rm var}(L)=\{a\}$).
        Finally we set the operation $\oplus_f$ that appears only in the root of $f$: $u \oplus_f v = u \odot v$. As $f(x,y) = L(x,y)\oplus_f R(x,y)= x \oplus_f y = x \odot y$, the claim (the second case) follows. Function $q$ is the identity: $q(x)=x$ for $x\in S$ and $S'=S$.
        \item ${\rm var}(L)={\rm var}(R)=\{a,b\}$.
        It holds $g\not\sim L$ or $g\not\sim R$, since otherwise we have $g\sim L, g\sim R$ and then $L\sim R$ (as $\sim$ is equivalence class by Lemma \ref{lemma:equivalence-of-similarity}) which gives that $f\sim g$, a contradiction.
        Without loss of generality, we assume $L\not\sim g$.
    Then we apply inductive hypothesis for $L$ and $g$ and set $u \oplus_f v= u$ and the claim follows by inductive hypothesis (the first or second case). Function $q$ is the identity: $q(x)=x$ for $x\in S$ and $S'=S$.
        \item ${\rm var}(L)=\{b\},{\rm var}(R)=\{a,b\}$. Now we apply a similar construction to that in the proof of Lemma \ref{lemma:collapsing-2-vars}. Let $S'=S\cup \{(x,1): x\in S\}$ and, to simplify presentation, we identify elements $x\in S$ with $(x,0)$. Then we set:
        \begin{align*}
        (u,0)\oplus_i(v,0) &= \begin{cases}
            (u,0) &\text{if } u=v\\
            (u\odot v,1) &\text{if } u\in A, v\in B\\
            (v\odot u,1) &\text{if } u\in B, v\in A\\
        \end{cases}\\
        (u,1)\oplus_i(v,1) &= (u,1) \quad \text{note that in that case } u=v\\
        (u,0)\oplus_i(v,1) &= (v,1)\\
        (u,1)\oplus_i(v,0) &= (u,1)
        \end{align*}
        Note that for $(u,0)\oplus_i(v,0)$ there are no more cases to consider, as our expression contains only variables $a$ and $b$ that are evaluated to the given values $x$ and $y$ respectively and the expression does not contain constants.
    Then, the only possible returned value of the form $(e,1)$ is for $e=x\odot y$.
    Next, $R(x,y)=g(x,y) = (x\odot y,1)$ as ${\rm var}(R)={\rm var}(g)=\{a,b\}$ and $L(x,y)=(y,0)$ as ${\rm var}(R)=\{b\}$.
    Finally we set $u \oplus_f v=u$ and then $f(x,y)=L(x,y)\oplus_f R(x,y) = L(x,y)=(y,0)$, so it suffices to set $q((e,b))=e$ for all $e\in S,b\in\{0,1\}$ and the claim (first case) follows.
    \end{enumerate}
\end{proof}
A very similar approach gives us the following generalization of the lemma above.
\begin{corollary}\label{cor:collapse}
    Let $f(a,b,c)\equiv J(I(a,b), c)$ and $g(a,b,c) \equiv J'(I(a,b),c)$ such that $J\not\sim J'$ and $\var{J} = \var{J'} = \{a,b,c\}$.
    Given a set $S$, two disjoint subsets $A,B\subseteq S$ and an operation $\odot: S\times S \to S$, we can construct in $\bigO(|S|^2)$ time a set $S'\supseteq S$ of size $\bigO(|S|)$, the Cayley tables for the operations $\oplus_1,\dots, \oplus_p: S'\times S'\to S'$ and a unary function $q: S'\to S$ such that at least one of the two cases holds:
    \begin{enumerate}
        \item $q(f(x,y,z)) = (x\odot y)\odot z$ and $q(g(x,y,z)) = z$ for every $x\in A$, $y\in B$, $z\in C$,
        \item $q(f(x,y,z)) = z$ and $q(g(x,y,z)) = (x\odot y)\odot z$ for every $x\in A$, $y\in B$, $z\in C$,
    \end{enumerate}
\end{corollary}

We are now ready to prove our first lower bound from Theorem \ref{thm:trichotomy}.
\begin{proposition}[Triangle hard regime]\label{prop:triangle-hardness}
    Let $f(a,b,c)$ be an expression that (without loss of generality) cannot be written as $f(a,b,c)\equiv H(G(a,b), c)$ for any two-variable expressions $G(a,b), H(a,b)$.
    Then for any tripartite graph $\mathcal{G}$, we can in time $\bigO(n^2)$ construct a set $S$ with subsets $A,B,C$, and Cayley tables for the operations in $f$ such that there exists a triple $x\in A, y\in B, z\in C$ such that $f(x,y,z) \ne \infty$ if and only if $\mathcal{G}$ contains a triangle.
\end{proposition}
\begin{proof}
    Let $f \equiv L\oplus R$ for some subexpressions $L,R$.
    Without loss of generality, we can assume that there are two variable expressions $G_L(a,b), H_L(a,b), G_R(a,b), H_R(a,b)$ such that either
    \begin{enumerate}[label=\roman*]
        \item $L(a,b,c)\equiv H_L(G_L(a,b), c)$ and $R(a,b,c)\equiv H_R(G_R(a,b), c)$, or
        \item $L(a,b,c)\equiv H_L(G_L(a,b), c)$ and $R(a,b,c)\equiv H_R(a, G_R(b,c))$.
    \end{enumerate}
    Note that either of these two cases must be true (up to symmetry of $L,R$ and $a,b,c$), since otherwise we can apply Lemma \ref{lemma:subexpression-monotonicity} recursively on subexpressions of $f$ until we reach a minimal expression that is not expressible in this form. 
    Assume we are not in the first case (i.e. $L,R$ cannot be written as $H_L'(G_L'(a,b), c)$ and $H_R'(G_R'(a,b),c)$) for any expressions $G_L', G_R', H_L', H_R'$.
    Then (without loss of generality) we can write $L(a,b,c)\equiv H_L(G_L(a,b), c)$ and $R(a,b,c)\equiv H_R(a, G_R(b,c))$.
    Moreover, notice that $\var{G_R} = \{b,c\}$ and $\var{G_L} = \{a,b\}$, since otherwise we would be in the first case (up to symmetry of $a,b,c$).
    Given a tripartite graph $\mathcal{G}$ with vertex parts $I,J,K$, define $S:= I\cup J \cup K\cup \{\infty\}$ and set $A:=I, B:=J, C:=K$.
    Order the elements in $S$ such that for each $\alpha\in A, \beta\in B, \gamma\in C$ we have that $\beta< \alpha <\gamma$.
    Define the operation $\oplus_i$ for each $i$ on $S$ as follows.
    First let $\infty \oplus_i x = x \oplus_i \infty = \infty$ for each $x \in S$.
    Now define
    \[x\oplus_i y = y\oplus_i x =\begin{cases}
        x & \text{if $x=y$}\\
        \infty & \text{if $x,y\not\in E(\mathcal{G})$}\\
        \max\{x,y\} & \text{otherwise}
    \end{cases}\]
    Since $\var{G_R} = \{b,c\}$, we know that $R(a,b,c)$ will check if there is an edge between $b$ and $c$ and return either $\infty$, or $c$. Similarly, since $\var{G_L} = \{a,b\}$, we know that $L(a,b,c)$ will check if there is an edge between $a$ and $b$.
    We note that if either of these edges is missing, the whole expression evaluates to $\infty$.
    Hence, assume that both of these edges are present.
    Then the subexpression $G_L$ will evaluate to $a$, while the subexpression $R$ evaluates to $c$.
    By our construction, if $\var{L} = \{a,b\}$, then $L$ also evaluates to $a$ and the root operation of $f$ will return $\infty$ if there is no edge between $a,c$ and $c$ otherwise.
    On the other hand, if $\var{L} = \{a,b,c\}$, then clearly by our construction, the operation $a\oplus_ic$ (or $c\oplus_i a$) will be performed in $L$ at some point and evaluate to either $\infty$ if the edge $a,c$ is missing and $c$ otherwise. In both cases, we get that for each $a\in A,b\in B, c\in C$ the expression $f(a,b,c)$ evaluates to $c$ if and only if $a,b,c$ form a triangle and $\infty$ otherwise. We can now assume we are in the first case, namely when $L,R$ can be written as $L(a,b,c)\equiv H_L(G_L(a,b), c)$ and $R(a,b,c)\equiv H_R(G_R(a,b), c)$.

    Now (up to symmetry), we have two cases that we will check separately:
    \begin{itemize}
        \item \emph{Case 1}: $L$ depends only on $a,b$ and $R$ depends only on $b,c$.
        \item \emph{Case 2}: $L$ depends on all three variables $a,b,c$.
    \end{itemize}
    Given a tripartite graph $\mathcal{G}$ with vertex parts $I,J,K$, define $S:= I\cup J \cup K\cup \{\infty\}$ and set $A:=I, B:=J, C:=K$.
    Order the elements in $S$ such that for each $\alpha\in A, \beta\in B, \gamma\in C$ we have that $\alpha<\beta<\gamma$.
    Define the operation $\odot$ on $S$ as follows.
    First let $\infty \odot x = x\odot \infty = \infty$ for each $x\in S$.
    Now define
    \[x\odot y = y\odot x =\begin{cases}
        \infty & \text{if $x,y\not\in E(\mathcal{G})$}\\
        \max\{x,y\} & \text{otherwise}
    \end{cases}\]
    Now applying the Collapsing Lemma (Lemma \ref{lemma:collapsing-2-vars}) on $L$ and $R$ separately (intuitively, we can do this since $A$ only intersects $L$ and $C$ only intersects $R$), we can construct Cayley tables for operations in $f$ such that for each $x\in A, y\in B, z\in C$ we have that $f(x,y,z) = (x\odot y) \odot (x\odot z)$. It is now easy to see that $\mathcal{G}$ contains a triangle $x,y,z$ if and only if $f(x,y,z) \ne \infty$.
    It remains to show that the lemma holds in Case 2 as well.
    Assume that $L$ depends on all three variables.
    By minimality of $f$, (without loss of generality) $L$ can be written as $L(a,b,c)\equiv H_L(G_L(a,b), c)$ where $G_L$ depends on both $a$ and $b$.
    Then we know that $R(a,b,c)$ cannot depend exclusively on $c$, since otherwise $f$ could be written as $H_L(G_L(a,b), c) \oplus H_R(c) \equiv H(G_L(a,b),c)$, where $H$ is the expression $H(a,b) \equiv H_L(a, b) \oplus H_R(b)$.
    \begin{claim}
        Assume that $R(a,b,c)\equiv H_R(G_R(a,b), c)$, where $G_R$ is an expression that depends on at least one of the variables $a,b$.
        Then $G_L(a,b)\not\sim G_R(a,b)$.
    \end{claim}
    \begin{subproof}
        Assume for contradiction that $G_L(a,b)\sim G_R(a,b)$. But then, there exists an expression $G$ such that both $G_L, G_R$ are decomposable into $G$, and in other words, there is an expression $H$ such that we can write $f(a,b,c)\equiv H(G(a,b), c)$, contradicting the assumption that no such pair of expressions exists.
    \end{subproof}
    Now by applying Lemma \ref{lemma:collapsing-non-similar}, we can construct Cayley tables in such a way that (without loss of generality) $G_L(a,b)$ always evaluates to $a\odot b$ (where $\odot$ is the operation defined above), while $G_R(a,b)$ always evaluates to $b$.
    In particular, this shows that we can use $f$ to evaluate the expression $H_L(a\odot b, c) \oplus H_R(a,c)$ where $H_L$ depends on both variables (since we are in Case 2) and $H_R$ depends on $a$ (we do not know if it also depends on $c$).
    We remark that one has to argue that we can apply the Collapsing Lemma without interfering with the operations in $H_L$ and $H_R$, but this follows quite easily from the construction of Subexpression Embedding gadget (similar idea as in the proof of Lemma \ref{lemma:operations}). For simplicity, we omit the details.
    Now, by definition of the operation $\odot$, if for any pair $x\in I, y\in J$ the edge $x,y$ is not present, the term $x\odot y$ will evaluate to~$\infty$.
    We can set all Cayley tables $\oplus_i$ in $H_L\oplus H_R$ in such a way that $\infty \oplus_i x = x\oplus_i \infty = \infty$.
    Hence for every triple $x,y,z$, the expression $H_L(a\odot b, c) \oplus H_R(a,c)$ can tell us if the edge $xy$ is missing.
    Hence we assume that for some triple $x,y,z$ the edge $xy$ is present.
    We prove that this expression can detect any other edge missing as well.
    In particular, if the edge $xy$ is present, by construction of $\odot$, the expression above evaluates to the expression $H_L(b, c) \oplus H_R(a,c)$.
    We order the elements in $S$ again, but this time in such a way that for each $x\in A, y\in B, z\in C$ we have that $y<z<x$.
    We can now construct the Cayley table each of the operations $\oplus_i$ appearing in $H_L$ and $H_R$, as well as the top-level operation $\oplus$ as follows.
    \[
    x\oplus_i y = y\oplus_i x = \begin{cases}
        x & \text{if $x=y$}\\
        \infty & \text{if $x\ne y$ and $x,y\not\in E(\mathcal{G})$}\\
        \max\{x,y\} & \text{otherwise}
    \end{cases}
    \]
    Also, as mentioned above, define $\infty \oplus_i x = x\oplus_i \infty = \infty$ for each $x\in S$.
    It is easy to verify that if $x,y,z$ form a triangle, $H_L(y,z)$ will evaluate to $z$, and $H_R(x,z)$ will evaluate to $x$.
    In particular, $H_L(x\odot y, z) \oplus H_R(x,z)$ will evaluate to $x$.
    Otherwise, if edge $y,z$ is missing, $H_L(y,z)$ will evaluate to $\infty$ (since $H_L$ depends on both variables), and hence the whole expression evaluates to $\infty$.
    Finally, if edge $xz$ is the only edge missing, either $H_R(x,z)$ will evaluate to $\infty$ (if $H_R$ depends on both variables), or similarly as above $H_R(x,z)$ will evaluate to $x$ and $H_L(y,z)$ will evaluate to $z$ and hence $H_L(x\odot y, z) \oplus H_R(x,z) = x\oplus z = \infty$.
    In particular, this shows that $f$ can be used (by collapsing) to evaluate the expression $T\equiv H_L(x\odot y, z) \oplus H_R(x,z)$, and for each $x\in A, y\in B, z\in C$ we have that $T(x,y,z) = \infty$ if and only if $x,y,z$ does not form a triangle in $\mathcal{G}$.
    This concludes our proof.
\end{proof}
The rest of this section is dedicated to proving the hardness of the identities that do not fall into the 'triangle regime'. To that end we show that all such identities can be efficiently reduced from $4$-AP problem and moreover, a subset of them can be reduced Strong Zero Triangle (as discussed in the statement of Theorem \ref{thm:trichotomy}).
The lemma below follows directly by applying the Subexpression Embedding Gadget to Lemma \ref{lem:4aphard1234} and Lemma \ref{lem:4aphard5}.
\begin{lemma}[Core $4$-AP Hard Identities]\label{lemma:core-identities}
    Assuming $4$-AP Hypothesis, for none of the following identities can we solve the Constant Term Identity problem in time $\bigO(n^{3-\varepsilon})$ for any $\varepsilon>0$.
    \begin{enumerate}
        \item \label{core:f1} $f_1(a,b,c) \equiv b\oplus(a\odot(b\odot(c\odot a))))$
        \item \label{core:f2} $f_2(a,b,c) \equiv c\oplus(a\odot(b\odot(c\odot a))))$
        \item \label{core:f3} $f_3(a,b,c) \equiv ((a\odot b) \oplus (a\odot c))\oplus a$
        \item \label{core:f4} $f_4(a,b,c) \equiv ((a\odot b) \oplus (a\odot c))\oplus b$
        \item \label{core:f5} $f_5(a,b,c) \equiv (a\odot b) \oplus (b\odot (a\odot c))$
        \item \label{core:f6} $f_6(a,b,c) \equiv (a\odot b) \oplus (a\odot (c\odot (a\odot b)))$
    \end{enumerate}
\end{lemma}
Our strategy is to show that each of the expressions that do not satisfy the triangle regime condition (as stated by Theorem \ref{thm:trichotomy}) can be efficiently reduced from at least one of the six core identities above (up to the symmetry of variables $a,b,c$).
More formally, if Constant Term Identity problem can be solved in $\bigO(n^{3-\varepsilon})$ on any expressions $f$ such that for no $G,H,I, J$ can $f$ be written as $J(I(H(a,b), c), G(a,b))$, then it can also be solved in $\bigO(n^{3-\varepsilon'})$ on at least one of the $f_i$'s above, thus refuting the $4$-AP Hypothesis.
We first look at the case when $f$ contains two subexpressions $g(a,b)$ and $h(a,c)$ (without loss of generality) that each depend on both variables.
The following lemma is in its nature similar to our first Collapsing Lemma, but will help us significantly simplify our main proof.
\begin{lemma}[Second Collapsing Lemma]\label{lemma:second-collapse}
    Let $f(a,b,c)\equiv L\oplus R$ for some operation $\oplus$ that only appears in the root. Assume the following conditions:
    \begin{itemize}
        \item $L$ contains a maximal subexpression $g(a,b)$ that depends on both $a,b$.
        \item $R$ contains no maximal subexpression $g'(a,b)$ depending on both $a,b$.
        \item $R$ contains a maximal subexpression $h(a,c)$ that depends on both $a,c$.
        \item ${\rm var}(R) = \{a,b,c\}$.
    \end{itemize}
    Assume $L,R$ are over operation $\oplus_1,\dots, \oplus_p$ and let $\odot$ be an arbitrary given operation (not necessarily in $\oplus_1,\dots, \oplus_p$). Then we can in time $\bigO(n^2)$ construct Cayley tables for all $\oplus_i$, as well as for $\oplus$ (root operation), such that 
    solving the (Multichromatic) Constant Term Identity problem on $f$ in time $\bigO(n^{3-\varepsilon})$ would imply an algorithm solving the problem on $f_5$ (as defined in Lemma \ref{lemma:core-identities}) in $\bigO(n^{3-\varepsilon'})$, refuting the $4$-AP Hypothesis.
\end{lemma}
\begin{proof}
    We first apply Lemma \ref{lemma:collapsing-2-vars} on $g$ and $h$ independently.
    Notice that this can be done independently by defining the operations between $a\in A,c\in C$ and operations between $a\in A,b\in B$ separately as in Lemma \ref{lemma:collapsing-2-vars} (we moreover note that collapsing $g$ can be done such that it does not even affect any operation in $R$, since $R$ contains no subexpression $g'(a,b)$ depending on both $a,b$).
    We may now assume that $g,h$ always evaluate to $a\odot b$ and $a\odot c$ and that they do not contain any of the $\oplus_i$ operations.
    Now apply the subexpression embedding construction, setting
    \[
    (x,t_1) \oplus_i (y,t_2) = \begin{cases}
        (x,t_1\oplus_i t_2) & \text{if $t_1$ contains $g$}\\
        (y,t_1\oplus_i t_2) & \text{else if $t_2$ contains $g$}\\
        (x,t_1) \oplus_i' (y, t_2) & \text{otherwise}
    \end{cases}
    \]
    Note that we will define the operations $\oplus_i'$ in a moment, but first we notice that the construction above implies that the expression $L$ will always evaluate to $a\odot b$, while the expression $R$ has all of the operations $\oplus_i$ replaced by $\oplus_i'$ (recall that $R$ does not contain $g$ as a subexpression).
    Particularly, the construction so far gives us that for any $a\in A, b\in B, c\in C$, the expression $f(a,b,c)$ will evaluate to $f(a,b,c) = (a\odot b)\oplus R'$, where $R'$ is obtained from $R$ by replacing each operation $\oplus_i$ by $\oplus_i'$.
    We now construct the $\oplus_i'$ operations. Consider the ordering $\succ$ of the elements in $\tilde S$ (input set with embedded subexpressions) such that if
    \begin{enumerate}
        \item $(x,t_1)$ is an arbitrary element where $t_1$ contains both $h$ and $\expr{b}$ as a subexpressions
        \item $(y,t_2)$ is an arbitrary element where $y\in B$ and $t_2$ does not contain $h$ as a subexpression
        \item $(z,t_3)$ is an arbitrary element where $t_3$ contains $h$ as a subexpression, but does not contain $\expr b$ as a subexpression
        \item $(w,t_4)$ is an arbitrary element where $t_4$ contains neither $\expr{b}$, nor $h$ as subexpressions,
    \end{enumerate}
    then $(x,t_1)\succ (y,t_2) \succ (z,t_3)\succ (w,t_4)$. The other elements are ordered arbitrarily. We notice that each element in $\tilde S$ falls into one of the four categories.
    Let $\oplus_i'$ be defined as follows.
    \[
    (x,t_1) \oplus_i' (y,t_2) = \begin{cases}
        (x\odot y, t_1\oplus_it_2) & \text{if $x\in B$ and $t_2$ contains $h$, but not \expr{b}}\\
        (y\odot x, t_1\oplus_i t_2)& \text{else if $y\in B$ and $t_1$ contains $h$, but not \expr{b}}\\
        (x, t_1\oplus_i t_2) & \text{else if $(x,t_1)\succ(y,t_2)$}\\
        (y,t_1\oplus_i t_2) & \text{otherwise}
    \end{cases}
    \]
    While the definition above is somewhat convoluted, let us break it down.
    Consider the categories 1-4 from our ordering. When both elements come from category 4, we promote any of the elements, while keeping track of the operations performed so far. Note that this is fine, since $h$ does not depend on variable $b$, nor does it contain any of the operations $\oplus_i$.
    Now an element from category 3 (resp. 2) will be promoted until it reaches an element from category 2 (resp. 3), when the $\odot$ operation will be performed in the right order on these two elements to give the required subexpression $b\odot (a\odot c)$.
    Finally, once this operation is performed once, we obtain an element from the first category, and this element will be promoted all the way to the root according to our $\succ$ operation, when the operation $\oplus$ is performed.
    For completeness, we define:
    \[
    (x,t_1)\oplus (y,t_2) = \begin{cases}
        x\oplus y & \text{if $t_1\equiv L$ and $t_2 \equiv R$}\\
        \infty &\text{otherwise}
    \end{cases}
    \]
    It is now straightforward to verify that there exists $x\in A, y\in B, z\in C$ such that $f(x,y,z) \ne \infty$ if and only if there exists $a\in A, b\in B, c\in C$, such that $(a\odot b) \oplus (b\odot (a\odot c))\ne \infty$ just like we wanted to prove.
\end{proof}
\begin{lemma}\label{lemma:4-ap-two-maximal-subtrees}
    Let $f(a,b,c)$ be an expression such that for no $G,H,I, J$ does the $f(a,b,c)\equiv J(I(H(a,b), c), G(a,b))$ hold.
    Assume that $f$ contains two inclusion maximal subexpressions $g,h$ such that $g$ depends precisely on variables $a,b$ and $h$ depends precisely on variables $a,c$ (without loss of generality, the remaining two cases are symmetric).
    Then any algorithm deciding Constant Term Identity on $f$ in time $\bigO(n^{3-\varepsilon})$ would refute the $4$-AP Hypothesis.
\end{lemma}
\begin{proof}
    Let $f(a,b,c)\equiv L\oplus' R$ and by applying the Subexpression Embedding Gadget, similarly as in Lemma \ref{lemma:operations} we can assume that $\oplus'$ operation does not appear in either $L$ nor $R$.
    Without loss of generality, we assume that $f(a,b,c)$ is a minimal expression satisfying the condition of the lemma (i.e., both $L$ and $R$ can be written as $J(I(H(a,b), c), G(a,b))$ for a suitable choice of expressions $G,H,I,J$).
    We also assume without loss of generality that $|{\rm var}(L)| \le |{\rm var}(R)|$.
    The proof of the following claim is essentially a reiteration of the same approach that was already used in a few places across the section, so we only provide a sketch here.
    \begin{claim}
        Let $u$ be an expression that contains two subexpressions $g,h$ such that ${\rm var}(g) = \{a,b\}$, and ${\rm var}(h) = \{a,c\}$. Then we can in time $\bigO(n^2)$ construct Cayley tables of operations appearing in $u$ such that for any given pair of operations $\odot, \oplus$, and for each triple $x\in A,y\in B,z\in C$, we have that $u(x,y,z) = (x\odot y) \oplus( x\odot z)$.
    \end{claim}
    \begin{subproof}[Proof sketch]
        By approach similar as in proof of Proposition \ref{prop:triangle-hardness}, we can apply the Collapsing Lemma on $g$ and $h$ separately using chromaticity property.
        We can thus collapse $g$ into $a\odot b$ and $h$ into $a\odot c$.
        On a high level, by applying the Subexpression Embedding Gadget, we keep track of where in the parse tree we are, while making sure that we only perform the relevant information once.
    \end{subproof}
    We now consider the case distinction depending on the value of $|{\rm var}(L)|$.

    \textbf{Case 1:} $|{\rm var}(L)| = 1$. Then by our assumption, $R$ contains both $g,h$ as subexpressions.
    Then, given any pair of operations $\oplus, \odot$, by the claim above, we can construct the Cayley tables for all the operations in $R$ such that for all $a\in A,b\in B,c\in C$, we have $R(a,b,c) = (a\odot b) \oplus (a\odot c)$.
    Recall that $f(a,b,c) \equiv L\oplus'R$, hence we can construct the Cayley table for operations in $L$, as well as $\oplus'$ such that for each $a\in A, b\in B, c\in C$ we have that (up to symmetry) $f(a,b,c)$ either evaluates to $a\oplus ((a\odot b)\oplus (a\odot c))$, or $b\oplus ((a\odot b)\oplus (a\odot c))$.
    Hence, deciding if there is an algorithm deciding if $f(a,b,c) = \infty$ in time $\bigO(n^{3-\varepsilon})$ for each $a\in A,b\in B,c\in C$, then there is an algorithm solving the Constant Term Identity on either $f_3(a,b,c)$, or $f_4(a,b,c)$ in time $\bigO(n^{3-\varepsilon})$, thus refuting the $4$-AP Hypothesis.

    \textbf{Case 2:}
    $|{\rm var}(L)| = 2$. Notice first that if $|{\rm var}(R)| \le 2$, then the conditions of the lemma are not met (i.e., we can represent $f$ trivially as $f(a,b,c)\equiv J(I(H(a,b), c), G(a,b))$). Hence, we assume that $|{\rm var}(R)| = 3$.
    Note that without loss of generality, we can assume that ${\rm var}(L) = \{a,b\}$ and in particular we can assume that $L\equiv g$.
    Consider the structure of $R$.
    By minimality of $f$, we can (up to symmetry of $a,b,c$; we prove the claim for this case and note that other cases can be handled similarly) write $R$ as $J_R(I_R(H_R(a,b), c), G_R(a,b))$.
    By the assumption that $h$ is a subexpession of $R$ and depends precisely on $a,c$, we can conclude that $H_R$ depends only on $a$.
    In particular, we can rewrite  $R(a,b,c)\equiv J_R(I_R(a, c), G_R(a,b))$.
    Assume now that $G_R\sim L$. But then there exists an expression $X(a,b)$, such that both $G_R$ and $L$ are decomposable into $X$.
    In particular, we can write $f\equiv J(X(a,b), I_R(a,c))$, which is a contradiction.
    Hence, we may assume that $G_R\not \sim L$.
    But then we can apply \ref{lemma:collapsing-non-similar} to either collapse $G_R$ into $a\odot b$ and $L$ into $b$, or vice versa.
    In the former case, we land into Case 1, which we already handled above.
    In the latter case, we notice that after collapsing, $f$ satisfies the conditions of Lemma \ref{lemma:second-collapse}, yielding the desired hardness in both cases.

    \textbf{Case 3:} $|{\rm var}(L)| = 3$.
    If both $g$ and $h$ are subexpressions of $L$, then a similar construction as in Lemma \ref{lemma:second-collapse} can be used to show that $L$ can be collapsed into $(a\odot b)\oplus (a\odot c)$, while making sure that $R$ collapses to either $a$ or $b$.
    Thus solving the Constant Term Identity problem on $f$ in subcubic time would give an algorithm for solving the problem on either $f_3$, or $f_4$, thus refuting the $4$-AP hypothesis.
    Assume now that $L$ contains $g$, and $R$ contains $h$ (otherwise, by our assumption, $R$ contains both $g$ and $h$ and we already know how to deal with this case).
    On top of this, we may assume that $R$ contains no subexpression $u$ with ${\rm var}(u) = \{a,b\}$, since we could simply relabel $g$ and $u$ to land in the previous case again.
    But now we notice that $R$ satisfies all conditions of Lemma \ref{lemma:second-collapse}, thus concluding the proof.
\end{proof}
Notice that in the lemma above we assumed that $f$ contains two maximal subexpressions $g,h$ each of which contained a different pair of variables (e.g., $a,b$ and $a,c$ respectively). It remains to show hardness of the problem on expressions that (1) cannot be written as $f(a,b,c)\equiv J(I(H(a,b), c), G(a,b))$, and (2) $f$ contains no subexpressions $u(a,c)$, nor $u'(b,c)$ that depend on both variables.
We again without loss of generality consider the minimal expression $f$ that satisfies these properties (i.e., each subexpression of $f$ can be written (without loss of generality) in the form $J(I(H(a,b), c), G(a,b))$), and moreover, we consider two subcases depending on whether $f$ contains a subexpression that is triangle-hard (according to Proposition \ref{prop:triangle-hardness}).
We start with the case when all of the subexpressions of $f$ can be written as $H(G(a,b),c)$ (i.e., are not triangle hard according to Proposition \ref{prop:triangle-hardness})
\begin{lemma}\label{lemma:4-ap-one-maximal-subtree-easy}
    Let $f(a,b,c)$ be an expression such that for no $G,H,I, J$ does the $f(a,b,c)\equiv J(I(H(a,b), c), G(a,b))$ hold, while for any subexpression $g$ of $f$, $g$ can be written as $g(a,b,c)\equiv H(G(a,b),c)$ (up to symmetry of $a,b,c$) for some expressions $G,H$.
    Assume further that each subexpression $z$ that depends on exactly two variables satisfies
    ${\rm var}(z) = \{a,b\}$ (without loss of generality, the remaining two cases are symmetric).
    Then any algorithm deciding Constant Term Identity on $f$ in time $\bigO(n^{3-\varepsilon})$ would refute the $4$-AP Hypothesis.
\end{lemma}
\begin{proof}
    Write $f\equiv L\oplus R$ for some operation $\oplus$ appearing only in the root. We first note that by assumption that we must have that $f$ cannot be written as $f(a,b,c)\equiv J(I(H(a,b), c), G(a,b))$, we must have ${\rm var}(L) = {\rm var}(R) = \{a,b,c\}$ (that is both subexpressions must depend on all three variables).
    Assume without loss of generality that $L\equiv H_L(G_L(a,b),c)$ and suppose now that $R$ cannot be written as $H(G(a,b),c)$. It is now straightforward to verify that we land precisely in the setting of Lemma \ref{lemma:second-collapse}, yielding the desired hardness.
    Hence, assume that also $R$ can be written as $H_R(G_R(a,b),c)$.
    It is also easy to see that $G_R\not \sim G_L$, since otherwise we could decompose both $G_R$ and $G_L$ into some $X(a,b)$, implying that $f$ can be written as $H(X(a,b),c)$, yielding a contradiction.
    Hence $G_R\not \sim G_L$.

    We now apply Lemma \ref{lemma:collapsing-non-similar} to collapse them such that either $G_L(a,b) = a\odot b$ (where $\odot$ is an operation that does not appear in $f$) and $G_R(a,b) = b$, or vice versa.
    Without loss of generality, assume that $G_L(a,b) = a\odot b$ and $G_R(a,b) = b$ (as other case is completely symmetric).
    It is straightforward to verify that this collapse preserves the property of $f$ that it cannot be written as $J(I(H(a,b), c), G(a,b))$.
    Moreover, we notice that by assumption that after collapsing, $R$ contains a subexpression that depends on $b,c$.
    Since $L$ contains a subexpression that depends on $a,b$, the result now follows directly by applying Lemma \ref{lemma:4-ap-two-maximal-subtrees}.
\end{proof}
We now show the hardness of the remaining expressions that cannot be written as $J(I(H(a,b), c), G(a,b))$.

\begin{lemma}\label{lemma:4-ap-one-maximal-subtree-triangle}
    Let $f(a,b,c) \equiv L\oplus R$ be an expression such that for no $G,H,I, J$ does the $f(a,b,c)\equiv J(I(H(a,b), c), G(a,b))$ hold.
    Assume further that each subexpression $z$ that depends on exactly two variables satisfies
    ${\rm var}(z) = \{a,b\}$ (without loss of generality, the remaining two cases are symmetric).
    Then any algorithm deciding Constant Term Identity on $f$ in time $\bigO(n^{3-\varepsilon})$ would refute the $4$-AP Hypothesis.
\end{lemma}
\begin{proof}
    We note that if both $L$ and $R$ can be written as $H(G(a,b),c)$ for suitable expressions $G,H$, then we are precisely in the setting of the Lemma \ref{lemma:4-ap-one-maximal-subtree-easy}.
    So assume without loss of generality that $R$ cannot be written in this form. Moreover, as before, we can assume minimality of $f$.
    We consider a case distinction based on the $|{\rm var(L)}|$.

    \textbf{Case 1:} $|{\rm var(L)}| = 1$.
    Without loss of generality, we can assume that ${\rm depth}(L) = 0$ (i.e., that the parse tree of $L$ contains only a single node), otherwise we can set $x\oplus_i x = x$ for each $x\in S$ and then clearly $L(x)$ will always evaluate to $x$.
    By the minimality of $f$, there are expressions $G_R, H_R, I_R, J_R$ such that $R(a,b,c)\equiv J_R(I_R(a,b), H_R(G_R(a,b),c))$. If $I_R\sim G_R$, we can write $R$ as $H(G(a,b),c)$ (for some $G,H$), which would contradict our assumption.
    Also, if $I_R\sim L$, we could write $f$ as $J(I(H(a,b), c), G(a,b))$ which would yield a contradiction.
    Note that we can also assume without loss of generality that ${\rm var}(G_R) = \{a,b\}$.
    Hence, $I_R, G_R$ and $L$ are all pairwise non-similar.
    Now if $I_R$ consists of a single variable (without loss of generality $b$), we can apply a similar approach as in Lemma \ref{lemma:second-collapse} to collapse $R$ to $(b\odot (c \odot (a\odot b))$.
    By assumption that $L\not\sim I_R$, we get that either $L\equiv a$, or $L\equiv c$. In particular, we have thus collapsed $f$ into $a\oplus (b\odot (c \odot (a\odot b))$, or $c\oplus (b\odot (c \odot (a\odot b))$ respectively. We remark that these two expressions are equivalent to $f_1$ and $f_2$ from Lemma \ref{lemma:core-identities} (up to symmetry of $a,b,c$).
    It remains to verify the case when ${\rm var}(I_R) = \{a,b\}$. But using that $I_R\not \sim G_R$, we can apply Lemma \ref{lemma:collapsing-non-similar} to collapse $I_R$ and $G_R$ into $b$ and $a\odot b$ respectively (or vice versa). In either of these two cases, we either collapse to the case when $I_R$ only depends on one variable, and we proceed as above, or the obtained expression contains subexpressions $u, w$ with ${\rm var}(u) = \{a,b\}$ and ${\rm var} (w) = \{b,c\}$, and in this case we can apply Lemma \ref{lemma:4-ap-two-maximal-subtrees}.

    \textbf{Case 2:} $|{\rm var(L)}| = 2$. In particular, ${\rm var(L)} = \{a,b\}$.
    Now if each subexpression $w$ of $R$ with ${\rm var}(w) = \{a,b\}$ satisfies $w\sim L$, we remark that a simple modification of Lemma \ref{lemma:second-collapse} allows us to collapse $f$ into $f_6$, yielding the desired.
    Hence, we can assume that there exists a subexpression $w(a,b)$ of $R$ such that (i) ${\rm var}(w) = \{a,b\}$ (ii) any subexpression $w'$ of $R$ that contains $w$ also depends on $c$ (i.e., $w$ is inclusion maximal) and (iii) $w\not \sim L$.
    By using Lemma \ref{lemma:collapsing-non-similar}, we can collapse $L$ and $w$ such that either $L(a,b)=a\odot b$ and $w(a,b)=b$, or vice versa.
    It is easy to verify that if the original expression cannot be written as $J(I(H(a,b), c), G(a,b))$, then neither can the collapsed expression.
    We note that this in the former case (by maximality of $w$), $R$ contains a subexpression that depends on $b,c$, and we can apply \ref{lemma:4-ap-two-maximal-subtrees}.
    Otherwise, $L$ is collapsed into a single variable, hence we land in \emph{Case 1} which we know how to handle.

     \textbf{Case 3:} $|{\rm var(L)}| = 3$. By minimality of $f$, we can write $R\equiv J_R(I_R(H_R(a,b), c), G_R(a,b))$. Note that by assumption that $f$ contains no subexpression $u$ with $\var{u} = \{a,c\}$ (resp. $\{b,c\}$), we can assume that $\var{H_R} = \{a,b\}$.
    We first assume that $L$ can be written as $L\equiv H_L(G_L(a,b), c)$ where $\var{G_L} = \{a,b\}$.
    Note that if $\var{G_R} = \{a,b\}$, then we know that $G_R \not \sim H_R$ (otherwise $R$ can be written as $J(H_R(a,b),c)$, for some expression $J$ contradicting the assumption of the lemma). Hence by transitivity of $\sim$ relation, we have that either $G_L\not\sim H_R$, or $G_L\not \sim G_R$. In the former case we can apply Lemma \ref{lemma:collapsing-non-similar} to collapse $G_L$ and $H_R$ into $a\odot b$ and $b$ (without loss of generality) respectively, or vice versa. Note that this would yield the conditions of Lemma \ref{lemma:4-ap-two-maximal-subtrees} and we are done.
    In the latter case, we can apply Lemma \ref{lemma:collapsing-non-similar} to collapse $G_L$ and $G_R$ into $a\odot b$ and $b$. If $G_L$ is collapsed into $b$, we are again in case when we can apply Lemma \ref{lemma:4-ap-two-maximal-subtrees}. Hence we can assume that $G_L$ is collapsed into $b$, but notice that this yields the case when $|\var{G_R}| = 1$.
    Moreover, we assume that $H_R\sim G_L$, that is, there exists an expression $X$ with $\var{X} = \{a,b\}$ such that both $G_L$ and $H_R$ can be decomposed into $X$.
    Furthermore, this implies that (without loss of generality) $R\equiv J_R(I_R(X(a,b), c), b)$ and $L\equiv H_L(X(a,b),c)$.
    Furthermore, it is easy to see that if $H_L \sim I_R$, we get a contradiction to the assumption that $f$ cannot be written in the form $J(I(H(a,b),c), G(a,b))$.
    Notice that we can now apply Corollary \ref{cor:collapse}, to collapse $I_R$ and $L$ into $(a\odot b) \odot c$ and $c$ respectively, or vice versa.
    We now consider these two cases separately.
    If $L$ is collapsed into $(a\odot b)\odot c$, and $I_R$ into $c$, we notice that we can then write $f(a,b,c) = ((a\odot b )\odot c) \oplus J_R(b,c)$.
    It is easy to see that we can set the tables of operations in $J_R$ to evaluate to $b\odot c$, yielding (up to symmetry of $a,b,c$ and $L,R$) an expression equivalent to $f_5$ in Lemma \ref{lemma:core-identities}.
    Similarly, if $L$ is collapsed into $c$, we can set up the tables for operations in $J_R$ such that $f$ can be used to evaluate the expression $c\oplus (((a\odot b)\odot c) \odot b)$, which is in a similar fashion equivalent to $f_2$. In both cases we obtain hardness under $4$-AP hypothesis.
    It only remains to verify the case when $L$ cannot be written as $H_L(G_L(a,b),c)$. In particular, using minimality of $f$, we can write $L$ as $J_L(I_L(H_L(a,b), c), G_L(a,b))$.
    Note that we can also assume that $\var{H_L} = \{a,b\}$ and that $H_L\sim H_R$ (otherwise, similar arguments as above conclude the proof).
    Also, if $\var{G_L} = \{a,b\}$, then (assuming we cannot write $L$ as $H(G(a,b),c)$), $G_L\not \sim H_L$ and we can apply Lemma \ref{lemma:collapsing-non-similar} to collapse $H_L$ and $G_L$ into $a\odot b$ and $b$ respectively (note that if $H_L$ is collapsed into $b$, we fall into the setting of Lemma \ref{lemma:4-ap-two-maximal-subtrees}). Hence, we can assume that $|\var{G_L}| = |\var{G_R}| = 1$.
    We now consider the expressions $I_L$ and $I_R$.
    Assume first that $I_R \not \sim I_L$. Then we can apply Corollary \ref{cor:collapse} to collapse $I_L$ and $I_R$ into $(a\odot b)\odot c$, and $c$ respectively (the other case is symmetric).
    Hence, we have shown that we can set up the operations such that $L(a,b,c)$ always evaluates to $J_L((a\odot b)\odot c, G_L(a,b))$ and $R$ always evaluates to $J_R(c, G_R(a,b))$ (recall that both $G_L$ and $G_R$ only depend on one variable each).
    Now, using Lemma \ref{lemma:collapsing-2-vars}, we can collapse $J_L((a\odot b)\odot c, G_L(a,b))$ into $((a\odot b)\odot c) *G_L(a,b)$ and $J_R(G_R(a,b), c)$ into $c*G_R(a,b)$ (or $G_R(a,b)*c$, but this case is symmetric), where $*$ is an operation to be defined in a moment.
    We also remind the reader that we can use the Subexpression Embedding Gadget to make sure that regardless of evaluation, the element $(a\odot b)\odot c)$ is not a member of the (input) set $C$ (the same holds for $A,B$ as well, but for our purpose it suffices to argue the non-membership in $C$).
    We can now define the operation $*$ as follows.
    \[
    x * y = \begin{cases}
        x & \text{if $x\in C$} \\
        x\odot y & \text{otherwise}
    \end{cases}
    \]
    In particular, after plugging in the $*$ operation, we get that $f$ can express $\left(((a\odot b)\odot c) \odot G_L(a,b)\right) \oplus c$,
    which is equivalent to $f_2$ (up to symmetry of $a,b,c$).
    Finally, we assume that $I_L \sim I_R$.
    If also $G_L \sim G_R$, then we get a contradiction, since we could then write $f$ as $J(I(H(a,b),c), G(a,b))$.
    Recall that we assumed that both $G_L$ and $G_R$ depend only on one variable, which means that $\var{G_L}\ne \var{G_R}$. Without loss of generality, we assume that $\var{G_L} = \{a\}$, $\var{G_R} = \{b\}$ (other case is symmetric).
    We apply two iterations of Lemma \ref{lemma:collapsing-2-vars}, in order to collapse both $I_L$ and $I_R$ into $(a\odot b) \odot c$.
    In particular, we get that $L$ can express $J_L((a\odot b) \odot c, a)$ and $R$ expresses $J_R((a\odot b) \odot c, b)$.
    By applying Lemma \ref{lemma:collapsing-2-vars} again, similarly as above, we can collapse $L(a,b,c)$ into $((a\odot b) \odot c) *' a$, and $R(a,b,c)$ into $((a\odot b) \odot c) *' b$, where $*'$ is defined as follows.
    \[
        x*' y = \begin{cases}
            y & \text{if $y\in B$}\\
            x\odot y & \text{otherwise}
        \end{cases}
    \]
    In particular, this gives us that $f$ can express $(((a\odot b) \odot c) \odot a) \oplus b$, which is (up to symmetry of $a,b,c$) equivalent to $f_2$.
\end{proof}
\begin{proposition}[$4$-AP Hardness]
    If $f$ cannot be written as $f(a,b,c)\equiv J(I(H(a,b),c), G(a,b))$, for no $\varepsilon>0$ is there an algorithm solving Constant Term Identity problem on $f$ in $\bigO(n^{3-\varepsilon})$, unless $4$-AP Hypothesis fails.
\end{proposition}
\begin{proof}
    Follows directly from Lemmas \ref{lemma:4-ap-two-maximal-subtrees}, \ref{lemma:4-ap-one-maximal-subtree-easy}, \ref{lemma:4-ap-one-maximal-subtree-triangle}.
\end{proof}
\paragraph{Extending to General Algebraic Identities}
We now show how one can extend the lower bounds above for the Constant Term Identity to a general class of algebraic identities that we call \emph{natural identities}.
Recall that a pair of expressions $f(a,b,c)$, $g(a,b,c)$ is \emph{natural} if neither $f$ nor $g$ is a subexpression of the other.
We say that the problem of checking if $f(a,b,c)=g(a,b,c)$ for each $a,b,c$ is \emph{natural} if the expressions $f,g$ form a natural pair.
\begin{lemma}
    Let $f(a,b,c)$, $g(a,b,c)$ be a natural pair of expressions. Then if we can check for any given algebraic structure $(S,\oplus_1,\dots, \oplus_p)$ if the identity $f(a,b,c) = g(a,b,c)$ is satisfied for each $a,b,c\in S$ in time $T(n)$, we can also solve the Constant Term Identity Checking for both $f$ and $g$ in $\bigO(T(n))$.
\end{lemma}
\begin{proof}
    Assuming that $g$ is not a subexpression of $f$, we can simply apply the Subexpression Embedding Gadget on $f$.
    More precisely, let $\mathcal T_f$ be the set of all subexpressions of $f$.
    Given an algebraic structure $(S,\oplus_1,\dots, \oplus_p)$, we construct $(\tilde S, \tilde\oplus_1,\dots, \tilde\oplus_p)$ as follows.
    Let $\tilde S = S\times \mathcal{T}_f\cup \{\infty\}$ and define each $\tilde \oplus_i$ as follows.
    \[
    (x,t_1) \tilde\oplus_i (y, t_2) = \begin{cases}
        (x\oplus_i y, t_1 \oplus_i t_2) & \text{if $x\oplus_i y\ne \infty$ and $\expr{t_1\oplus_i t_2} \in \mathcal T_f$}\\
        \infty & \text{otherwise}
    \end{cases}
    \]
    Let $\tilde f, \tilde g$ be the expressions obtained from $f,g$ by replacing each $\oplus_i$ with $\tilde \oplus_i$.
    It is straightforward to verify (similar approach as in Observation \ref{obs:chromaticity}) that $\tilde g(\tilde x, \tilde y,\tilde z) = \infty$ for every $\tilde x,\tilde y,\tilde z \in \tilde S$ and moreover that there exists a triple $\tilde x, \tilde y,\tilde z\in \tilde S$ such that $\tilde f(\tilde x, \tilde y,\tilde z)\ne \infty$ if and only if there exists a triple $x,y,z\in S$ such that $f(x,y,z) \ne \infty $. 
    This concludes a reduction from Constant Term Identity Checking on $f$ to Identity Checking on $f,g$.
    The reduction from Constant Term Identity Checking on $g$ is analogous.
\end{proof}

\subsection{Algorithms}\label{se:algorithms-main}

In this subsection we show how to achieve the running times promised in Theorem~\ref{thm:trichotomy}. Observe that we need to provide the following three algorithms for the Constant Term Identity problem:
\begin{enumerate}[label=(\roman*)]
    \item\label{item:alg_quadratic_outer} $\bigO(n^2)$-time approach for the case when $f(a,b,c)\equiv H(G(a,b),c)$,
    \item\label{item:alg_omega_outer} $\bigO(n^{\omega})$-time randomized approach for the case when $f(a,b,c) \equiv J(I(H(a,b),c), G(a,b))$,
    \item\label{item:alg_cubic_outer} $\bigO(n^3)$-time approach for all the other cases.
\end{enumerate}
Clearly, we can naively verify all identities over 3 variables in $\bigO(n^3)$ time by enumerating all triples $(a,b,c)\in S^3$ and comparing $f(a,b,c)$ with $\infty$, which satisfies case~\ref{item:alg_cubic_outer}.
For case~\ref{item:alg_quadratic_outer} we proceed in two steps. First, we calculate the set $\mathcal{G}: = \{G(a,b): a,b\in S\}$ and then verify if $\{H(g,c): g\in \mathcal{G}\} = \{\infty\}$, both in quadratic time.

However, the second case~\ref{item:alg_omega_outer} is more involved and we need to apply triangle counting and verification of polynomial equality, as we did in Section~\ref{sec:distributivity_alg}. We start with the following technical lemma:
\begin{lemma}\label{le:evaluate_polynomials}
    Given an expression $f(a,b,c)$, set $S$ of $n$ elements and $4n$ values: $\bar x= \{x_s\}_{s\in S},\bar y = \{y_s\}_{s\in S},\bar z= \{z_s\}_{s\in S},\bar w =\{w_s\}_{s\in S}$, we can calculate $P_f(\bar x, \bar y, \bar z, \bar w):= \sum_{a,b,c\in S} x_ay_bz_cw_{f(a,b,c)}$ in:
    \begin{enumerate}[label=(\roman*)]
    \item\label{item:alg_quadratic} $\bigO(n^2)$-time when $f(a,b,c)\equiv H(G(a,b),c)$,
    \item\label{item:alg_omega} $\bigO(n^{\omega})$-time when $f(a,b,c) \equiv J(I(H(a,b),c), G(a,b))$,
    \item\label{item:alg_cubic} $\bigO(n^3)$-time otherwise.
    \end{enumerate}
    We denote by $\texttt{complexity}(f)$ the complexity $\bigO(n^2),\bigO(n^\omega)$ or $\bigO(n^3)$ to calculate $P_f(\bar x, \bar y, \bar z, \bar w)$, according to the above cases.
\end{lemma}
\begin{proof}
As the variables $\bar x, \bar y, \bar z, \bar w$ are fixed, we abbreviately write $P_f:=P_f(\bar x, \bar y, \bar z, \bar w)$
First, we calculate $P_f$ in $\bigO(n^3)$ time from the definition, which satisfies case~\ref{item:alg_cubic}.
For case~\ref{item:alg_quadratic} we start with calculating $n$ values $E_t:= \sum_{a,b\in S : G(a,b)=t} x_ay_b$ in total $\bigO(n^2)$ time, as every pair $(a,b)\in S^2$ appears in exactly one term of the sum. Then we can rearrange the terms in $P_f$ and obtain:

\[P_f=\sum_{a,b,c\in S}x_a\cdot y_b\cdot z_c\cdot w_{f(a,b,c)} =\sum_{t\in S}\left( \sum_{\substack{a,b\in S\\G(a,b)=t}}x_a\cdot y_b \right)\sum_{c\in S}z_c \cdot w_{H(t,c)}  = \sum_{t,c\in S}E_t\cdot z_c \cdot w_{H(t,c)}\]
which is calculated in $\bigO(n^2)$ time and case~\ref{item:alg_quadratic} follows.

For the main case ~\ref{item:alg_omega}, we define an edge-weighted tripartite graph $(P,Q,R)$, where $P,Q,R$ are vertex sets each of size $n$ and labeled by elements of $S$. Let $p_s\in P$ denote the vertex in $P$ labeled by $s\in S$, and similarly define vertices $q_s\in Q,r_s\in R$. Now we define the edges in the graph (note that we allow parallel edges):
\begin{itemize}
    \item For every pair of $a\in S,b\in S$, add an edge between $p_{H(a,b)}\in P$ and $q_{G(a,b)}\in Q$ with weight $x_ay_b$.
    \item For every pair of $h\in S,c\in S$, add an edge between $p_{h}\in P$ and $r_{I(h,c)}\in R$ with weight $z_c$.
    \item For every pair of $g\in S, i\in S$, add an edge between $q_{g}\in Q$ and $r_i\in R$ with weight $w_{J(i,g)}$.
\end{itemize}

Three edges with one of each type described above form a triangle if and only if the vertices coincide, i.e., $h=H(a,b), g=G(a,b),i=I(h,c)$. Then, the weight of this triangle (namely, the product of the weights of its three edges) equals
\[(x_ay_b)\cdot z_c \cdot w_{J(i,g)} = x_a\cdot y_b \cdot z_c \cdot w_{J(i,g)}\]
The triangles are in bijection with the triples $(a,b,c)\in S\times S\times S$. Therefore, $P_f$ equals the sum of the weights of all triangles in the graph.

It remains to compute the total weight of all triangles in the graph $(P,Q,R)$.  For a pair of vertices $p_i\in P,q_j\in Q$, let $W^{PQ}_{i,j}$ denote the total weight of edges connecting $p_i,q_j$ (there could be zero, one, or multiple such edges). Similarly define $W^{QR}_{j,k}$ and $W^{RP}_{k,i}$. Then, clearly, the total weight of all triangles equals the trace of the matrix product $W^{PQ}W^{QR}W^{RP}$, which can be computed in $O(n^{\omega})$ time.
\end{proof}

Finally, we show how to use the above lemma to verify if an arbitrary identity $f(a,b,c)=g(a,b,c)$ holds for all $a,b,c\in S$. In particular, this will provide the missing $\bigO(n^\omega)$-time approach for case~\ref{item:alg_omega_outer} of Theorem~\ref{thm:trichotomy}.

\begin{theorem}
    Consider two algebraic expressions $f,g$ over a set $S$ of $n$ elements.
    There exists an algorithm that verifies if for all $a,b,c\in S$ it holds $f(a,b,c)=g(a,b,c)$ that runs in $\max(\texttt{complexity}(f),\texttt{complexity}(g))$-time and succeeds with high probability.
\end{theorem}
\begin{proof}
Consider two polynomials $P_f,P_g$ on $4n$ variables, $\bar x= \{x_s\}_{s\in S},\bar y = \{y_s\}_{s\in S},\bar z= \{z_s\}_{s\in S},\bar w =\{w_s\}_{s\in S}$ where: $P_e(\bar x, \bar y, \bar z, \bar w):= \sum_{a,b,c\in S} x_a\cdot y_b\cdot z_c\cdot w_{e(a,b,c)}$
Observe that, $P_f-P_g =\sum_{a,b,c\in S} x_a\cdot y_b\cdot z_c\cdot (w_{f(a,b,c)}- w_{g(a,b,c)})$ is the zero polynomial if and only if $f(a,b,c) = g(a,b,c)$ holds for all $a,b,c\in S$.
Therefore, our goal reduces to polynomial identity testing on $P_f$ and $P_g$. We pick a prime field $\mathbb{F}_p$ (for sufficiently large $p$) and assign independently uniformly random values from $\mathbb{F}_p$  to the variables $\{x_s\}_{s\in S},\{y_s\}_{s\in S},\{z_s\}_{s\in S},\{w_s\}_{s\in S}$, and then evaluate $P_f-P_g$ on these values using Lemma~\ref{le:evaluate_polynomials}. If the result is non-zero, we return NO; otherwise we return YES. Since $P_f-P_g$ has degree four, by Schwartz--Zippel lemma, the probability of incorrectly reporting YES is at most $4/p$.

Now we choose $p=\Omega(n)$. As a single assignment of variables $\bar x, \bar y, \bar z, \bar w$ from $\mathbb{F}_p$ might give us false equality between non-equal polynomials, we need to repeat the above procedure $O(1)$ times in order to boost success probability.
\end{proof}

\section{Strong Zero Triangle-Hard Problems}\label{sec:strongzerotriangle}
\zerotrianglehard*
\begin{proof}
    Given graph $G$ on $n$ vertices $V$ with edge weights $w: E \to [-n,n]$ we define set $S$ of size $7n+4$ with the following meaning of its elements: $S=\{\infty,\Delta\}\cup V \cup ([-n,n]\times 1)\cup ([-2n,2n]\times 2)$ (i.e., $\infty$ maps to 0, $\Delta$ to 1 etc.).
    Now we need to provide Cayley tables for operations $\odot$ and~$\oplus$.
    Intuitively, $\odot$ returns the weight of an edge between two nodes  (if exists) while $\oplus$ adds the weights and the expression will almost always evaluate to $\infty$, with the only exception if there exists a zero-weight triangle in $G$.
    More precisely, we define the operations in the following way:
    \[
    u \odot v = \begin{cases}
        (w(\{u,v\}),1) & \text{if } u,v\in V \text{ and } \{u,v\}\in E \\
        \infty & \text{otherwise}
        \end{cases}
    \]
    \begin{align*}
    (a,1) \oplus (b,1) &= (a+b,2)  &\text{for }a,b\in[-n,n]\\
    (c,2) \oplus (b,1) &= \begin{cases}
        \infty &\text{if } c+b\ne 0 \\
        \Delta &\text{if } c+b=0
    \end{cases}&\text{for }c\in[-2n,2n],b\in[-n,n]
    \end{align*}
and all unspecified operations always evaluate to $\infty$.
Now, in order to prove the theorem, we need to show that
\[
\left((a\odot b)\oplus(a\odot c)\right)\oplus(b\odot c)=\begin{cases}
    \Delta &\text{if } a,b,c\in V \text{ and they form a zero-triangle in }G\\
    \infty &\text{otherwise}
\end{cases}
\]
Indeed, if any of $a,b,c$ is not from $V$, then $\odot$ already returns $\infty$ which is preserved with $\oplus$ ($\infty \oplus x = x \oplus \infty = \infty$ for all $x\in S$).
Otherwise, $\odot$ verifies if there is an edge between two nodes and, if yes, returns its weight with 1 at the second coordinate.
Then $a\odot b,a\odot c,b\odot c$ are all different from $\infty$ iff $a,b,c$ form a triangle in $G$.
When we add two weights, we obtain an element corresponding to their sum (with 2 at the second coordinate): $(a,1) \oplus (b,1) = (a+b,2)$.
Finally, in the last step, while adding sum of two weights to one weight we only verify if they sum up to 0 or not and evaluate to $\Delta$ or $\infty$ respectively.
\end{proof}

By extending the above idea, we can show that counting triples of elements for which given operations are distributive is also Strong Zero Triangle-hard:

\DistCountingTheorem*

Before proving the theorem above, let us first construct the relevant reduction from Zero Triangle Detection.
Given a weighted graph $G = (V,E,w)$, with weight function $w:E\to [-n,n]$, set $S := V \cup \big([-10n,10n]\times\{0,1,2\}\big) \cup \{\infty\}$.
Now we describe how to set the Cayley tables for operations $\odot$ and $\oplus$.
Intuitively, $\oplus$ or $\odot$ applied to two elements from $S$ that correspond to vertices in $G$ will give the weight of the edge between them and we will also need to handle the cases of operations between different "types" of elements. Our goal is to set the Cayley tables in such a way that for any triple of vertices $u,v,w$, the identity $u\odot(v\oplus w) = (u\odot v) \oplus (u \odot w)$ evaluates to $-w(v,w)=w(u,v)+w(u,w)$ which is satisfied if and only if they form a triangle of weight zero in $G$.
On the other hand, we also exploit the structure of the identity we want to verify, to make sure that if we plug in any element that does not correspond to a vertex, the identity is satisfied.
We achieve this by introducing the second coordinate to each weight ($0,1$ or $2$) that allows us, at each state, to keep track of the operations that were performed before reaching the state and if at least one of the three elements $x,y,z$ is in $S\setminus V$, we get $x\odot (y\oplus z) = \infty = (x\odot y) \oplus (x\odot z)$.

Note that without loss of generality we may assume that the given graph $G$ is complete, with self-loops. Indeed, by setting the weight of each non-edge to $3n+1$, we obtain an equivalent instance with only a constant blow-up in edge weights.
Now, for each pair of vertices $u,v\in V$, define:
\begin{align*}
    u\oplus v &= (w(u,v), 1) \\
    u\odot v &= (w(u,v), 2).
\end{align*}
For any two integers $a,b$, we define the following operation:
\[
a\dagger b =
\begin{cases}
a+b & \text{if $|a+b|\leq 10n$}\\
10n & \text{otherwise}.
\end{cases} \\
\]
For any $w_1,w_2\in [-10n, 10n]$, we set:
\[
    (w_1, 2) \oplus (w_2, 2) = (w_1\dagger w_2, 0)\\
\]
Moreover, for any $u\in V$, $w\in [-10n, 10n]$, we set:
\begin{align*}
    u\odot (w, 1) = (w,1) \odot u = (-w, 0).
\end{align*}
Finally, for any $x,y\in S$ such that $x\oplus y$ (resp. $x\odot y$) is not specified above, set $x\oplus y = \infty$ (resp. $x\odot y = \infty$). In particular, if at least one operand is $\infty$ then the result is always $\infty$, e.g. $x\oplus \infty = \infty$ for all $x\in S$.

We now proceed to prove some properties of the instance constructed above that will imply Theorem \ref{th:dist-counting}.
\begin{lemma}\label{lemma:counting-dist}
    Let $G= (V,E,w)$ be a complete weighted graph with weight function $w:E\to[-n,n]$. Define the set $S$, and the two binary operations $\oplus, \odot: S\times S\to S$ as above. Then the following properties are satisfied.
    \begin{enumerate}
        \item The operations $\oplus$ and $\odot$ are both commutative and associative.
        \item For any triple $x,y,z$ such that $\{x,y,z\}\cap (S\setminus V) \neq \emptyset$, the identity $x\odot(y\oplus z) = \infty = (x\odot y) \oplus (x\odot z)$.
        \item For any triple of vertices $a,b,c\in V$, the identity $a\odot(b\oplus c) = (a\odot b) \oplus (a \odot c)$ is satisfied if and only if $a,b,c$ form a triangle of weight zero in $G$.
    \end{enumerate}
\end{lemma}
\begin{proof}
    The commutativity is easy to see from the definition of the operations and the associativity follows by observing that for any triple $a,b,c\in S$ the following holds: \[a\oplus (b\oplus c) = (a\oplus b)\oplus c = \infty = a\odot (b\odot c) = (a\odot b)\odot c.\]
    We now verify the second property by checking all the cases separately. First, if one of the elements is $\infty$, both sides of the identity evaluate to $\infty$:
    \begin{claim}
        If $\{x,y,z\}\cap\{\infty\}\neq \emptyset$, then $x\odot(y\oplus z) = \infty = (x\odot y) \oplus (x\odot z)$.
    \end{claim}
    Similarly, since for each $w\in [-10n,10n]$ and any $x\in S$, by definition we have that $(w,0)\oplus x = (w,0)\odot x = \infty$, we also get the second case for free.
    \begin{claim}
        If $\{x,y,z\}\cap \big([-10n,10n]\times \{0\}\big)\neq \emptyset$, then $x\odot(y\oplus z) = \infty = (x\odot y) \oplus (x\odot z)$.
    \end{claim}
    We now verify the remaining two cases more carefully.
    \begin{claim}
        If $\{x,y,z\}\cap \big([-10n,10n]\times \{1\}\big)\neq \emptyset$, then $x\odot(y\oplus z) = \infty = (x\odot y) \oplus (x\odot z)$.
    \end{claim}
    \begin{subproof}
        Assume first that $x\in ([-10n,10n]\times \{1\})$. Then by definition of our operations, we have $x\odot(y\oplus z)\neq \infty$ if and only if $y\oplus z \in V$. But for no pair $y,z\in S$ does $y\oplus z \in V$ hold.
        Using commutativity of $\oplus$, it only remains to check the case when $y\in ([-10n,10n]\times \{1\})$. But then $y\oplus z = \infty$ and hence if $\{x,y,z\}\cap \big([-10n,10n]\times \{1\}\big)\neq \emptyset$, then $x\odot(y\oplus z) = \infty$ always holds. 
        Consider now the right-hand side of the identity. We observe that if $\{a,b\}\cap ([-10n,10n]\times \{1\})$ is nonempty, then either $a\odot b = \infty$, or $a\odot b \in [-10n,10n]\times \{0\}$. 
        In both cases, $a\odot b \oplus c = c \oplus a\odot b = \infty$ for any $c\in S$, and hence we can conclude that also whenever $\{x,y,z\}\cap \big([-10n,10n]\times \{1\}\big)\neq \emptyset$, the identity $(x\odot y) \oplus (x\odot z) = \infty$ holds.
    \end{subproof}
     \begin{claim}
        If $\{x,y,z\}\cap \big([-10n,10n]\times \{2\}\big)\neq \emptyset$, then $x\odot(y\oplus z) = \infty = (x\odot y) \oplus (x\odot z)$.
    \end{claim}
    \begin{subproof}
        We first note that for any $a\in ([-10n,10n]\times \{2\})$ and $b\in S$ we have that $a\odot b = \infty$. Hence, clearly $(x\odot y) \oplus (x\odot z) = \infty$. This also implies that if $x\in ([-10n,10n]\times \{2\})$, similarly $x\odot(y\oplus z) = \infty$.
        On the other hand, if $\{y,z\}\cap \big([-10n,10n]\times \{2\}\big)\neq \emptyset$, then by definition of $\oplus$, we have that $y\oplus z\neq \infty$ if and only if $\{y,z\}\subset \big([-10n,10n]\times \{2\}\big)$. But then $y\oplus z \in \big([-10n,10n]\times \{0\}\big)$ and hence $x\odot(y\oplus z) = \infty$ as desired.
    \end{subproof}
    The previous four claims now imply the second property, hence it only remains to prove the third, that is that for any triple of vertices $a,b,c\in V$, the identity $a\odot(b\oplus c) = (a\odot b) \oplus (a \odot c)$ is satisfied if and only if $a,b,c$ form a triangle of weight zero in $G$.
    To prove this, we will simply compute both sides of the identity.
    \begin{align*}
        a\odot(b\oplus c) &= a \odot (w(b,c), 1) \\
        & = (-w(b,c), 0).
    \end{align*}
    \begin{align*}
        (a\odot b) \oplus (a \odot c) &= (w(a,b), 2) \oplus (w(a,c), 2)\\
        & = (w(a,b)\dagger w(a,c), 0)\\
        & = (w(a,b) + w(a,c), 0) & (|w(a,b) + w(a,c)|<10n).
    \end{align*}
    Thus $a\odot(b\oplus c) = (a\odot b) \oplus (a \odot c)$ is satisfied if and only if $w(a,b) + w(a,c) = -w(b,c)$, that is if $a,b,c$ form a triangle of weight zero.
\end{proof}
Assume now that there exists an algorithm $\mathcal{A}$ that given a structure $(S,\oplus, \odot)$ counts the number of non-distributive triples in $S$ in time $\bigO(n^{3-\varepsilon})$ for some $\varepsilon>0$. Given a weighted graph $G=(V,E,w)$ construct the structure $(S,\oplus, \odot)$ as above and run the algorithm $\mathcal A$ on it. By Lemma \ref{lemma:counting-dist}, the only non-distributive triples come from nodes that do not form zero triangle, so the number of non-distributive triples equals $|V|^3$ if and only if $G$ contains no triangle of weight $0$. Theorem \ref{th:dist-counting} now follows.

\section{Field and Ring Verification}
\label{sec:ring-field}
In this section we construct efficient algorithms for checking if a given structure forms a field or a ring respectively. We start by proving the following theorem for fields.

\FieldVerificationTheorem*

\begin{lemma} \label{lemma:basis-construction} \cite{groupverifictaion2023}
    Given a set $G$ of $n$ elements and a binary operation $+$ on $G$ such that $(G,+)$ forms a group, we can in time $\bigO(n^2)$ construct a basis of $(G,+)$, that is a set $S\subset G$ satisfying the following two properties:
    \begin{itemize}
        \item $|S| = \bigO(\sqrt{n})$.
        \item For every $a\in G$ there exists a pair $\alpha, \beta\in S$ such that $a = \alpha + \beta$.
    \end{itemize}
\end{lemma}
\begin{lemma}\label{lemma:b-in-Sa-distributivity}
    Let $(G,+)$ be a group consisting of $n$ elements and $\cdot$ be a binary operation on $G$. Let $S$ be a basis of $G$. Then the following are equivalent:
    \begin{enumerate}
        \item \label{item:dist-in-g} $a\cdot (b+c) = a\cdot b + a\cdot c$ for every $a,b,c\in G$.
        \item \label{item:dist-b-in-s} $x\cdot (y+z) = x\cdot y + x\cdot z$ for every $x,z\in G$, $y\in S$.
    \end{enumerate}
\end{lemma}
\begin{proof}
    It is clear that \ref{item:dist-in-g} implies \ref{item:dist-b-in-s}, so we only prove the converse.
    Assume that the left distributive law holds for every $x,z \in G$ and $y\in S$.
    Let $a,b,c\in G$ be arbitrary. Since $S$ is a basis of $G$, we can write $b = \beta + \beta'$ for some $\beta, \beta'\in S$.
    Plugging this in, we obtain
    \begin{align*}
        a\cdot (b+c) &= a\cdot ((\beta + \beta') + c) & \\
        & = a\cdot (\beta + (\beta' + c)) & \text{(associativity of $(G,+)$)} \\
        & = a\cdot \beta + a\cdot(\beta' + c) & (a, (\beta' + c) \in G, \beta \in S) \\
        & = a\cdot \beta + (a\cdot\beta' + a\cdot c) & (a, c \in G, \beta' \in S) \\
        & = (a\cdot \beta + a\cdot\beta') + a\cdot c & \text{(associativity of $(G,+)$)} \\
        & = a\cdot(\beta + \beta') + a\cdot c & (a,\beta'\in G, \beta \in S) \\
        & = a\cdot b + a\cdot c
    \end{align*}
    as desired.
\end{proof}

\begin{lemma}\label{lemma:a-in-Sm-dist}
    Let $G$ be a set consisting of $n$ elements and $+,\cdot$ be binary operations on $G$, such that $(G,+)$ and $(G,\cdot)$ are both groups. Let $S_a$ be a basis of $(G,+)$ and $S_m$ be a basis of $(G,\cdot)$. Then the following are equivalent:
    \begin{itemize}
        \item $a\cdot (b+c) = a\cdot b + a\cdot c$ for every $a,b,c\in G$.
        \item $x\cdot (y+z) = x\cdot y + x\cdot z$ for every $x\in S_m$, $y\in S_a$, $z\in G$.
    \end{itemize}
\end{lemma}
\begin{proof}
    Again, the first property clearly implies the second one, so we only prove the other direction.
    Let $a,b,c\in G$ be arbitrary. By Lemma \ref{lemma:b-in-Sa-distributivity}, we get to assume $b\in S_a$ for free.
    Moreover, since $S_m$ is a basis of $(G,\cdot)$, we can write $a = \alpha\cdot \alpha'$ for some $\alpha, \alpha'\in S_m$, to obtain the following chain of equalities.
    \begin{align*}
        a\cdot(b+c) &= (\alpha \cdot \alpha')\cdot (b+c) \\
        & = \alpha \cdot (\alpha'\cdot (b+c)) & \text{(associativity of $(G,\cdot)$)}\\
        & = \alpha \cdot (\alpha'\cdot b + \alpha' \cdot c) & (\alpha' \in S_m, b\in S_a, c\in G)
    \end{align*}
    Now, if $\alpha'\cdot b$ is not in $S_a$, we can apply the exact same argument as in the proof of Lemma \ref{lemma:b-in-Sa-distributivity} to break it into two elements $\beta, \beta'\in S_a$, and we can perform the same computation as above, hence without loss of generality, we can assume that $\alpha'\cdot b\in S_a$.
    \begin{align*}
        \alpha \cdot (\alpha'\cdot b + \alpha' \cdot c) &= \alpha \cdot (\alpha'\cdot b) + \alpha \cdot (\alpha'\cdot c) & (\alpha \in S_m, (\alpha'\cdot b)\in S_a, (\alpha' \cdot c)\in G) \\
        & = (\alpha \cdot \alpha')\cdot b + (\alpha \cdot \alpha')\cdot c & \text{(associativity of $(G,\cdot)$)} \\
        & = a\cdot b + a\cdot c
    \end{align*}
\end{proof}
We note that a similar computation can be carried out for the right distributivity law to show that it is always sufficient to assume that $a\in S_m, b\in S_a$.
The lemmas above now give us a straightforward way to verify if a set with two binary operations forms a field in linear time, giving us the proof of Theorem \ref{th:field-verification}
\begin{proof}[Proof (of Theorem \ref{th:field-verification})]
    By \cite{groupverifictaion2023}, we can verify in time $\bigO(n^2)$ if $(F,\oplus)$ and $(F\setminus\{0\},\odot)$ (where $0$ is the additive identity) form groups. If any of the two conditions is not satisfied, we reject, otherwise we proceed to verify commutativity and distributivity. Commutativity of both operations can clearly be verified in $\bigO(n^2)$ by exhaustive search.
    Since $(F,\oplus)$ and $(F\setminus\{0\},\odot)$ are both groups, we can construct their respective bases $S_a$ and $S_m$ of size $\bigO(\sqrt{n})$ in $\bigO(n^2)$ by Lemma \ref{lemma:basis-construction}.
    The only remaining property to verify is the distributivity. Note that for all triples $a,b,c$ such that (without loss of generality) $a=0$, distributivity can be checked in $\bigO(n^2)$ by a simple exhaustive search.
    Hence, by Lemma \ref{lemma:a-in-Sm-dist} it remains to verify that the distributivity holds for all triples $a\in S_m, b\in S_a, c\in F$. This can be done by a simple brute-force approach in $\bigO(n^2)$, as desired.
\end{proof}
The approach above clearly fails for ring verification, since $(R,\cdot)$ is not necessarily a group. We nevertheless prove a weaker version of Lemma \ref{lemma:a-in-Sm-dist}, that does not assume multiplicative inverse, but in turn cannot be used to separately verify left and right distributivity, but rather checks both identities simultaneously. Luckily, this is sufficient for checking that a structure forms a ring.
\begin{lemma}\label{lemma:ring-distributivity}
    Let $R$ be a set consisting of $n$ elements, together with binary operations $+,\cdot$, such that $(R,+)$ forms an abelian group and $(R,\cdot)$ is associative.
    Let $S$ be a basis of $(R,+)$. Then the two distributive laws
    \begin{enumerate}
        \item $a(b+c) = ab + ac$ for every $a,b,c\in R$
        \item $(b+c)a = ba + ca$ for every $a,b,c \in R$
    \end{enumerate}
    hold if and only if each of the four "restricted" distributivity laws hold:
    \begin{enumerate}[label=\roman*.]
        \item \label{item:1} $a(b+c) = ab + ac$ for every $a,b \in S ,c\in R$
        \item \label{item:3} $a(b+c) = ab + ac$ for every $b,c\in S, a\in R$
        \item \label{item:4} $(b+c)a = ba + ca$ for every $a,b \in S,c \in R$
        \item \label{item:6} $(b+c)a = ba + ca$ for every $b,c\in S, a \in R$
    \end{enumerate}
\end{lemma}
\begin{proof}
    One direction is trivial, so we prove only that assuming the four restricted distributivity laws, the left distributivity law is satisfied (the proof for the right distributivity is completely symmetric).
    Let $a,b,c\in R$ be arbitrary. We write $a$ as $\alpha + \alpha'$ for some $\alpha, \alpha'\in S$. Furthermore, by Lemma \ref{lemma:b-in-Sa-distributivity} we can assume that $b\in S$.
    We then compute
    \begin{align*}
        a\cdot (b+c) & = (\alpha + \alpha')\cdot(b+c) \\
        & = (\alpha\cdot(b+c)) + (\alpha'\cdot (b+c)) & \text{(\ref{item:6})} \\
        & = ((\alpha\cdot b)+(\alpha \cdot c)) + ((\alpha'\cdot b)+(\alpha' \cdot c)) & \text{(\ref{item:1})} \\
        &= ((\alpha\cdot b)+(\alpha' \cdot b)) + ((\alpha \cdot c)+(\alpha' \cdot c))  & \text{(associativity and commutativity of $(R,+)$)}\\
        & = ((\alpha + \alpha')\cdot b) + ((\alpha + \alpha')\cdot c) & \text{(\ref{item:6})}\\
        & = a\cdot b + a\cdot c.
    \end{align*}
\end{proof}
The following theorem follows by checking that $(R,+)$ forms an abelian group, checking the two distributive laws from Lemma \ref{lemma:ring-distributivity}, and verifying the associativity of $(R,\cdot)$ by the algorithm of Rajagopalan and Schulman, which dominates the runtime.
\RingVerificationTheorem*

\section*{Acknowledgements}
M. K. thanks Eamonn May for stimulating discussions during May's work on his (unpublished) Master's thesis, which contains a proof of Corollary~\ref{col:distributivity_lb}.
B. D. was partially supported by the Polish National Science Centre grant number 2023/51/B/ST6/01505. G.~G., M.~K., and M.~R. were partially supported by the Deutsche Forschungsgemeinschaft (DFG, German Research Foundation) – 462679611.
C. J. is supported by the Miller Research Fellowship at the Miller Institute for Basic Research in Science, UC Berkeley.
X.~M. is supported by NSF CAREER Award CCF-2337901.
Part of this research was done during the 2025 EnCORE Workshop on Fine-Grained Complexity at University of California, San Diego.
\bibliographystyle{plainurl}
\bibliography{refs}

@inproceedings{groupverifictaion2023,
  author       = {Shai Evra and
                  Shay Gadot and
                  Ohad Klein and
                  Ilan Komargodski},
  title        = {Verifying Groups in Linear Time},
  booktitle    = {65th {IEEE} Annual Symposium on Foundations of Computer Science, {FOCS}
                  2024, Chicago, IL, USA, October 27-30, 2024},
  pages        = {2131--2147},
  publisher    = {{IEEE}},
  year         = {2024},
  url          = {https://doi.org/10.1109/FOCS61266.2024.00126},
  doi          = {10.1109/FOCS61266.2024.00126},
  bibsource    = {dblp computer science bibliography, https://dblp.org}
}

@article{rajagopalan00,
author = {Rajagopalan, Sridhar and Schulman, Leonard J.},
title = {Verification of Identities},
journal = {SIAM Journal on Computing},
volume = {29},
number = {4},
pages = {1155-1163},
year = {2000},
doi = {10.1137/S0097539797325387},
URL = {https://doi.org/10.1137/S0097539797325387},
OPTeprint = {https://doi.org/10.1137/S0097539797325387},
note= {Conference version at FOCS 1996.}
}

@inproceedings{AbboudFS24,
  author       = {Amir Abboud and
                  Nick Fischer and
                  Yarin Shechter},
  editor       = {Jos{\'{e}} A. Soto and
                  Andreas Wiese},
  title        = {Faster Combinatorial k-Clique Algorithms},
  booktitle    = {{LATIN} 2024: Theoretical Informatics - 16th Latin American Symposium,
                  Puerto Varas, Chile, March 18-22, 2024, Proceedings, Part {I}},
  series       = {Lecture Notes in Computer Science},
  volume       = {14578},
  pages        = {193--206},
  publisher    = {Springer},
  year         = {2024},
  url          = {https://doi.org/10.1007/978-3-031-55598-5\_13},
  doi          = {10.1007/978-3-031-55598-5\_13},
  bibsource    = {dblp computer science bibliography, https://dblp.org}
}

@inproceedings{AlmanDWXXZ25,
  author       = {Josh Alman and
                  Ran Duan and
                  Virginia {Vassilevska Williams} and
                  Yinzhan Xu and
                  Zixuan Xu and
                  Renfei Zhou},
  editor       = {Yossi Azar and
                  Debmalya Panigrahi},
  title        = {More Asymmetry Yields Faster Matrix Multiplication},
  booktitle    = {Proceedings of the 2025 Annual {ACM-SIAM} Symposium on Discrete Algorithms,
                  {SODA} 2025, New Orleans, LA, USA, January 12-15, 2025},
  pages        = {2005--2039},
  publisher    = {{SIAM}},
  year         = {2025},
  url          = {https://doi.org/10.1137/1.9781611978322.63},
  doi          = {10.1137/1.9781611978322.63},
  bibsource    = {dblp computer science bibliography, https://dblp.org}
}

@inproceedings{Freivalds1977,
  author       = {Rusins Freivalds},
  editor       = {Jir{\'{\i}} Becv{\'{a}}r},
  title        = {Fast Probabilistic Algorithms},
  booktitle    = {Mathematical Foundations of Computer Science 1979, Proceedings, 8th
                  Symposium, Olomouc, Czechoslovakia, September 3-7, 1979},
  series       = {Lecture Notes in Computer Science},
  volume       = {74},
  pages        = {57--69},
  publisher    = {Springer},
  year         = {1979},
  url          = {https://doi.org/10.1007/3-540-09526-8\_5},
  doi          = {10.1007/3-540-09526-8\_5},
  bibsource    = {dblp computer science bibliography, https://dblp.org}
}

@article{Behrend1946,
  title={On sets of integers which contain no three terms in arithmetical progression},
  author={Behrend, Felix A},
  journal={Proceedings of the National Academy of Sciences},
  volume={32},
  number={12},
  pages={331--332},
  year={1946}
}

@misc{Polak_personal_communication,
author = {Adam Polak},
note  = {Personal communication},
year = {2025}
}

@article{Clifford1961,
  title={The algebraic theory of semigroups, vol. 1},
  author={Clifford, Alfred H and Preston, Gordon B},
  journal={American Mathematical Society, Providence},
  volume={2},
  number={7},
  year={1961}
}

@inproceedings{AbboudW14,
  author       = {Amir Abboud and
                  Virginia {Vassilevska Williams}},
  title        = {Popular Conjectures Imply Strong Lower Bounds for Dynamic Problems},
  booktitle    = {55th {IEEE} Annual Symposium on Foundations of Computer Science, {FOCS}
                  2014, Philadelphia, PA, USA, October 18-21, 2014},
  pages        = {434--443},
  publisher    = {{IEEE} Computer Society},
  year         = {2014},
  url          = {https://doi.org/10.1109/FOCS.2014.53},
  doi          = {10.1109/FOCS.2014.53},
  bibsource    = {dblp computer science bibliography, https://dblp.org}
}

@inproceedings{AbboudBKZ22,
  author       = {Amir Abboud and
                  Karl Bringmann and
                  Seri Khoury and
                  Or Zamir},
  editor       = {Stefano Leonardi and
                  Anupam Gupta},
  title        = {Hardness of approximation in p via short cycle removal: cycle detection,
                  distance oracles, and beyond},
  booktitle    = {{STOC} '22: 54th Annual {ACM} {SIGACT} Symposium on Theory of Computing,
                  Rome, Italy, June 20 - 24, 2022},
  pages        = {1487--1500},
  publisher    = {{ACM}},
  year         = {2022},
  url          = {https://doi.org/10.1145/3519935.3520066},
  doi          = {10.1145/3519935.3520066},
  bibsource    = {dblp computer science bibliography, https://dblp.org}
}

@inproceedings{VassilevskaWY07,
  author       = {Virginia Vassilevska and
                  Ryan Williams and
                  Raphael Yuster},
  editor       = {David S. Johnson and
                  Uriel Feige},
  title        = {All-pairs bottleneck paths for general graphs in truly sub-cubic time},
  booktitle    = {39th Annual {ACM} Symposium on Theory of Computing ({STOC} 2007)},
  pages        = {585--589},
  publisher    = {{ACM}},
  year         = {2007},
  url          = {https://doi.org/10.1145/1250790.1250876},
  doi          = {10.1145/1250790.1250876},
  bibsource    = {dblp computer science bibliography, https://dblp.org}
}

@inproceedings{AbboudBDN18,
  author       = {Amir Abboud and
                  Karl Bringmann and
                  Holger Dell and
                  Jesper Nederlof},
  editor       = {Ilias Diakonikolas and
                  David Kempe and
                  Monika Henzinger},
  title        = {More consequences of falsifying {SETH} and the orthogonal vectors
                  conjecture},
  booktitle    = {Proceedings of the 50th Annual {ACM} {SIGACT} Symposium on Theory
                  of Computing, {STOC} 2018, Los Angeles, CA, USA, June 25-29, 2018},
  pages        = {253--266},
  publisher    = {{ACM}},
  year         = {2018},
  url          = {https://doi.org/10.1145/3188745.3188938},
  doi          = {10.1145/3188745.3188938},
  bibsource    = {dblp computer science bibliography, https://dblp.org}
}

@inproceedings{DudekGS20,
  author       = {Bartlomiej Dudek and
                  Pawel Gawrychowski and
                  Tatiana Starikovskaya},
  editor       = {Konstantin Makarychev and
                  Yury Makarychev and
                  Madhur Tulsiani and
                  Gautam Kamath and
                  Julia Chuzhoy},
  title        = {All non-trivial variants of 3-LDT are equivalent},
  booktitle    = {Proceedings of the 52nd Annual {ACM} {SIGACT} Symposium on Theory
                  of Computing, {STOC} 2020, Chicago, IL, USA, June 22-26, 2020},
  pages        = {974--981},
  publisher    = {{ACM}},
  year         = {2020},
  url          = {https://doi.org/10.1145/3357713.3384275},
  doi          = {10.1145/3357713.3384275},
  bibsource    = {dblp computer science bibliography, https://dblp.org}
}

@inproceedings{JinX23,
  author       = {Ce Jin and
                  Yinzhan Xu},
  editor       = {Barna Saha and
                  Rocco A. Servedio},
  title        = {Removing Additive Structure in 3SUM-Based Reductions},
  booktitle    = {Proceedings of the 55th Annual {ACM} Symposium on Theory of Computing,
                  {STOC} 2023, Orlando, FL, USA, June 20-23, 2023},
  pages        = {405--418},
  publisher    = {{ACM}},
  year         = {2023},
  url          = {https://doi.org/10.1145/3564246.3585157},
  doi          = {10.1145/3564246.3585157},
  bibsource    = {dblp computer science bibliography, https://dblp.org}
}

@inproceedings{Zamir23,
  author       = {Or Zamir},
  editor       = {Barna Saha and
                  Rocco A. Servedio},
  title        = {Algorithmic Applications of Hypergraph and Partition Containers},
  booktitle    = {Proceedings of the 55th Annual {ACM} Symposium on Theory of Computing,
                  {STOC} 2023, Orlando, FL, USA, June 20-23, 2023},
  pages        = {985--998},
  publisher    = {{ACM}},
  year         = {2023},
  url          = {https://doi.org/10.1145/3564246.3585163},
  doi          = {10.1145/3564246.3585163},
  bibsource    = {dblp computer science bibliography, https://dblp.org}
}

@inproceedings{AbboudFKLM24,
  author       = {Amir Abboud and
                  Nick Fischer and
                  Zander Kelley and
                  Shachar Lovett and
                  Raghu Meka},
  editor       = {Bojan Mohar and
                  Igor Shinkar and
                  Ryan O'Donnell},
  title        = {New Graph Decompositions and Combinatorial Boolean Matrix Multiplication
                  Algorithms},
  booktitle    = {Proceedings of the 56th Annual {ACM} Symposium on Theory of Computing,
                  {STOC} 2024, Vancouver, BC, Canada, June 24-28, 2024},
  pages        = {935--943},
  publisher    = {{ACM}},
  year         = {2024},
  url          = {https://doi.org/10.1145/3618260.3649696},
  doi          = {10.1145/3618260.3649696},
  bibsource    = {dblp computer science bibliography, https://dblp.org}
}

@inproceedings{KelleyLM24,
  author       = {Zander Kelley and
                  Shachar Lovett and
                  Raghu Meka},
  editor       = {Bojan Mohar and
                  Igor Shinkar and
                  Ryan O'Donnell},
  title        = {Explicit Separations between Randomized and Deterministic Number-on-Forehead
                  Communication},
  booktitle    = {Proceedings of the 56th Annual {ACM} Symposium on Theory of Computing,
                  {STOC} 2024, Vancouver, BC, Canada, June 24-28, 2024},
  pages        = {1299--1310},
  publisher    = {{ACM}},
  year         = {2024},
  url          = {https://doi.org/10.1145/3618260.3649721},
  doi          = {10.1145/3618260.3649721},
  bibsource    = {dblp computer science bibliography, https://dblp.org}
}

@inproceedings{DalirrooyfardW22,
  author       = {Mina Dalirrooyfard and
                  Virginia {Vassilevska Williams}},
  title        = {Induced Cycles and Paths Are Harder Than You Think},
  booktitle    = {63rd {IEEE} Annual Symposium on Foundations of Computer Science, {FOCS}
                  2022, Denver, CO, USA, October 31 - November 3, 2022},
  pages        = {531--542},
  publisher    = {{IEEE}},
  year         = {2022},
  url          = {https://doi.org/10.1109/FOCS54457.2022.00057},
  doi          = {10.1109/FOCS54457.2022.00057},
  bibsource    = {dblp computer science bibliography, https://dblp.org}
}

@inproceedings{KelleyM23,
  author       = {Zander Kelley and
                  Raghu Meka},
  title        = {Strong Bounds for 3-Progressions},
  booktitle    = {64th {IEEE} Annual Symposium on Foundations of Computer Science, {FOCS}
                  2023, Santa Cruz, CA, USA, November 6-9, 2023},
  pages        = {933--973},
  publisher    = {{IEEE}},
  year         = {2023},
  url          = {https://doi.org/10.1109/FOCS57990.2023.00059},
  doi          = {10.1109/FOCS57990.2023.00059},
  bibsource    = {dblp computer science bibliography, https://dblp.org}
}

@inproceedings{Kunnemann22,
  author       = {Marvin K{\"{u}}nnemann},
  title        = {A tight (non-combinatorial) conditional lower bound for {Klee}'s Measure
                  Problem in {3D}},
  booktitle    = {63rd {IEEE} Annual Symposium on Foundations of Computer Science, {FOCS}
                  2022, Denver, CO, USA, October 31 - November 3, 2022},
  pages        = {555--566},
  publisher    = {{IEEE}},
  year         = {2022},
  url          = {https://doi.org/10.1109/FOCS54457.2022.00059},
  doi          = {10.1109/FOCS54457.2022.00059},
  bibsource    = {dblp computer science bibliography, https://dblp.org}
}

@inproceedings{Chan15,
  author       = {Timothy M. Chan},
  editor       = {Piotr Indyk},
  title        = {Speeding up the Four Russians Algorithm by About One More Logarithmic
                  Factor},
  booktitle    = {26th Annual {ACM-SIAM} Symposium on Discrete Algorithms ({SODA} 2015)},
  pages        = {212--217},
  publisher    = {{SIAM}},
  year         = {2015},
  url          = {https://doi.org/10.1137/1.9781611973730.16},
  doi          = {10.1137/1.9781611973730.16},
  bibsource    = {dblp computer science bibliography, https://dblp.org}
}

@inproceedings{LincolnWW18,
  author       = {Andrea Lincoln and
                  Virginia {Vassilevska Williams} and
                  R. Ryan Williams},
  editor       = {Artur Czumaj},
  title        = {Tight Hardness for Shortest Cycles and Paths in Sparse Graphs},
  booktitle    = {Proceedings of the Twenty-Ninth Annual {ACM-SIAM} Symposium on Discrete
                  Algorithms, {SODA} 2018, New Orleans, LA, USA, January 7-10, 2018},
  pages        = {1236--1252},
  publisher    = {{SIAM}},
  year         = {2018},
  url          = {https://doi.org/10.1137/1.9781611975031.80},
  doi          = {10.1137/1.9781611975031.80},
  bibsource    = {dblp computer science bibliography, https://dblp.org}
}

@inproceedings{KunnemannN22,
  author       = {Marvin K{\"{u}}nnemann and
                  Andr{\'{e}} Nusser},
  editor       = {Joseph (Seffi) Naor and
                  Niv Buchbinder},
  title        = {Polygon Placement Revisited: (Degree of Freedom + 1)-{SUM} Hardness
                  and an Improvement via Offline Dynamic Rectangle Union},
  booktitle    = {Proceedings of the 2022 {ACM-SIAM} Symposium on Discrete Algorithms,
                  {SODA} 2022, Virtual Conference / Alexandria, VA, USA, January 9 -
                  12, 2022},
  pages        = {3181--3201},
  publisher    = {{SIAM}},
  year         = {2022},
  url          = {https://doi.org/10.1137/1.9781611977073.124},
  doi          = {10.1137/1.9781611977073.124},
  bibsource    = {dblp computer science bibliography, https://dblp.org}
}

@inproceedings{CyganMWW17,
  author       = {Marek Cygan and
                  Marcin Mucha and
                  Karol Wegrzycki and
                  Michal Wlodarczyk},
  editor       = {Ioannis Chatzigiannakis and
                  Piotr Indyk and
                  Fabian Kuhn and
                  Anca Muscholl},
  title        = {On Problems Equivalent to (min, +)-Convolution},
  booktitle    = {44th International Colloquium on Automata, Languages, and Programming,
                  {ICALP} 2017, July 10-14, 2017, Warsaw, Poland},
  series       = {LIPIcs},
  volume       = {80},
  pages        = {22:1--22:15},
  publisher    = {Schloss Dagstuhl - Leibniz-Zentrum f{\"{u}}r Informatik},
  year         = {2017},
  url          = {https://doi.org/10.4230/LIPIcs.ICALP.2017.22},
  doi          = {10.4230/LIPICS.ICALP.2017.22},
  bibsource    = {dblp computer science bibliography, https://dblp.org}
}

@inproceedings{KunnemannPS17,
  author       = {Marvin K{\"{u}}nnemann and
                  Ramamohan Paturi and
                  Stefan Schneider},
  editor       = {Ioannis Chatzigiannakis and
                  Piotr Indyk and
                  Fabian Kuhn and
                  Anca Muscholl},
  title        = {On the Fine-Grained Complexity of One-Dimensional Dynamic Programming},
  booktitle    = {44th International Colloquium on Automata, Languages, and Programming,
                  {ICALP} 2017, July 10-14, 2017, Warsaw, Poland},
  series       = {LIPIcs},
  volume       = {80},
  pages        = {21:1--21:15},
  publisher    = {Schloss Dagstuhl - Leibniz-Zentrum f{\"{u}}r Informatik},
  year         = {2017},
  url          = {https://doi.org/10.4230/LIPIcs.ICALP.2017.21},
  doi          = {10.4230/LIPICS.ICALP.2017.21},
  bibsource    = {dblp computer science bibliography, https://dblp.org}
}

@inproceedings{BringmannS21,
  author       = {Karl Bringmann and
                  Jasper Slusallek},
  editor       = {Nikhil Bansal and
                  Emanuela Merelli and
                  James Worrell},
  title        = {Current Algorithms for Detecting Subgraphs of Bounded Treewidth Are
                  Probably Optimal},
  booktitle    = {48th International Colloquium on Automata, Languages, and Programming,
                  {ICALP} 2021, July 12-16, 2021, Glasgow, Scotland (Virtual Conference)},
  series       = {LIPIcs},
  volume       = {198},
  pages        = {40:1--40:16},
  publisher    = {Schloss Dagstuhl - Leibniz-Zentrum f{\"{u}}r Informatik},
  year         = {2021},
  url          = {https://doi.org/10.4230/LIPIcs.ICALP.2021.40},
  doi          = {10.4230/LIPICS.ICALP.2021.40},
  bibsource    = {dblp computer science bibliography, https://dblp.org}
}

@inproceedings{FischerKRS25,
  author       = {Nick Fischer and
                  Marvin K{\"{u}}nnemann and
                  Mirza Redzic and
                  Julian Stie{\ss}},
  editor       = {Keren Censor{-}Hillel and
                  Fabrizio Grandoni and
                  Jo{\"{e}}l Ouaknine and
                  Gabriele Puppis},
  title        = {The Role of Regularity in (Hyper-)Clique Detection and Implications
                  for Optimizing Boolean {CSPs}},
  booktitle    = {52nd International Colloquium on Automata, Languages, and Programming,
                  {ICALP} 2025, July 8-11, 2025, Aarhus, Denmark},
  series       = {LIPIcs},
  volume       = {334},
  pages        = {78:1--78:18},
  publisher    = {Schloss Dagstuhl - Leibniz-Zentrum f{\"{u}}r Informatik},
  year         = {2025},
  url          = {https://doi.org/10.4230/LIPIcs.ICALP.2025.78},
  doi          = {10.4230/LIPICS.ICALP.2025.78},
  bibsource    = {dblp computer science bibliography, https://dblp.org}
}

@inproceedings{GokajK25,
  author       = {Geri Gokaj and
                  Marvin K{\"{u}}nnemann},
  editor       = {Raghu Meka},
  title        = {Completeness Theorems for {k-SUM} and Geometric Friends: Deciding Fragments
                  of Linear Integer Arithmetic},
  booktitle    = {16th Innovations in Theoretical Computer Science Conference, {ITCS}
                  2025, January 7-10, 2025, Columbia University, New York, NY, {USA}},
  series       = {LIPIcs},
  volume       = {325},
  pages        = {55:1--55:25},
  publisher    = {Schloss Dagstuhl - Leibniz-Zentrum f{\"{u}}r Informatik},
  year         = {2025},
  url          = {https://doi.org/10.4230/LIPIcs.ITCS.2025.55},
  doi          = {10.4230/LIPICS.ITCS.2025.55},
  bibsource    = {dblp computer science bibliography, https://dblp.org}
}

@inproceedings{GokajKST25,
  author       = {Geri Gokaj and
                  Marvin K{\"{u}}nnemann and
                  Sabine Storandt and
                  Carina Truschel},
  editor       = {Anne Benoit and
                  Haim Kaplan and
                  Sebastian Wild and
                  Grzegorz Herman},
  title        = {(Multivariate) k-SUM as Barrier to Succinct Computation},
  booktitle    = {33rd Annual European Symposium on Algorithms, {ESA} 2025, September
                  15-17, 2025, Warsaw, Poland},
  series       = {LIPIcs},
  volume       = {351},
  pages        = {42:1--42:19},
  publisher    = {Schloss Dagstuhl - Leibniz-Zentrum f{\"{u}}r Informatik},
  year         = {2025},
  url          = {https://doi.org/10.4230/LIPIcs.ESA.2025.42},
  doi          = {10.4230/LIPICS.ESA.2025.42},
  bibsource    = {dblp computer science bibliography, https://dblp.org}
}

@inproceedings{BringmannFK19,
  author       = {Karl Bringmann and
                  Nick Fischer and
                  Marvin K{\"{u}}nnemann},
  editor       = {Amir Shpilka},
  title        = {A Fine-Grained Analogue of {Schaefer}'s Theorem in {P:} Dichotomy of
                  {$\exists^k\forall$}-Quantified First-Order Graph Properties},
  booktitle    = {34th Computational Complexity Conference, {CCC} 2019, July 18-20,
                  2019, New Brunswick, NJ, {USA}},
  series       = {LIPIcs},
  volume       = {137},
  pages        = {31:1--31:27},
  publisher    = {Schloss Dagstuhl - Leibniz-Zentrum f{\"{u}}r Informatik},
  year         = {2019},
  url          = {https://doi.org/10.4230/LIPIcs.CCC.2019.31},
  doi          = {10.4230/LIPICS.CCC.2019.31},
  bibsource    = {dblp computer science bibliography, https://dblp.org}
}

@inproceedings{KunnemannM20,
  author       = {Marvin K{\"{u}}nnemann and
                  D{\'{a}}niel Marx},
  editor       = {Shubhangi Saraf},
  title        = {Finding Small Satisfying Assignments Faster Than Brute Force: {A}
                  Fine-Grained Perspective into Boolean Constraint Satisfaction},
  booktitle    = {35th Computational Complexity Conference, {CCC} 2020, July 28-31,
                  2020, Saarbr{\"{u}}cken, Germany (Virtual Conference)},
  series       = {LIPIcs},
  volume       = {169},
  pages        = {27:1--27:28},
  publisher    = {Schloss Dagstuhl - Leibniz-Zentrum f{\"{u}}r Informatik},
  year         = {2020},
  url          = {https://doi.org/10.4230/LIPIcs.CCC.2020.27},
  doi          = {10.4230/LIPICS.CCC.2020.27},
  bibsource    = {dblp computer science bibliography, https://dblp.org}
}

@inproceedings{CarmeliZBKS20,
  author       = {Nofar Carmeli and
                  Shai Zeevi and
                  Christoph Berkholz and
                  Benny Kimelfeld and
                  Nicole Schweikardt},
  editor       = {Dan Suciu and
                  Yufei Tao and
                  Zhewei Wei},
  title        = {Answering (Unions of) Conjunctive Queries using Random Access and
                  Random-Order Enumeration},
  booktitle    = {Proceedings of the 39th {ACM} {SIGMOD-SIGACT-SIGAI} Symposium on Principles
                  of Database Systems, {PODS} 2020, Portland, OR, USA, June 14-19, 2020},
  pages        = {393--409},
  publisher    = {{ACM}},
  year         = {2020},
  url          = {https://doi.org/10.1145/3375395.3387662},
  doi          = {10.1145/3375395.3387662},
  bibsource    = {dblp computer science bibliography, https://dblp.org}
}

@inproceedings{AnGIJKN21,
  author       = {Haozhe An and
                  Mohit Gurumukhani and
                  Russell Impagliazzo and
                  Michael Jaber and
                  Marvin K{\"{u}}nnemann and
                  Maria Paula Parga Nina},
  editor       = {Petr A. Golovach and
                  Meirav Zehavi},
  title        = {The Fine-Grained Complexity of Multi-Dimensional Ordering Properties},
  booktitle    = {16th International Symposium on Parameterized and Exact Computation,
                  {IPEC} 2021, September 8-10, 2021, Lisbon, Portugal},
  series       = {LIPIcs},
  volume       = {214},
  pages        = {3:1--3:15},
  publisher    = {Schloss Dagstuhl - Leibniz-Zentrum f{\"{u}}r Informatik},
  year         = {2021},
  url          = {https://doi.org/10.4230/LIPIcs.IPEC.2021.3},
  doi          = {10.4230/LIPICS.IPEC.2021.3},
  bibsource    = {dblp computer science bibliography, https://dblp.org}
}

@inproceedings{GorbachevK23,
  author       = {Egor Gorbachev and
                  Marvin K{\"{u}}nnemann},
  editor       = {Erin W. Chambers and
                  Joachim Gudmundsson},
  title        = {Combinatorial Designs Meet Hypercliques: Higher Lower Bounds for {Klee}'s
                  Measure Problem and Related Problems in Dimensions {$d \geq 4$}},
  booktitle    = {39th International Symposium on Computational Geometry, SoCG 2023,
                  June 12-15, 2023, Dallas, Texas, {USA}},
  series       = {LIPIcs},
  volume       = {258},
  pages        = {36:1--36:14},
  publisher    = {Schloss Dagstuhl - Leibniz-Zentrum f{\"{u}}r Informatik},
  year         = {2023},
  url          = {https://doi.org/10.4230/LIPIcs.SoCG.2023.36},
  doi          = {10.4230/LIPICS.SOCG.2023.36},
  bibsource    = {dblp computer science bibliography, https://dblp.org}
}

@article{AlonYZ95,
  author       = {Noga Alon and
                  Raphael Yuster and
                  Uri Zwick},
  title        = {Color-Coding},
  journal      = {J. {ACM}},
  volume       = {42},
  number       = {4},
  pages        = {844--856},
  year         = {1995},
  url          = {https://doi.org/10.1145/210332.210337},
  doi          = {10.1145/210332.210337},
  bibsource    = {dblp computer science bibliography, https://dblp.org}
}

@article{VassilevskaW18,
  author       = {Virginia {Vassilevska Williams} and
                  R. Ryan Williams},
  title        = {Subcubic Equivalences Between Path, Matrix, and Triangle Problems},
  journal      = {J. {ACM}},
  volume       = {65},
  number       = {5},
  pages        = {27:1--27:38},
  year         = {2018},
  url          = {https://doi.org/10.1145/3186893},
  doi          = {10.1145/3186893},
  bibsource    = {dblp computer science bibliography, https://dblp.org}
}

@article{SchmidtS90,
  author       = {Jeanette P. Schmidt and
                  Alan Siegel},
  title        = {The Spatial Complexity of Oblivious k-Probe Hash Functions},
  journal      = {{SIAM} J. Comput.},
  volume       = {19},
  number       = {5},
  pages        = {775--786},
  year         = {1990},
  url          = {https://doi.org/10.1137/0219054},
  doi          = {10.1137/0219054},
  bibsource    = {dblp computer science bibliography, https://dblp.org}
}

@article{Raz10,
  author       = {Ran Raz},
  title        = {Elusive Functions and Lower Bounds for Arithmetic Circuits},
  journal      = {Theory Comput.},
  volume       = {6},
  number       = {1},
  pages        = {135--177},
  year         = {2010},
  url          = {https://doi.org/10.4086/toc.2010.v006a007},
  doi          = {10.4086/TOC.2010.V006A007},
  bibsource    = {dblp computer science bibliography, https://dblp.org}
}

@article{BansalW12,
  author       = {Nikhil Bansal and
                  Ryan Williams},
  title        = {Regularity Lemmas and Combinatorial Algorithms},
  journal      = {Theory Comput.},
  volume       = {8},
  number       = {1},
  pages        = {69--94},
  year         = {2012},
  url          = {https://doi.org/10.4086/toc.2012.v008a004},
  doi          = {10.4086/TOC.2012.V008A004},
  bibsource    = {dblp computer science bibliography, https://dblp.org}
}

@article{Blaser13,
  author       = {Markus Bl{\"{a}}ser},
  title        = {Fast Matrix Multiplication},
  journal      = {Theory Comput.},
  volume       = {5},
  pages        = {1--60},
  year         = {2013},
  url          = {https://doi.org/10.4086/toc.gs.2013.005},
  doi          = {10.4086/TOC.GS.2013.005},
  bibsource    = {dblp computer science bibliography, https://dblp.org}
}

@article{AustrinKK22,
  author       = {Per Austrin and
                  Petteri Kaski and
                  Kaie Kubjas},
  title        = {Tensor Network Complexity of Multilinear Maps},
  journal      = {Theory Comput.},
  volume       = {18},
  pages        = {1--54},
  year         = {2022},
  url          = {https://doi.org/10.4086/toc.2022.v018a016},
  doi          = {10.4086/TOC.2022.V018A016},
  bibsource    = {dblp computer science bibliography, https://dblp.org}
}

@article{Yu18,
  author       = {Huacheng Yu},
  title        = {An improved combinatorial algorithm for Boolean matrix multiplication},
  journal      = {Inf. Comput.},
  volume       = {261},
  pages        = {240--247},
  year         = {2018},
  url          = {https://doi.org/10.1016/j.ic.2018.02.006},
  doi          = {10.1016/j.ic.2018.02.006},
  bibsource    = {dblp computer science bibliography, https://dblp.org}
}

@article{GajentaanO95,
  author       = {Anka Gajentaan and
                  Mark H. Overmars},
  title        = {On a Class of {$O(n^2)$} Problems in Computational Geometry},
  journal      = {Comput. Geom.},
  volume       = {5},
  pages        = {165--185},
  year         = {1995},
  url          = {https://doi.org/10.1016/0925-7721(95)00022-2},
  doi          = {10.1016/0925-7721(95)00022-2},
  bibsource    = {dblp computer science bibliography, https://dblp.org}
}

@inproceedings{ArlazarovDKF70,
  title        = {On economical construction of the transitive closure of an oriented graph},
  author       = {Vladimir L. Arlazarov and
                  Yefim A. Dinic and
                  Aleksandr Kronrod and
                  Igor Aleksandrovich Farad\v{z}ev},
  booktitle    = {Doklady Akademii Nauk},
  year         = {1970},
  pages        = {487--488},
  volume       = {194},
  organization = {Russian Academy of Sciences}
}

@article{Gowers01,
  author       = {W. Timothy Gowers},
  year         = {2001},
  month        = {08},
  pages        = {465--588},
  title        = {A new proof of {Szemerédi}'s theorem},
  volume       = {11},
  journal      = {GAFA Geometric And Functional Analysis},
  doi          = {10.1007/s00039-001-0332-9}
}

@article{Gowers07,
  author       = {W. Timothy Gowers},
  title        = {Hypergraph regularity and the multidimensional Szemerédi theorem},
  journal      = {Annals of Mathematics},
  volume       = {166},
  number       = {3},
  pages        = {897--946},
  year         = {2007},
  url          = {http://dx.doi.org/10.4007/annals.2007.166.897},
  doi          = {10.4007/annals.2007.166.897},
}

@article{Nagle10,
  author       = {Brendan Nagle},
  title        = {On Computing the Frequencies of Induced Subhypergraphs},
  journal      = {{SIAM} J. Discret. Math.},
  volume       = {24},
  number       = {1},
  pages        = {322--329},
  year         = {2010},
  url          = {https://doi.org/10.1137/090752961},
  doi          = {10.1137/090752961},
  bibsource    = {dblp computer science bibliography, https://dblp.org}
}

@article{ShpilkaY10,
  author       = {Amir Shpilka and
                  Amir Yehudayoff},
  title        = {Arithmetic Circuits: {A} survey of recent results and open questions},
  journal      = {Found. Trends Theor. Comput. Sci.},
  volume       = {5},
  number       = {3-4},
  pages        = {207--388},
  year         = {2010},
  url          = {https://doi.org/10.1561/0400000039},
  doi          = {10.1561/0400000039},
  bibsource    = {dblp computer science bibliography, https://dblp.org}
}

@article{GreenT17,
  author       = {Ben Green and
                  Terence Tao},
  title        = {New bounds for {Szemerédi}’s theorem, {III}: {A} polylogarithmic bound for {$r_4(N)$}},
  journal      = {Mathematika},
  volume       = {63},
  number       = {3},
  pages        = {944--1040},
  year         = {2017},
  url          = {http://dx.doi.org/10.1112/S0025579317000316},
  doi          = {10.1112/s0025579317000316},
}

@article{AlonGM97,
  author       = {Noga Alon and
                  Zvi Galil and
                  Oded Margalit},
  title        = {On the Exponent of the All Pairs Shortest Path Problem},
  journal      = {J. Comput. Syst. Sci.},
  volume       = {54},
  number       = {2},
  pages        = {255--262},
  year         = {1997},
  url          = {https://doi.org/10.1006/jcss.1997.1388},
  doi          = {10.1006/JCSS.1997.1388},
  bibsource    = {dblp computer science bibliography, https://dblp.org}
}

@inproceedings{KopelowitzW20,
  author       = {Tsvi Kopelowitz and
                  Virginia {Vassilevska Williams}},
  editor       = {Artur Czumaj and
                  Anuj Dawar and
                  Emanuela Merelli},
  title        = {Towards Optimal Set-Disjointness and Set-Intersection Data Structures},
  booktitle    = {47th International Colloquium on Automata, Languages, and Programming,
                  {ICALP} 2020, July 8-11, 2020, Saarbr{\"{u}}cken, Germany (Virtual
                  Conference)},
  series       = {LIPIcs},
  volume       = {168},
  pages        = {74:1--74:16},
  publisher    = {Schloss Dagstuhl - Leibniz-Zentrum f{\"{u}}r Informatik},
  year         = {2020},
  url          = {https://doi.org/10.4230/LIPIcs.ICALP.2020.74},
  doi          = {10.4230/LIPICS.ICALP.2020.74},
  bibsource    = {dblp computer science bibliography, https://dblp.org}
}

@article{ChatterjeeDHS21,
  author       = {Krishnendu Chatterjee and
                  Wolfgang Dvor{\'{a}}k and
                  Monika Henzinger and
                  Alexander Svozil},
  title        = {Algorithms and conditional lower bounds for planning problems},
  journal      = {Artif. Intell.},
  volume       = {297},
  pages        = {103499},
  year         = {2021},
  url          = {https://doi.org/10.1016/j.artint.2021.103499},
  doi          = {10.1016/J.ARTINT.2021.103499},
  bibsource    = {dblp computer science bibliography, https://dblp.org}
}

@article{ItaiR78,
  author       = {Alon Itai and
                  Michael Rodeh},
  title        = {Finding a Minimum Circuit in a Graph},
  journal      = {{SIAM} J. Comput.},
  volume       = {7},
  number       = {4},
  pages        = {413--423},
  year         = {1978},
  url          = {https://doi.org/10.1137/0207033},
  doi          = {10.1137/0207033},
  bibsource    = {dblp computer science bibliography, https://dblp.org}
}

@book{TaoV06,
  title     = {Additive Combinatorics},
  author    = {Terence Tao and Van H. Vu},
  series    = {Cambridge Studies in Advanced Mathematics},
  publisher = {Cambridge University Press},
  year      = {2006},
  doi       = {10.1017/CBO9780511755149},
}

@article{BloomS23,
  author       = {Thomas F. Bloom and
                  Olof Sisask},
  title        = {An improvement to the Kelley-Meka bounds on three-term arithmetic progressions},
  journal      = {arXiv},
  year         = {2023},
  url          = {https://arxiv.org/abs/2309.02353},
}

@article{LengSS24,
  author       = {James Leng and
                  Ashwin Sah and
                  Mehtaab Sawhney},
  title        = {Improved Bounds for Szemerédi's Theorem},
  journal      = {arXiv},
  year         = {2024},
  url          = {https://arxiv.org/abs/2402.17995},
}

@inproceedings{FuLR25,
  author       = {Cheng{-}Hao Fu and
                  Andrea Lincoln and
                  Rene Reyes},
  editor       = {Keren Censor{-}Hillel and
                  Fabrizio Grandoni and
                  Jo{\"{e}}l Ouaknine and
                  Gabriele Puppis},
  title        = {Worst-Case and Average-Case Hardness of Hypercycle and Database Problems},
  booktitle    = {52nd International Colloquium on Automata, Languages, and Programming,
                  {ICALP} 2025, July 8-11, 2025, Aarhus, Denmark},
  series       = {LIPIcs},
  volume       = {334},
  pages        = {81:1--81:20},
  publisher    = {Schloss Dagstuhl - Leibniz-Zentrum f{\"{u}}r Informatik},
  year         = {2025},
  url          = {https://doi.org/10.4230/LIPIcs.ICALP.2025.81},
  doi          = {10.4230/LIPICS.ICALP.2025.81},
  bibsource    = {dblp computer science bibliography, https://dblp.org}
}
\appendix
\section{Hardness Assumptions}
\label{sec:assumptions}

In this section, we give an overview of the hardness hypotheses we use throughout the paper. We work with the standard Word-RAM model with $\bigO(\log n)$-bit words.

Given a graph $G$ with $n$ vertices, the \emph{Triangle Detection} problem is to decide if $G$ contains a triple of pairwise adjacent vertices. A simple reduction to matrix multiplication allows us to solve Triangle Detection in graphs with $n$ vertices in time $\bigO(n^{\omega})$ \cite{ItaiR78}.
Despite extensive effort, so far no algorithm is known that solves Triangle Detection in time $\bigO(n^{\omega-\epsilon})$ for any $\epsilon>0$. This prompted the following hypothesis which serves as a tool to prove the hardness of many related problems, e.g., \cite{AbboudW14,KopelowitzW20,ChatterjeeDHS21,VassilevskaW18} and constitutes a prominent hypothesis in fine-grained complexity theory.
\begin{hypothesis}[Triangle Detection Hypothesis]\label{hypot:triangle}
    Triangle Detection cannot be solved in $\bigO(n^{\omega-\varepsilon})$ for any $\epsilon>0$.
\end{hypothesis}

Another central fine-grained complexity problem is the zero-triangle problem, which essentially asks if a given weighted graph contains a triangle whose edge weights sum up to zero.
\begin{problem}[Zero Triangle]
    Given a tripartite weighted graph $G=(X\cup Y\cup Z, E, w)$ with $n$ vertices in each part, and weight function $w:E\to [-M, M]$, the \emph{Zero Triangle Detection} problem is to decide if $G$ contains a triangle of weight equal to $0$, i.e., do there exist $x \in X, y \in Y, z \in Z$ such that $w(x,y)+w(y,z)+w(z,x)=0$?
\end{problem}

Zero Triangle takes a special position in the web of fine-grained reductions as a cubic-time problem known to be harder than both the seminal 3SUM and APSP problems~\cite{VassilevskaW18}. For graphs with small weights bounded by $M$, a faster $\tilde O(M n^\omega)$-time algorithm is known~\cite{AlonGM97}. Note that for weights of size $M\ge n^{3-\omega-o(1)}$, this algorithm takes cubic time as well. For this reason, Abboud, Bringmann, Khoury and Zamir introduced the Strong Zero Triangle Conjecture~\cite[Conjecture 7.5]{AbboudBKZ22}, essentially postulating that Zero Triangle requires time $\min\{Mn^2,n^3\}^{1- o(1)}$ for essentially any choice of $M$, matching the upper bound if $\omega=2$. For our purposes, we will use the following formulation:\footnote{Abboud et al.~\cite{AbboudBKZ22} give a more general formulation, which implies the special case we use.}

\begin{hypothesis}[Strong Zero Triangle Conjecture \cite{AbboudBKZ22}]
There is no algorithm to detect a zero-weight triangle in a graph with edge weights $\set{-n, \dots, n}$ that runs in time $O(n^{3-\epsilon})$ (for any constant $\epsilon > 0$).
\end{hypothesis}

As described further in detail in Section~\ref{sec:ap}, we consider the following hypothesis for $k$-arithmetic progressions ($k$-AP for short).
\hypkap*

\section{Alternative Matrix-Based Algorithm for Distributivity Checking} \label{sec:freivalds_distributivity}
In order to bring down the running time for associativity checking from cubic to quadratic, Rajagopalan and Schulman \cite{rajagopalan00} come up with a clever way to check this property of many triples in parallel.
On a high level their algorithm works as follows.
Given an algebraic structure $(S,\oplus)$, one can treat the elements of $S$ as vectors (e.g., each element can be identified with a column of the identity matrix $I_n$, where $n=|S|$). 
They then sample random (scalar) weights $w_x$ for each $x\in S$, from some large finite field $\mathbb F_p$.
Finally, they show that instead of checking for each triple of elements $a,b,c$ if $(a\oplus b)\oplus c$ holds separately, one can instead check the identity
\begin{equation*}
    \left(\left(\sum_{a\in S} w_a\cdot a\right) \oplus \left(\sum_{b\in S} w_b\cdot b\right)\right)\oplus\left(\sum_{c\in S} w_c\cdot c\right)
=
\left(\sum_{a\in S} w_a\cdot a\right) \oplus \left(\left(\sum_{b\in S} w_b\cdot b\right)\oplus\left(\sum_{c\in S} w_c\cdot c\right)\right)\end{equation*}
where the operations between scalars and vectors are defined in a standard way (i.e. $(w_aa) \oplus (w_bb) = w_aw_b(a\oplus b)$).
They then show that if this previous equality is satisfied, then with high probability the structure $(S,\oplus)$ is associative.
In this section, we show how this approach can be extended in a non-trivial way to also allow for checking distributivity by a careful combination of Freivalds' Lemma and matrix multiplication, thus providing an alternative algorithm for distributivity checking that is more in line with previous work on the topic.

\begin{theorem}
    Given a set $S$ of size $n$ together with two binary operations $\oplus, \odot: S\times S \to S$, there exists a randomized algorithm checking if the algebraic structure $(S,\oplus, \odot)$ is distributive in time $\bigO(n^{\omega})$.
\end{theorem}
Towards proving this theorem, we first provide the algorithmic results, and then dedicate the next subsection to the hardness result.
Our algorithm consists of four different parts:
\begin{enumerate}
    \item Reducing to \emph{Matrix Product Equality Verification}.
    \item Adding random weights to reduce the number of Matrix Product Equality Verification instances.
    \item Applying Freivalds' algorithm.
    \item Using matrix multiplication to bring the running time down from $\bigO(n^3)$ to $\bigO(n^{\omega})$.
\end{enumerate}
\paragraph{(1) Reduction to Matrix Product Equality Verification} Roughly speaking, in this part we want to represent our operations $\oplus,\odot$ on the set $S$ by matrices, which will allow us to construct a matrix equation that corresponds to checking distributivity on a subset of $S$.
To this end, we first label each element in $S$ uniquely by a column vector from $\{0,1\}^n$ and then for each $a\in S$ we can construct a linear map (represented by a matrix) that maps each vector corresponding to $x\in S$ to the vector corresponding to $a\oplus x$ (resp. $a\odot x$).
More formally, we first order the elements in $S$ arbitrarily and let $\lambda: S\to \{0,1\}^n$ be the labeling function that maps $i$-th element from $S$ to the column vector whose $i$-th entry is equal to $1$ and all other entries are zeros.
For each $a\in S$ define the maps $s_a, p_a:\lambda(S)\to \lambda(S)$ as follows.
\begin{align*}
    s_a:\lambda(x)\mapsto \lambda(a\oplus x)\\
    p_a: \lambda(x)\mapsto \lambda(a\odot x).
\end{align*}
Since $\lambda$ is an injective function, we get the following observation.
\begin{observation}\label{obs:dist-as-functions}
    Let $a,b\in S$ be fixed. Then for any $c\in S$, the identity $a\odot(b\oplus c) = (a\odot b) \oplus (a\odot c)$ holds if and only if $(p_a\circ s_b)(c) = (s_{a\odot b}\circ p_a)(c)$ holds.
\end{observation}
This observation gives us a nice way to think about distributivity checking in terms of functions on real vectors, rather than arbitrary algebraic structures.
More precisely, in order to show that for all $a,b,c\in S$ the identity $a\odot(b\oplus c) = (a\odot b) \oplus (a\odot c)$ is satisfied, it suffices to show that for all $a,b\in S$ the maps $(p_a\circ s_b)$ and $(s_{a\odot b}\circ p_a)$ are the same.
This is particularly useful, since we can now prove that for each $a\in S$ there exists a matrix $\sigma_a$ (resp. $\pi_a$), such that for any $x\in S$ $\sigma_a\cdot \lambda(x) = s_a\left(\lambda(x)\right)$.
Hence, we only need to prove that for each $a,b$ the matrices $(\pi_a\cdot\sigma_b)$ and $(\sigma_{(a\odot b)}\cdot \pi_a)$ are equal.
Formally, we prove the following lemma.
\begin{lemma}
    Let $a\in S$ be arbitrary.
    Let $\sigma_a$ be a matrix whose $i$-th column is equal to $\lambda(a\oplus x)$, where $x$ is the $i$-th element in the ordering of $S$.
    Define $\pi_a$ analogously. Then
    \begin{enumerate}
        \item For any $x\in S$ $\sigma_a\cdot \lambda(x) = s_a\left(\lambda(x)\right)$.
        \item For any $x\in S$ $\pi_a\cdot \lambda(x) = p_a\left(\lambda(x)\right)$.
    \end{enumerate}
\end{lemma}
\begin{proof}
    Let $x$ be the $i$-th element in ordering of $S$.
    Then $\lambda(x)$ has all zeros except in row $i$ where it has a $1$.
    Hence, $\sigma_a\cdot \lambda(x)$ (resp. $\pi_a\cdot \lambda(x)$) is precisely equal to the $i$-th column of $\sigma_a$ (resp. $\pi_a$), which is by definition equal to $\lambda(a\oplus x) = s_a\left(\lambda(x)\right)$ (resp. $\lambda(a\odot x) = p_a\left(\lambda(x)\right)$).
\end{proof}
The following corollary now follows from the previous lemma together with the injectivity of $\lambda$.
\begin{corollary}
    Let $a,b\in S$ be arbitrary. Then the identity $a\odot(b\oplus c) = (a\odot b) \oplus(a\odot c)$ is satisfied for all $c\in S$ if and only if the equality $\pi_a \cdot \sigma_b = \sigma_{(a\odot b)}\cdot \pi_a$ is satisfied.
\end{corollary}
This gives us a rather straightforward reduction to Matrix Product Equality Verification: for every pair $a,b$ construct the matrices $\pi_a, \sigma_b, \sigma_{(a\odot b)}$ and test the equality of the matrix product $\pi_a \cdot \sigma_b = \sigma_{(a\odot b)}\cdot \pi_a$.
Unfortunately however, this yields too many instances of the Matrix Product Equality Verification problem.
The next part of the algorithm is dedicated to resolving this problem by introducing a random weight for each element of $S$ and then instead of checking each matrix product we sum up the weighted matrices and show that (with high probability) the answer to the scaled version of Matrix Product Equality Verification is equal to the answer of the original problem.
This part of the algorithm is closely related to the algorithm for associativity testing by Rajagopalan and Schulman \cite{rajagopalan00}.
\paragraph{(2) Adding Random Weights} Let $p$ be a fixed prime number.
For the rest of this section, we perform all additions and multiplications over the field $\mathbb F_p$.
For each $x\in S$, chose independently uniformly at random an element $w_x\in \mathbb F_p$ (\emph{weight} of $x$).
Now instead of having to verify the matrix product equality $\pi_a\cdot \sigma_b = \sigma_{(a\odot b)}\cdot \pi_a$ for all pairs $a,b\in S$, we can instead sum all of these equations (or a fraction of them) together and just check whether the sums are equal, thus significantly reducing the number of Matrix Product Equalities we have to verify (even bringing it down to one if we sum over all equalities).
This of course preserves the yes-instances (no false negatives), however we cannot guarantee that no false positives are generated, and in fact many times we cannot avoid the false positives.
We nevertheless show that by choosing the weights independently at random keeps the probability of seeing the false positives low.
\begin{lemma}
    Let $a\in S$ be fixed. Assume that there exists $b'\in S$ such that $\pi_a\cdot \sigma_{b'} \neq \sigma_{(a\odot b')}\cdot \pi_a$.
    Then
    \[\mathbb P \left[\pi_a\cdot \sum_{b\in S}w_b\sigma_b = \left(\sum_{b\in S}w_b\sigma_{(a\odot b)} \right)\pi_a\right] \le \frac{1}{p}\]
\end{lemma}
\begin{proof}
    Let $b'\in S$ be the element for which $\pi_a\cdot \sigma_{b'} \neq \sigma_{(a\odot b')}\cdot \pi_a$.
    Then there exists $i,j\in [n]$, such that
    \[
    \left(\pi_a\cdot \sigma_{b'} - \sigma_{(a\odot b')}\cdot \pi_a\right)[i,j]\neq 0.
    \]
    Moreover, the following inequality is trivial
    \begin{equation*}\label{eq:bounding-prob-by-fixed-entry}\mathbb P \left[\pi_a\cdot \sum_{b\in S}w_b\sigma_b = \left(\sum_{b\in S}w_b\sigma_{(a\odot b)} \right)\pi_a\right] \le \mathbb P \left[\left(\pi_a\cdot \sum_{b\in S}w_b\sigma_b - \left(\sum_{b\in S}w_b\sigma_{(a\odot b)} \right)\pi_a\right)\left[i,j\right] = 0\right].\end{equation*}
    So it suffices to bound the latter probability by $\frac{1}{p}$. By pulling out the term involving $b'$ out of the sum, we can write:
    \[
    \pi_a\cdot \sum_{b\in S}w_b\sigma_b - \left(\sum_{b\in S}w_b\sigma_{(a\odot b)} \right)\pi_a = \left(\pi_a\cdot w_{b'}\sigma_{b'} - w_{b'}\sigma_{(a\odot b')}\pi_a \right)+ \underbrace{\left(\pi_a\cdot \sum_{b\in S\setminus \{b'\}}w_b\sigma_b - \left(\sum_{b\in S\setminus\{b'\}}w_b\sigma_{(a\odot b)} \right)\pi_a\right)}_{R}.
    \]
    In particular, we have
    \begin{align*}
        \left(\pi_a\cdot \sum_{b\in S}w_b\sigma_b - \left(\sum_{b\in S}w_b\sigma_{(a\odot b)} \right)\pi_a\right)\left[i,j\right] &= w_{b'}\left(\pi_a\cdot \sigma_{b'} - \sigma_{(a\odot b')}\pi_a \right)[i,j]+ R[i,j].
    \end{align*}
    Finally, since both $R[i,j]$ and $w_{b'}\left(\pi_a\cdot \sigma_{b'} - \sigma_{(a\odot b')}\pi_a \right)[i,j]$ are elements in $\mathbb F_p$, their sum is equal to zero if and only if $w_{b'} = (-R)[i,j]\cdot\left( \left(\pi_a\cdot \sigma_{b'} - \sigma_{(a\odot b')}\pi_a \right)[i,j]\right)^{-1}$.
    This element is clearly well-defined in $\mathbb F_p$, by the choice of $b'$.
    Furthermore, since $w_{b'}$ is chosen uniformly at random, the probability of hitting this element is at most $\frac{1}{p}$.
\end{proof}
The previous lemma gives us an elegant way to reduce the number of triples we need to check: for each fixed $a\in S$, we introduce random weights for each $b\in S$, and check the equality of the matrices $\pi_a\cdot \sum_{b\in S}w_b\sigma_b$, and $ \left(\sum_{b\in S}w_b\sigma_{(a\odot b)} \right)\pi_a$.
If they are not equal, we can stop and conclude that the input is a no-instance.
If, on the other hand, the matrices are equal, we repeat with fresh randomness.
We stop after $100\lceil\log_p(n)\rceil$ repetitions and conclude that with probability at least $1-n^{-100}$ for every pair $b,c\in S$ the identity $a\odot(b\oplus c) = (a\odot b) \oplus (a\odot c)$ is satisfied.
\paragraph{(3) Magic of Freivalds'} Although in the second part of the algorithm we significantly reduce the number of matrix multiplications we need to perform, it still is not sufficient to reduce the running time sufficiently to beat the simple brute-force algorithm, since we need to perform $n$ matrix multiplications (one for every fixed $a$), each taking $\bigO(n^\omega)$ time.
However, we do not need to compute the matrix products at all, since the only thing we care about is whether the two matrix products are equal.
A standard approach to verifying the equality of matrix products is to use Freivalds' algorithm \cite{Freivalds1977}.
This approach turns out to be extremely helpful in our case, since using Freivalds' not only allows us to circumvent computing the matrix products, but it additionally gives us a way to modularize the algebraic expressions we are verifying in such a way that we only need to compute the expensive computations few times and assure that the operations that need to be computed $\mathcal O(n)$ times are cheap. 
We first state the Freivalds' Lemma applied to our setting, and then show how it applies to our problem.
\begin{lemma}[Freivalds' Lemma; \cite{Freivalds1977}]
    Let $\ell= c \log(n)$ for a sufficiently large constant $c$, and let $u_1,\dots, u_\ell \in \mathbb F_p^n$ be the vectors chosen independently uniformly at random. Let $a\in S$ be fixed. Then with probability at least $1-n^{-100}$, the following two conditions are equivalent.
    \begin{enumerate}
        \item The equality $\pi_a\cdot \sum_{b\in S}w_b\sigma_b = \left(\sum_{b\in S}w_b\sigma_{(a\odot b)} \right)\pi_a$ holds.
        \item \label{freivalds-item:2} For each $i\in [\ell]$ the equality $u_i \cdot \left(\pi_a\cdot \sum_{b\in S}w_b\sigma_b\right) = u_i\left(\left(\sum_{b\in S}w_b\sigma_{(a\odot b)} \right)\pi_a\right)$ is satisfied.
    \end{enumerate}
\end{lemma}
By using associativity of $\mathbb F_p$, we can now rewrite the expression from \ref{freivalds-item:2} as 
\[
   \left( u_i \cdot \pi_a\right)\cdot \sum_{b\in S}w_b\sigma_b = \left(u_i\left(\sum_{b\in S}w_b\sigma_{(a\odot b)} \right)\right)\cdot \pi_a.
\]
We can now prove that writing the expression in this form allows us to compute the left hand side of the expression for all pairs $a,b\in S$ efficiently. 
\begin{lemma}
    The vector $\left( u_i \cdot \pi_a\right)\cdot \sum_{b\in S}w_b\sigma_b$ can be computed for every $a \in S$ in time $\bigO(n^\omega)$.
\end{lemma}
\begin{proof}
    Recall that the matrix $\sigma_b$ has precisely $n$ ones by construction. In particular, this means that we can compute the matrix $Y := \sum_{b\in S}w_b\sigma_b$ in time $\bigO(n^2)$.
    Moreover, notice that we only need to compute this matrix once, since it is independent of $a$.
    Similarly, since matrix $\pi_a$ contains only $n$ non-zeros, we can for each $a$ compute the product $u_i\pi_a$ in time $\bigO(n)$.
    Hence, we can for each $a\in S$ compute the vector $x_a := u_i\pi_a$ in time $\bigO(n)$.
    Let $X$ be an $n\times n$ matrix whose $a$'th row is the vector $x_a$.
    As argued above, we can construct $X$ in time $\bigO(n^2)$.
    Finally, in time $\bigO(n^\omega)$ we can compute the matrix $(XY)$ whose row $a$ is precisely equal to the product $\left( u_i \cdot \pi_a\right)\cdot \sum_{b\in S}w_b\sigma_b$, as desired.
\end{proof}
\paragraph{(4) Power of Matrix Multiplication}
While a straightforward matrix multiplication approach allows us to compute the left-hand side of the expression efficiently for all pairs $a,b$ in parallel, such a simple construction is far from obvious for the right hand side. 
Nevertheless, by a careful, and a nontrivial construction, we are able to build matrices whose product yields for all pairs $a,b$ a part of the expression from the right-hand side. 
We are then able to prove that the part that we can parallelize is the \emph{expensive} part of the expression and that we can compute the rest of the expression efficiently.
In particular we first prove the following lemma.
\begin{lemma}
    There exist $n\times n$ matrices $X,Y$, that can be constructed in $\bigO(n^2)$ time, such that the row $a$ of the product matrix $(XY)$ is precisely equal to $\sum_{b\in S} w_b u_i \sigma_{a\odot b}$
\end{lemma}
\begin{proof}
    For any $a,k\in S$, let $P_a(k):=\{b\in S \mid a\odot b = k\}$. We now construct the matrix $X$ as follows.
    \[
    X[a,k]:= \sum_{b\in P_a(k)} w_b.
    \]
    We further define the matrix $Y$ as the matrix whose row $k$ is equal to the vector $u_i\sigma_k$ for each $k\in S$.
    Clearly, both matrices can be constructed in quadratic time.
    We now show that the matrix $(XY)$ satisfies the desired property.
    \begin{claim}
        Let $X,Y$ be matrices constructed as above. Then the row $a$ in the product matrix is equal to the vector $\sum_{b\in S} w_b u_i\sigma_{a\odot b}$.
    \end{claim}
    \begin{subproof}
        We can express the entry $(XY)[a,j]$ as follows.
        \begin{align*}
            (XY)[a,j] & = \sum_{k\in S} X[a,k]\cdot Y[k,j] \\
            & = \sum_{k\in S} \left(\left(\sum_{b\in P_a(k)} w_b \right)\cdot \left(  (u_i\sigma_k)[j] \right)\right)  \\
            & = \sum_{k\in S} \left(\left(\sum_{b\in P_a(k)} w_b \right)\cdot \left(  u_i(\sigma_k[\cdot, j]) \right)\right)& \text{(vector-matrix multiplication)} \\
            & = \sum_{k\in S} \sum_{b\in P_a(k)} (w_b \cdot u_i)(\sigma_k[\cdot, j])
            & \text{(rearranging the parentheses)}
        \end{align*}
        We now observe that the set $\{P_a(k)\mid k\in S\}$ partitions $S$, and hence we can rewrite the last expression as follows.
        \begin{align*}
            (XY)[a,j] & = \sum_{b\in S} (w_b \cdot u_i)(\sigma_{a\odot b}[\cdot, j])
            \\& = \sum_{b\in S} w_b (u_i \cdot \sigma_{a\odot b})[j].
        \end{align*}
        The desired claim now follows directly.
    \end{subproof}
\end{proof}
We can now prove that we can also efficiently compute the vector $\sum_{b\in S} w_b u_i \sigma_{a\odot b} \pi_a$ efficiently for all pairs $a,b\in S$.
\begin{lemma}
    The vector $u_i\sum_{b\in S} w_b\sigma_{a\odot b} \pi_a$ can be computed for every $a$ in time $\bigO(n^\omega)$.
\end{lemma}
\begin{proof}
    We can rewrite this expression as $\left(\sum_{b\in S} w_b u_i \sigma_{a\odot b}\right) \pi_a$.
    By using the previous lemma, we can construct in time $\bigO(n^\omega)$ the matrix $Z$ such that the row $a$ of $Z$ is equal to $\sum_{b\in S} w_b u_i \sigma_{a\odot b}$.
    We can now in constant time access the vector $z_a := \sum_{b\in S} w_b u_i \sigma_{a\odot b}$ for any $a\in S$.
    Recall that by construction $\pi_a$ contains only $\bigO(n)$ many non-zero entries, hence the product $z_a\pi_a$ can be computed in $\bigO(n)$ time for any fixed $a$.
    Iterating over all $a\in S$ yields the desired.
\end{proof}

\end{document}